\documentclass{fundam}
\usepackage{amsmath}
\usepackage[T1]{fontenc}
\usepackage{graphicx}
\usepackage{tabularray}
\usepackage{amssymb}
\usepackage{url} 
\usepackage{multicol}
\usepackage{enumitem}
\usepackage{tasks}
\usepackage{marginnote}
\usepackage{comment}

\usepackage{pgf,tikz}

\usetikzlibrary{calc,positioning,shapes,arrows.meta,fit}
\usetikzlibrary{decorations.markings}

\tikzset{%
 shaded/.style={draw, shape=circle, fill=black!35, inner sep=1.4pt},
 unshaded/.style={draw, shape=circle, fill=white, inner sep=1.4pt},
 quasi/.style={draw, shape=rectangle, rounded corners=3pt, fill=white, inner sep=2.5pt, minimum height=14.5pt},
 blob/.style={draw, shape=rectangle, rounded corners=12pt, thin, densely dotted},
 arrow/.style={->, thin, >=latex, shorten >=2.5pt, shorten <=2.5pt},
 order/.style={thin},
 curvy/.style={thin, looseness=1.2, bend angle=70},
 fatcurvy/.style={thin, looseness=1.7, bend angle=75},
 label/.style={shape=rectangle, inner sep=6pt},
 auto}
\newcommand{\A}{\mathbf{A}}
\newcommand{\B}{\mathbf{B}}
\newcommand{\C}{\mathbf{C}}
\newcommand{\D}{\mathbf{D}}

\newcommand{\JA}{A_+}
\newcommand{\JB}{B_+}
\newcommand{\cmax}{\mathbin{\circ^{\max}}}

\title{Wreath products of distributive quasi relation algebra frames and their application to representability}
\runninghead{W. Morton, C. Robinson}{Wreath products of DqRA frames and representability}

\author{Wilmari Morton  \\
Mathematics Education, School of Education\\
 University of the Witwatersrand, South Africa \\ wilmari.morton@wits.ac.za \\
\and  Claudette Robinson \\
Department of Mathematics and Applied Mathematics\\
University of Johannesburg, South Africa\\
claudetter@uj.ac.za\thanks{Funding from the National Research Foundation grant KIC250820357419 is gratefully acknowledged.}
}

\begin{document}

\maketitle

\begin{abstract}
We adapt a construction by Comer (1983) for polygroups to define
the wreath product 
of the dual frames of two distributive quasi relation
algebras satisfying certain conditions and show that it is itself a frame. Frames play a role analogous to that of atom structures for relation algebras. After introducing the notion of a
frame network, we establish a correspondence between frame networks satisfying certain conditions and representability of finite
distributive quasi relation algebras. Next we describe a frame network for the wreath
product of two finite representable distributive quasi relation
algebras that satisfies the aforementioned conditions, as well as additional conditions, and
use this to conclude that the complex algebra of the wreath product is a representable distributive quasi relation algebra. Thus, our construction provides a method for obtaining representable DqRAs from certain finite DqRAs that are already known to be representable.

\end{abstract}

\section{Introduction}\label{sec:introduction}

In~\cite{GJ13}, Galatos and Jipsen introduced quasi relation algebras (qRAs) as a generalisation of relation algebras~\cite{Tar41}. Unlike relation algebras, qRAs have a decidable equational theory. They can also be viewed as expansions of involutive Full Lambek (InFL-)algebras~\cite[Chapter 2]{GJKO} by an additional unary operation. A distributive quasi relation algebra (DqRA) is a qRA whose underlying lattice is distributive. The class of DqRAs includes relation algebras, which are based on Boolean algebras, as well as Sugihara monoids and MV-algebras, two classes that provide algebraic semantics for non-classical logics. Relation algebras have applications in databases ~\cite{Codd1970,VDB2001} and in temporal and spatial reasoning~\cite{Allen1983,Duntsch2005}. These applications motivate investigating whether similar methods can be developed for generalisations of relation algebras.

A central question in the study of DqRAs is whether a given abstract DqRA is representable and, if so, how to construct a representation.
Mimicking the well-studied concept of representable relation 
algebras, the notion of \emph{representable DqRAs} was introduced
in~\cite{CR-RDqRA} using posets satisfying certain symmetry requirements. 
The question of representability has a long history in the study of relation algebras: J\'{o}nsson and Tarski~\cite{JT48} asked whether every relation algebra is isomorphic to an algebra of binary relations.
This question was answered in the negative in~\cite{Lyn50}. In general, it is very hard to determine whether a given relation algebra or DqRA is representable or not. One approach to the representability problem is to construct new algebras from known representable ones and establish that the construction preserves representability. This approach has also been pursued for DqRAs using several constructions~\cite{CMR2026, CMR2026_2}, and it motivates the present study.


In this paper, we adapt the wreath product construction introduced by Comer~\cite{Com83} for polygroups to the setting of DqRAs. In Comer's construction polygroups serve a similar purpose to that of atom structures in the study of relation algebras. In the DqRA setting, the corresponding structures are DqRA frames, introduced in~\cite{CJR24, CJR26}. These are partially ordered structures equipped with additional operations and relations. They further established a duality between these frames and complete perfect DqRAs.
One of the main results of this paper shows that, under suitable conditions, the wreath product of two DqRA frames is itself a DqRA frame, thereby providing a framework for constructing larger DqRAs from smaller DqRAs.

Another approach to studying representability uses two-player representation games. Hirsch and Hodkinson developed these techniques for relation algebras across a series of papers throughout the 1990s, bringing this work together in their 2002 book~\cite{HH2002book}. Similar techniques have been applied to weakening relation algebras in~\cite{JS23}. Motivated by these approaches, we introduce the notion of a \emph{frame network}, which provides a structure suitable for use in representation games. However, we do not develop such games for DqRAs here. Instead, we use frame networks to play a role analogous to that of Comer's colour schemes in~\cite{Com83}.

The paper is organised as follows. In Section~\ref{sec:preliminaries}, we present the necessary preliminaries on DqRAs, representable DqRAs, and DqRA frames. In Section~\ref{sec:wreath_products}, we define the wreath product of the dual frames of two complete perfect DqRAs and show that, under suitable conditions, the resulting structure is again a DqRA frame. In Section~\ref{sec:networks_representations}, we introduce frame networks and establish two results connecting the existence of consistent saturated networks with representability. Proposition~\ref{prop:consistent_sat_network_imply_representable} shows that the complex algebra of a finite frame is representable if, for every frame element, there exists a consistent saturated frame network in which that element labels an edge. Conversely, Proposition~\ref{prop:SatConNetworkExists} shows that, for every join-irreducible element of a finite representable DqRA, there exists a consistent saturated network for its dual frame in which that join-irreducible labels an edge. In Section~\ref{sec:NetworkForWreathPRoduct}, we construct a frame network for the wreath product defined in Section~\ref{sec:wreath_products} and prove in Proposition~\ref{prop:existence_network_wreath_product} that it is consistent and saturated under suitable conditions. Finally, we use Proposition~\ref{prop:existence_network_wreath_product}
to conclude in Theorem~\ref{thm:ComplexOfWreathProductRepresentable} that the complex algebra of the wreath product of the
dual frames of two finite representable DqRAs satisfying certain conditions is again representable. To illustrate how the results developed in this paper can be used to find representable DqRAs, we construct representations for two $6$-element DqRAs.

\section{Preliminaries}\label{sec:preliminaries}

\subsection{Distributive quasi relation algebras}\label{subsec:DqRAs}

We begin by recalling the definitions of the algebras considered in this paper. These algebras may be regarded either as generalisations of relation algebras or as expansions of residuated lattices~\cite{GJKO}.

Recall that a \emph{residuated lattice} is an algebra $\A = \langle A, \wedge, \vee, \cdot, \backslash, /, 1\rangle$ such that $\langle A,\wedge,\vee\rangle$ is a lattice, 
$\langle A,\cdot,1\rangle$ is a monoid, and for all $a,b,c\in A$,
\[a\cdot b\leqslant c\quad \textnormal{ iff } \quad 
a\leqslant c/b\quad \textnormal{ iff } \quad
b\leqslant a\backslash c.\]

We will often write $ab$ for $a\cdot b$ whenever no confusion can arise. An element $a$ in a residuated lattice $\A$ is called \emph{idempotent} if $a^2 = a$, and a residuated lattice $\mathbf{A}$ is said to be \emph{idempotent} if all its elements are idempotent. 
We say that a residuated lattice $\mathbf A$  is 
\emph{commutative} if $ab = ba$ for all $a, b \in A$ and \emph{distributive} if its underlying lattice is distributive.

A \emph{Full Lambek algebra} or
an \emph{FL-algebra} $\A = \langle A, \wedge, \vee, \cdot, \backslash, /, 1, 0\rangle$ is a residuated lattice expanded with an additional constant $0$ (cf.~\cite[Chapter~2.2]{GJKO}). This constant $0$ can be any element of $A$ and no additional properties are assumed about it. Residuated lattices are exactly the $0$-free reducts of FL-algebras. 

The residuals and $0$ are used to define two unary operations (known as \emph{linear negations}) as follows: 
${\sim} a=a\backslash 0$ and ${-}a=0/a$.
It follows that ${\sim}1 = 1\backslash 0 = 0= 0/1=-1$.
By residuation, 
${\sim} (a\vee b)={\sim}a \wedge {\sim}b$ and ${-}(a\vee b)=
{-a}\wedge{-}b$ for all $a,b\in A$.
An FL-algebra $\mathbf{A}$ is said to be \emph{cyclic} if $-a = {\sim} a$ for all $a \in A$. It is easy to check that if an FL-algebra $\mathbf{A}$ is commutative, then it is cyclic.  

If an FL-algebra $\mathbf{A}$ satisfies the condition 
\[\mathsf{(In)}:\quad   {\sim}{-}a=a={-}{\sim}a\]
for all $a\in A$, then it is called an {\em involutive Full Lambek {\normalfont(}InFL-{\normalfont)}algebra}.
These algebras satisfy 
$a\leqslant b \textnormal{ iff } a\,\cdot({\sim} b)\leqslant -1 \textnormal{ iff } (-b)\cdot a\leqslant-1$
for all $a,b\in A$. Further, since both $-$ and $\sim$ are order-reversing, 
it follows that 
$-$ and $\sim$ are dual lattice isomorphisms in an InFL-algebra. If the underlying lattice of an InFL-algebra $\A$ is distributive, then we refer to $\A$ as a \emph{distributive InFL-algebra} (or \emph{DInFL-algebra}). 

Galatos and Jipsen~\cite[Lemma 2.2]{GJ13} established that 
any InFL-algebra is term-equivalent to an algebra $\mathbf{A}=\langle A, \wedge,\vee,\cdot, {\sim}, -, 1\rangle$ such that $\langle A, \wedge,\vee\rangle$ is a lattice, $\langle A, \cdot, 1\rangle$ is a monoid, and for all $a, b, c \in A$, we have
\begin{equation}
a\cdot b\leqslant c\quad \textnormal{ iff } \quad 
a\leqslant {-}\left(b\cdot {\sim}c\right)\quad \textnormal{ iff } \quad 
b\leqslant {\sim}\left(-c\cdot a\right).\label{eq:InFL}
\end{equation}
If needed, the residuals can be expressed in terms of $\cdot$ and the linear negations as follows: $ c/b = -\left(b\cdot {\sim}c\right)$ and 
$a\backslash c = {\sim}\left(-c\cdot a\right)$. Note also that $0 = -1 = {\sim}1$.  In the sequel we use the signature 
$\mathbf{A}=\langle A, \wedge,\vee,\cdot, {\sim}, -, 1\rangle$ to denote an InFL-algebra.

An InFL-algebra $\mathbf{A}$ can be expanded with an additional
unary operation ${\neg}:A\to A$ 
to form an {\em InFL$'$-algebra} 
$\mathbf{A}=\langle A, \wedge,\vee,\cdot, \sim, -, {\neg}, 1\rangle$
such that $\neg\neg a=a$ for all $a\in A$. 
The involution is sometimes denoted by $'$, but for the sake of readability 
we will use $\neg$.

A DmInFL$'$-algebra $\mathbf{A}$ is an InFL$'$-algebra that satisfies, for all $a,b\in A$, the {\em De Morgan law}
\[
\mathsf{(Dm)}:\quad \neg(a\vee b)=\neg a\wedge\neg b.
\]
If a DmInFL$'$-algebra $\A$ satisfies the following identity,
called the \emph{De Morgan product}, 
\[
\mathsf{(Dp)}:\quad
\neg(a\cdot b)={\sim}(-\neg b\cdot -\neg a),
\] 
for all $a, b \in A$, then it is called a \emph{quasi relation algebra} (qRA).
If $\mathbf{A}$ is a qRA, then it can be shown that $\neg 1 = - 1= {\sim}1 = 0$.
In \cite{GJ13}, the definition of a qRA included a third identity, $\mathsf{(Di)}$:  $\neg{\sim} a={-}\neg a$, 
but we note here that it is implied by $\mathsf{(Dp)}$ and (\ref{eq:InFL}) from above. 

Finally, a {\em distributive quasi relation
algebra} (or \emph{DqRA}) is a qRA  $\mathbf{A}=\langle A,\wedge,\vee,\cdot, {\sim}, -, {\neg}, 1\rangle$
such that the underlying lattice $\langle A,\wedge,\vee\rangle$
is distributive. 

Figure~\ref{fig:small_examples_DqRAs} contains some examples of DqRAs that will be considered later. A complete list of finite DInFL-algebras and DqRAs of size at most $8$ can be found in~\cite{CJR-DqRAs}. Note that the
two chains $\mathbf S_3$ and $\mathbf S_4$ are examples of finite Sugihara chains, an important class of examples of DqRAs.  
The reader is referred to~\cite{CR-Sugihara} for more details on Sugihara chains.

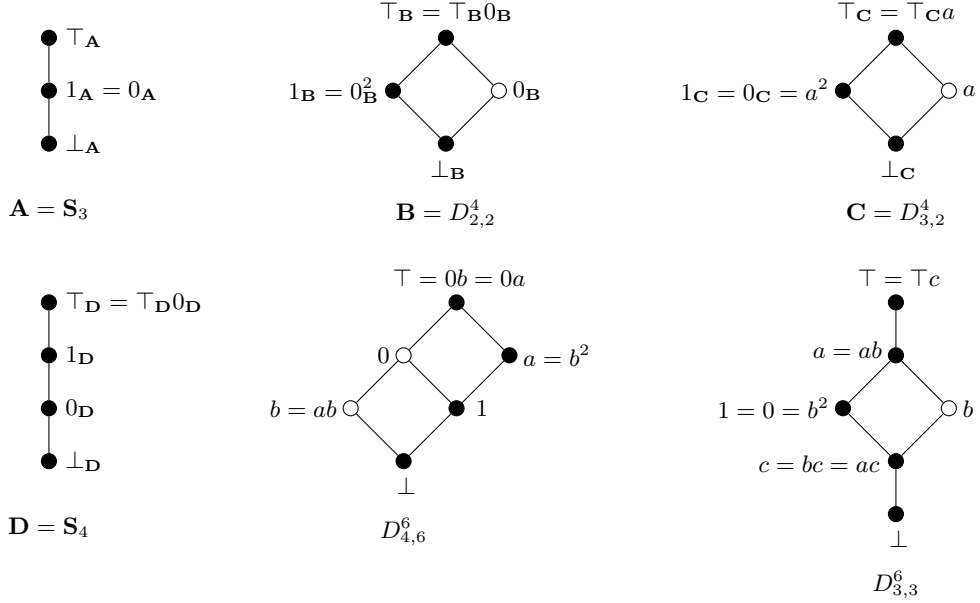
\begin{figure}[ht]
\begin{center}
\begin{tikzpicture}[scale=0.7]
\begin{scope}[xshift=-9cm, yshift=1cm]
\node[draw,circle,inner sep=2pt,fill] (a) at (0,-1) {};
\node[draw,circle,inner sep=2pt,fill] (b) at (0,0) {};
\node[draw,circle,inner sep=2pt,fill] (c) at (0,1) {};
\draw[order] (a)--(b)--(c);
\node[label,anchor=west] at (a) {\footnotesize$\bot_\A$\normalsize};
\node[label,anchor=west] at (b) {\footnotesize$1_\A=0_\A$\normalsize};
\node[label,anchor=west] at (c) {\footnotesize$\top_\A$\normalsize};
\node[label,anchor=north,yshift=-15pt] at (a) {\footnotesize$\A=\mathbf{S}_3$\normalsize};
\end{scope}

\begin{scope}[xshift=-2.5cm]
\node[draw,circle,inner sep=2pt,fill] (bot) at (1,0) {};
\node[draw,circle,inner sep=2pt,fill] (top) at (1,2) {};
\node[draw,circle,inner sep=2pt,fill] (1) at (0,1) {};
\node[draw,circle,inner sep=2pt] (0) at (2,1) {};
\draw[order] (bot)--(1);
\draw[order] (bot)--(0);
\draw[order] (1)--(top);
\draw[order] (0)--(top);
\node[label,anchor=south] at (top) {\footnotesize$\top_\B = \top_\B 0_\B$\normalsize};
\node[label,anchor=north,xshift=1pt] at (bot) {\footnotesize$\bot_\B$\normalsize};
\node[label,anchor=east,xshift=1pt] at (1) {\footnotesize$1_\B = 0_\B^2$\normalsize};
\node[label,anchor=west,xshift=-1pt] at (0) {\footnotesize$0_\B$\normalsize};
\node[label,anchor=north,yshift=-15pt] at (bot) {\footnotesize$
\B=D^4_{2,2}$\normalsize};
\end{scope} 

\begin{scope}[xshift=6cm]
\node[draw,circle,inner sep=2pt,fill] (bot) at (1,0) {};
\node[draw,circle,inner sep=2pt,fill] (top) at (1,2) {};
\node[draw,circle,inner sep=2pt,fill] (1) at (0,1) {};
\node[draw,circle,inner sep=2pt] (a) at (2,1) {};
\draw[order] (bot)--(1);
\draw[order] (bot)--(a);
\draw[order] (1)--(top);
\draw[order] (a)--(top);
\node[label,anchor=south] at (top) {\footnotesize$\top_\C = \top_\C a$\normalsize};
\node[label,anchor=north,xshift=1pt] at (bot) {\footnotesize$\bot_\C$\normalsize};
\node[label,anchor=east,xshift=1pt] at (1) {\footnotesize$1_\C = 0_\C = a^2$\normalsize};
\node[label,anchor=west,xshift=-1pt] at (a) {\footnotesize$a$\normalsize};
\node[label,anchor=north,yshift=-15pt] at (bot) {\footnotesize$\C = D^4_{3,2}$\normalsize};
\end{scope} 

\begin{scope}[xshift=-9cm, yshift=-4cm]
\node[draw,circle,inner sep=2pt,fill] (d) at (0,-2) {};
\node[draw,circle,inner sep=2pt,fill] (e) at (0,-1) {};
\node[draw,circle,inner sep=2pt,fill] (f) at (0,0) {};
\node[draw,circle,inner sep=2pt,fill] (g) at (0,1) {};
\draw[order] (d)--(e)--(f)--(g);
\node[label,anchor=west] at (d) {\footnotesize$\bot_\D$\normalsize};
\node[label,anchor=west] at (e) {\footnotesize$0_\D$\normalsize};
\node[label,anchor=west] at (f) {\footnotesize$1_\D$\normalsize};
\node[label,anchor=west] at (g) {\footnotesize$\top_\D=\top_\D 0_\D$\normalsize};
\node[label,anchor=north,yshift=-15pt] at (d) {\footnotesize$\mathbf{D}=\mathbf{S}_4$\normalsize};
\end{scope}


\begin{scope}[xshift=-1.3cm,yshift=-6cm]
\node[draw,circle,inner sep=2pt,fill] (bot) at (-1,0) {};
\node[draw,circle,inner sep=2pt] (a) at (-2,1) {};
\node[draw,circle,inner sep=2pt,fill] (1) at (0,1) {};
\node[draw,circle,inner sep=2pt] (0) at (-1,2) {};
\node[draw,circle,inner sep=2pt,fill] (b) at (1,2) {};
\node[draw,circle,inner sep=2pt,fill] (top) at (0,3) {};
\path [-] (bot) edge node {} (a);
\path [-] (bot) edge node {} (1);
\path [-] (a) edge node {} (0);
\path [-] (1) edge node {} (0);
\path [-] (1) edge node {} (b);
\path [-] (0) edge node {} (top);
\path [-] (b) edge node {} (top);   
\node[label,anchor=south,xshift=1pt] at (top) {\footnotesize$\top=0b=0a$\normalsize};
\node[label,anchor=east,xshift=1pt] at (0) {\footnotesize$0$\normalsize};
\node[label,anchor=west,xshift=-1pt] at (b) {\footnotesize$a =b^2$\normalsize};
\node[label,anchor=east,xshift=1pt] at (a) {\footnotesize$b=ab$\normalsize};
\node[label,anchor=west,xshift=1pt] at (1) {\footnotesize$1$\normalsize};
\node[label,anchor=north,xshift=1pt] at (bot) {\footnotesize$\bot$\normalsize};
\node[label,anchor=north,yshift=-15pt] at (bot) {\footnotesize$
 D^6_{4,6}$\normalsize};
\end{scope} 

\begin{scope}[xshift=6cm, yshift=-7cm]
\node[draw,circle,inner sep=2pt,fill] (a1) at (1,0) {};
\node[draw,circle,inner sep=2pt,fill] (a2) at (1,4) {};
\node[draw,circle,inner sep=2pt,fill] (b0) at (1,1) {};
\node[draw,circle,inner sep=2pt,fill] (b1) at (1,3) {};
\node[draw,circle,inner sep=2pt,fill] (b2) at (0,2) {};
\node[draw,circle,inner sep=2pt] (b3) at (2,2) {};
\draw[order] (b0)--(b2);
\draw[order] (b0)--(b3);
\draw[order] (b2)--(b1);
\draw[order] (b3)--(b1);
\draw[order] (a1)--(b0);
\draw[order] (b1)--(a2);
\node[label,anchor=east,xshift=1pt,yshift=2pt] at (b1) {\footnotesize$a=ab$\normalsize};
\node[label,anchor=east,yshift=-1pt] at (b0) {\footnotesize$c=bc=ac$\normalsize};
\node[label,anchor=east,xshift=1pt] at (b2) {\footnotesize$1=0=b^2$\normalsize};
\node[label,anchor=west,xshift=-1pt] at (b3) {\footnotesize$b$\normalsize};
\node[label,anchor=north,xshift=1pt] at (a1) {\footnotesize$\bot$\normalsize};
\node[label,anchor=south,xshift=1pt] at (a2) {\footnotesize$\top=\top c$\normalsize};
\node[label,anchor=north,yshift=-15pt] at (a1) {\footnotesize$
D^6_{3,3}$\normalsize};
\end{scope}

\end{tikzpicture}
\end{center}
\caption{Examples of DqRAs. The names of the algebras come from~\cite{CJR-DqRAs}. Black nodes denote
idempotents. All algebras are commutative. The products given for each algebra are
sufficient to recover the monoid operation.}\label{fig:small_examples_DqRAs}
\end{figure}

\subsection{Representable distributive quasi relation algebras}\label{subsec:representableDqRAs}

{\em Representable DqRAs} were  first defined in~\cite{CR-RDqRA}. While representable relation algebras are
embeddable in relation algebras built from equivalence relations (or sets), 
concrete  DqRAs are built from partially ordered equivalence relations (or posets).
The
lattice structure in the construction is the same as that used in~\cite{GJ20-AU, GJ20-ramics, JS23}.
Below we recall the construction, the definition of representable DqRAs, and some key results that will be needed in later sections.

For a set $X$, an equivalence relation $E$ on $X$, 
and binary relations $R,S\subseteq X^2$, let $R^\smallsmile$ denote the converse of $R$, i.e., $R^\smallsmile=\{(x,y)\mid (y,x)\in R\}$, $R^c$ the complement of $R$ relative to $E$, i.e., $R^c=\{(x,y)\in E\mid (x,y)\notin R\}$,
and $R\mathbin{;}S$ the relational composition of $R$ and $S$, i.e.,
$R\mathbin{;}S=\{(x,z)\mid(\exists y\in X)((x,y)\in R\text{ and }(y,z)\in S)\}$. Further, denote the identity relation on $X$ by $\text{id}_X$, i.e., $\text{id}_X =\{(x,x)\mid x\in X\}$.
Then, for $R,S,T\subseteq E$ we have $(R^\smallsmile)^\smallsmile=R$, $\;(R^c)^\smallsmile=(R^\smallsmile)^c$, $\;\text{id}_X\mathbin{;}R=R\mathbin{;}\text{id}_X=R$,
$\;(R\mathbin{;}S)\mathbin{;}T=R\mathbin{;}(S\mathbin{;}T)$ and $(R\mathbin{;}S)^\smallsmile=S^\smallsmile\mathbin{;}R^\smallsmile$.

If $\gamma: X \to X$ is a function, then $\gamma$ will also be used to denote its graph, i.e., the binary relation $\{(x,\gamma(x)) \mid x\in X\}$. 
%
Note that the lemma below applies when $\gamma$ is a bijective function $\gamma : X \to X$. 

\begin{lemma}
{\cite[Lemma 3.4]{CR-RDqRA}}
\label{lem:important_eq_injective_map}
Let $E$ be an equivalence relation on a set $X$, and let $R, \gamma \subseteq E$. If 
$\gamma$
satisfies 
$\gamma^{\smallsmile}\mathbin{;} \gamma = \text{id}_X$ and $\gamma \mathbin{;} \gamma^{\smallsmile}=\text{id}_X$
then the following hold:
\begin{multicols}{2}
\begin{enumerate}[label=(\roman*)]
\item $\left(\gamma\mathbin{;} R\right)^c = \gamma \mathbin{;} R^c$
\item $\left(R\mathbin{;} \gamma\right)^c = R^c \mathbin{;} \gamma$.
\end{enumerate}	
\end{multicols}
\end{lemma}

For a poset $\mathbf{X}=(X,\leqslant_X)$ and an equivalence
relation $E$ on $X$ such that ${\leqslant_X}\subseteq E$, define
a binary relation $\preccurlyeq_E$ on $E$ by $(x, y)\preccurlyeq_E (u, v)$ if, and only if, $u\leqslant_X x$ and $y\leqslant_X v$
for all $(x, y), (u, v) \in E$.  Then $\mathbf{E}=(E,\preccurlyeq_E)$ is a poset
which implies that $\mathsf{Up}(\mathbf{E})$, the set of up-sets 
of $\mathbf{E}$ ordered by inclusion, is a distributive lattice.
Observe that $\mathsf{Up}(\mathbf{E})$ is closed under
composition $\mathbin{;}$, ${\leqslant_X} \in \mathsf{Up}(\mathbf{E})$ and, 
for all $R \in \mathsf{Up}(\mathbf{E})$, we have $R \mathbin{;} {\leqslant_X} = R = {\leqslant_X} \mathbin{;} R$. Furthermore, composition $\mathbin{;}$ is residuated with residuals $\backslash$ and $/$ defined as follows: $R \backslash S = (R^{\smallsmile}\mathbin{;}S^c)^c$ and $R/S=(R^c\mathbin{;}S^{\smallsmile})^c$.  
Hence,
the algebra
$\left<\mathsf{Up}(\mathbf{E}),\cap,\cup,\mathbin{;},\backslash,/,\leqslant_X\right>$
is a distributive residuated lattice. 

Now let $\alpha:X\to X$ be an order automorphism of $\mathbf X$ such that $\alpha\subseteq E$. If we set $0 = \alpha \mathbin{;} {\leqslant_X^{c\smallsmile}} = {\leqslant_X^{c\smallsmile}}\mathbin{;} \alpha$, then $0  \in \mathsf{Up}(\mathbf{E})$ and $\left<\mathsf{Up}(\mathbf{E}),\cap,\cup,\mathbin{;},\backslash,/,\leqslant_X, 0\right>$ is an InFL-algebra. Moreover, 
${\sim} R = R^{c\smallsmile}\mathbin{;} \alpha$ and ${-}R=\alpha\mathbin{;} R^{c\smallsmile}$ (\cite[Lemma 3.10]{CR-RDqRA}).  

Let $\beta:X \to X$ be a self-inverse dual order automorphism of $\mathbf X$ such that $\beta \subseteq E$. Define 
$\neg R= \alpha\mathbin{;} \beta \mathbin{;} R^c \mathbin{;} \beta$.  Then $\mathsf{Up}(\mathbf{E})$
is closed under $\neg$ (see~\cite[Lemma 3.5]{CR-RDqRA}
for more details). If we require that $\beta = \alpha \mathbin{;}\beta \mathbin{;}\alpha$, we obtain the following result.

\begin{theorem}{\cite[Theorem 3.15]{CR-RDqRA}}\label{Theorem:Dq(E)}
Let $\mathbf{X}=\left(X,\leqslant_X\right)$ be a poset and $E$ an equivalence relation on $X$ such that ${\leqslant_X} \subseteq E$.  Let $\alpha: X \to X$ be an order automorphism of $\mathbf X$ and $\beta: X \to X$ a self-inverse dual order automorphism of $\mathbf X$ such that $\alpha, \beta \subseteq E$ and $\beta = \alpha \mathbin{;} \beta\mathbin{;} \alpha$. 
For 
$R \in \mathsf{Up}(\mathbf E)$, define
$\neg R =  \alpha\mathbin{;} \beta \mathbin{;} R^c \mathbin{;} \beta$.
Then the algebra $\mathbf{D}( E,\preccurlyeq_{E}) = \left\langle \mathsf{Up}\left(\mathbf{E}\right),\cap, \cup, \mathbin{;}, {\sim}, {-}, {\neg}, {\leqslant_X} \right\rangle$ 
is a distributive quasi relation algebra. If $\alpha$ is the identity, then $\mathbf{D}(E,\preccurlyeq_{E})$ is a cyclic distributive quasi relation algebra.	
\end{theorem}	

An algebra is said to be an \emph{equivalence distributive quasi relation algebra} if it is of the form described by Theorem~\ref{Theorem:Dq(E)}. $\mathsf{EDqRA}$
denotes the class of such algebras.  If $E=X^2$ then the
algebra $\mathbf{D}(E,\preccurlyeq_{E})$ is called a \emph{full distributive quasi relation algebra}. $\mathsf{FDqRA}$ denotes the
class of such algebras.
In~\cite[Theorem 4.4]{CR-RDqRA} it was shown that 
$\mathbb{IP}(\mathsf{FDqRA})=\mathbb{I}(\mathsf{EDqRA})$, 
analogous to the case for relation algebras (cf.~\cite[Chapter 3]{Mad06}).
\begin{definition}
{\cite[Definition 4.5]{CR-RDqRA}}\label{Definition:RDqRA}
A DqRA $\mathbf{A} = \left\langle A, \wedge, \vee, \cdot,  {\sim}, {-}, {\neg}, 1 \right\rangle$ 
is said to be \emph{representable} if 
$\mathbf{A} \in \mathbb{ISP}\left(\mathsf{FDqRA}\right)$
or, equivalently, $\mathbf{A} \in \mathbb{IS}\left(\mathsf{EDqRA}\right)$. 
\end{definition}

If the poset $(X, \leqslant_X)$ used in the representation of $\mathbf{A}$ is finite, then we say that $\mathbf{A}$ is  \emph{finitely} representable. If $\A$ is bounded and the embedding $\varphi$ used in the representation preserves the bounds,  we say that $\A$ is \emph{normally representable}. 

Below we give representations for the DqRAs $\A$, $\B$, $\C$ and $\D$ in Figure~\ref{fig:small_examples_DqRAs}. These representations will be used later to construct representations for the DqRAs with identifiers $D^6_{4, 6}$ and $D^6_{3,3}$. 

\begin{example}{\cite[Example 5.1]{CR-RDqRA}}\label{ex:AandCRepresentations}
Consider the algebras $\A=\mathbf{S}_3$
and $\C=D^{4}_{3,2}$ depicted in Figure~\ref{fig:small_examples_DqRAs}. Let $\mathbf{X}_1 = (X_1, \leqslant_{X_1})$ be the $2$-element antichain with $X_{1} = \left\{x, y\right\}$, and let $E_1 = (X_{1})^2$. Further, set the order automorphism $\alpha_1 = \left\{\left(x, y\right), \left(y, x\right)\right\}$ 
and the dual order automorphism $\beta_1= \textnormal{id}_{X_{1}}$ 
(all of which are depicted in Figure~\ref{fig:S3Representation}).
The lattice 
$\mathsf{Up}(E_1, \preccurlyeq_{E_1})$ is the $16$-element Boolean lattice and
$\alpha_1 \mathbin{;}\leqslant_{X_{1}}^{c\smallsmile} =
{\leqslant_{X_{1}}^{c\smallsmile}} \mathbin{;} \alpha_1
= {\sim}{\leqslant_{X_{1}}} 
=
{\leqslant_{X_{1}}}$.
Since $\{\varnothing, \leqslant_{X_{1}}, (X_{1})^2\}$ and 
$\{\varnothing, \leqslant_{X_{1}},\{(x,y),(y,x)\}, (X_{1})^2\}$ are subuniverses  of the DqRA
$\mathbf{D}(E_1, \preccurlyeq_{E_1}) = \langle \mathsf{Up}(E_1,\preccurlyeq_{E_1}),\cap, \cup, \mathbin{;},  {\sim},{-}, \neg, \leqslant_{X_{1}}\rangle$, the algebra
$\A=\mathbf{S}_3$ can be represented by the subalgebra
$\langle \{\varnothing, \leqslant_{X_{1}}, (X_{1})^2\} ,\cap, \cup, \mathbin{;},  {\sim},{-}, \neg, \leqslant_{X_{1}} \rangle$ of $\mathbf D(E_1, \preccurlyeq_{E_1})$,
and $\C=D^{4}_{3,2}$ can be represented by
$\langle \{\varnothing, \leqslant_{X_{1}},\{(x,y),(y,x)\}, (X_{1})^2\} ,\cap, \cup, \mathbin{;},  {\sim},{-}, \neg, \leqslant_{X_{1}} \rangle$ of $\mathbf{D}(E_1, \preccurlyeq_{E_1})$
{\upshape(}also depicted in 
Figure~\ref{fig:S3Representation}{\upshape)}. Note that both $\A$ and $\C$ are normally representable. 
\end{example}

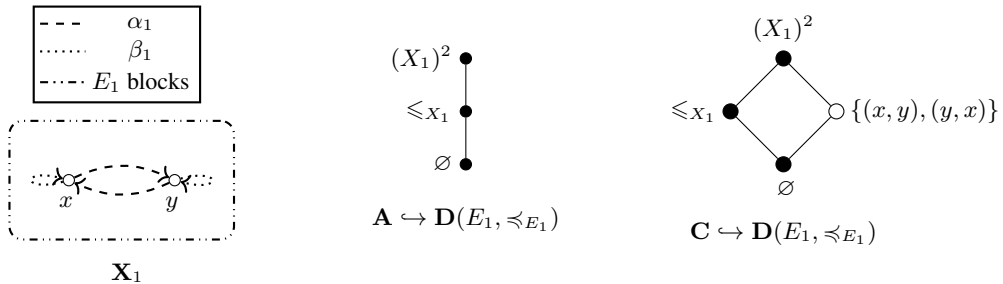
\begin{figure}[ht]
\centering
\begin{tikzpicture}[scale=0.7,pics/sample/.style={code={\draw[#1] (0,0) --(0.6,0) ;}},
    Dotted/.style={
    dash pattern=on 0.1\pgflinewidth off #1\pgflinewidth,line cap=round,
    shorten >=#1\pgflinewidth/2,shorten <=#1\pgflinewidth/2},
    Dotted/.default=3]
\begin{scope}[ box/.style = {draw,dashdotdotted,inner sep=20pt,rounded corners=5pt,thick}]
\node[draw,circle,inner sep=1.5pt] (x) at (-2.5,0.7) {};
\node[draw,circle,inner sep=1.5pt] (y) at (-0.5,0.7) {};
\path (x) edge [->, bend left=25, dashed, thick] node {} (y);
\path (y) edge [->, bend left=25, dashed, thick] node {} (x);
\draw [->,dotted, thick] (x) edge[loop left,looseness=35]node{} (x);
\draw [->,dotted, thick] (y) edge[loop right, looseness=35]node{} (y);
\node[label,anchor=north,xshift=-1pt] at (x) {\footnotesize$x$\normalsize};
\node[label,anchor=north,xshift=-1pt] at (y) {\footnotesize$y$\normalsize};
\path (0,4) 
 node[matrix,anchor=north east,draw,nodes={anchor=center},inner sep=2pt, thick]  {
  \pic{sample=dashed}; & \node{\footnotesize$\alpha_1$\normalsize}; \\
  \pic{sample=dotted}; & \node{\footnotesize$\beta_1$\normalsize}; \\
  \pic{sample=dashdotdotted}; & \node{\footnotesize$E_1$ blocks\normalsize}; \\
 };
 \node[box,fit=(x)(y)] {};
\node[label,anchor=north,xshift=22pt,yshift=-25pt] at (x) {\footnotesize$\mathbf{X}_{1}$\normalsize};
\end{scope}

\begin{scope}[xshift=5cm,scale=1, yshift=1cm]
\node[draw,circle,inner sep=1.5pt,fill] (bot) at (0,0) {};
\node[draw,circle,inner sep=1.5pt,fill] (id) at (0,1) {};
\node[draw,circle,inner sep=1.5pt,fill] (top) at (0,2) {};
\draw[order] (bot)--(id)--(top);
\node[label,anchor=east,xshift=1pt] at (top) {\footnotesize$(X_{1})^2$\normalsize};
\node[label,anchor=east,xshift=1pt] at (id) {\footnotesize${\leqslant_{X_{1}}}$\normalsize};
\node[label,anchor=east,xshift=1pt] at (bot) {\footnotesize$\varnothing$\normalsize};
\node[label,anchor=north,xshift=0pt,yshift=-10pt] at (bot) {\footnotesize$\A \hookrightarrow \mathbf{D}(E_1, \preccurlyeq_{E_1})$\normalsize};
\end{scope}

\begin{scope}[xshift=10cm, yshift=1cm]
\node[draw,circle,inner sep=2pt,fill] (bot) at (1,0) {};
\node[draw,circle,inner sep=2pt,fill] (top) at (1,2) {};
\node[draw,circle,inner sep=2pt,fill] (1) at (0,1) {};
\node[draw,circle,inner sep=2pt] (a) at (2,1) {};
\draw[order] (bot)--(1);
\draw[order] (bot)--(a);
\draw[order] (1)--(top);
\draw[order] (a)--(top);
\node[label,anchor=south] at (top) {\footnotesize$(X_{1})^2$\normalsize};
\node[label,anchor=north,xshift=1pt] at (bot) {\footnotesize$\varnothing$\normalsize};
\node[label,anchor=east,xshift=1pt] at (1) {\footnotesize${\leqslant_{X_{1}}}$\normalsize};
\node[label,anchor=west,xshift=-1pt] at (a) {\footnotesize$\{(x,y),(y,x)\}$\normalsize};
\node[label,anchor=north,yshift=-15pt] at (bot) {\footnotesize$\C \hookrightarrow \mathbf{D}(E_1, \preccurlyeq_{E_1})$\normalsize};
\end{scope}

\end{tikzpicture}
\caption{The poset $\mathbf{X}_{1}$ used to represent $\A$  and $\mathbf{C}$ (left), the
representation of $\A$ (middle) and the
representation of $\mathbf{C}$.}
\label{fig:S3Representation}
\end{figure}

\begin{example}{\cite[Example 5.7]{CR-RDqRA}}\label{ex:BRepresentation}
For the algebra $\B=D^4_{2,2}$ depicted in Figure~\ref{fig:small_examples_DqRAs}, let $\mathbf X_2=(X_2, \leqslant_{X_2})$ be the $2$-element antichain with $X_2=\{u, v\}$. Set $E_2= (X_2)^2$ and $\alpha_2 = \beta_2 =\textnormal{id}_{X_{2}}$.
Then $\B$ is isomorphic to the subalgebra $\langle \{\varnothing, \leqslant_{X_{2}},\{(u, v),(v, u)\}, (X_{2})^2\} ,\cap, \cup, \mathbin{;},  {\sim},{-}, \neg, \leqslant_{X_{2}}\rangle$
of 
$\mathbf{D}\left(E_2, \preccurlyeq_{E_2}\right)$. 
\end{example}

\begin{example}{\cite[Example 5.2]{CMR2026}}
\label{ex:DRepresentation}
The $4$-element Sugihara chain, $\mathbf{D}=\mathbf{S}_4$, also depicted in 
Figure~\ref{fig:small_examples_DqRAs}, can be (normally) represented over the poset $\mathbf X_3 = (X_3, \leqslant_{X_3})$ where  $X_3=\{s, t, u, v, w\}$ and
$$\leqslant_{X_3}=\textnormal{id}_{X_{3}}\cup
\left(\{s,t\}\times \{u\}\right)\cup
\left(\{u\}\times\{v,w\}\right)\cup
\left(\{s,t\}\times\{v,w\}\right).$$
Here $E_3 = (X_3)^2$,
$\alpha_3=\{(v,w),(w,v),(u,u),(s,t),(t,s)\}$ and 
$\beta_3=\{(s,v),(v,s),(u,u),(t,w),(w,t)\}$. 
Then  $\mathbf{D}$ is isomorphic to the subalgebra
$\langle \{\varnothing, \leqslant_{X_{3}},{\leqslant_{X_{3}}}{\setminus}
\{(u,u)\}, (X_{3})^2\} ,\cap, \cup, \mathbin{;},  {\sim},{-}, \neg, \leqslant_{X_{3}}\rangle$
of $\mathbf{D}(E_3, \preccurlyeq_{E_3})$.
See Figure~\ref{Fig:S4Representation} for a depiction of this poset and the representation of $\D = \mathbf S_4$.
\end{example}

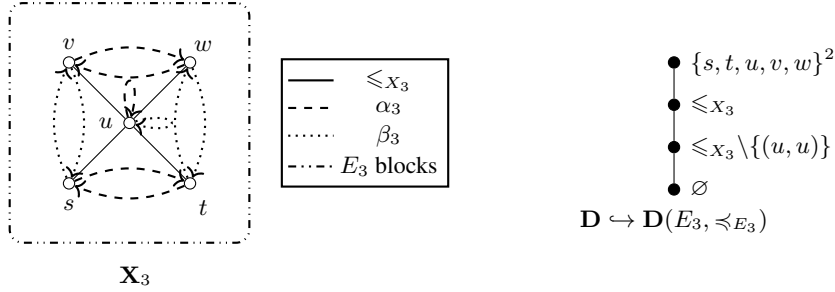
\begin{figure}[ht]
\centering
\begin{tikzpicture}[scale=0.8,pics/sample/.style={code={\draw[#1] (0,0) --(0.6,0) ;}},
Dotted/.style={
dash pattern=on 0.1\pgflinewidth off #1\pgflinewidth,line cap=round,
shorten >=#1\pgflinewidth/2,shorten <=#1\pgflinewidth/2},
Dotted/.default=3]
\begin{scope}[xshift=-3cm, yshift=0cm,box/.style = {draw,dashdotdotted,inner sep=20pt,rounded corners=5pt,thick}]
\node[draw,circle,inner sep=1.5pt] (a) at (-1,1) {};
\node[draw,circle,inner sep=1.5pt] (b) at (1,1) {};
\node[draw,circle,inner sep=1.5pt] (c) at (0,0) {};
\node[draw,circle,inner sep=1.5pt] (d) at (-1,-1) {};
\node[draw,circle,inner sep=1.5pt] (e) at (1,-1) {};
\draw[order] (a)--(c)--(d);
\draw[order] (b)--(c)--(e);
\node[label,anchor=south,yshift=-1pt] at (a) {\footnotesize$v$\normalsize};
\node[label,anchor=south,yshift=-2pt, xshift=5pt] at (b) {\footnotesize$w$\normalsize};
\node[label,anchor=north] at (d) {\footnotesize$s$\normalsize};
\node[label,anchor=north, xshift=5pt] at (e) {\footnotesize$t$\normalsize};
\node[label,anchor=east] at (c) {\footnotesize$u$\normalsize};
\path (a) edge [->, bend left=23, dashed, thick] node {} (b);
\path (b) edge [->, bend left=23, dashed, thick] node {} (a);
\path (d) edge [->, bend left=23, dashed, thick] node {} (e);
\path (e) edge [->, bend left=23, dashed, thick] node {} (d);
\draw [->,dashed, thick] (c) edge[loop above,looseness=40]node{} (c);
\path (a) edge [->, bend left=23, dotted, thick] node {} (d);
\path (b) edge [->, bend left=23, dotted, thick] node {} (e);
\path (d) edge [->, bend left=23, dotted, thick] node {} (a);
\path (e) edge [->, bend left=23, dotted, thick] node {} (b);
\draw [->,dotted, thick] (c) edge[loop right,looseness=40]node{} (c);
\node[box,fit=(a)(b)(d)(e)(c),] {};
\node[label,anchor=north,xshift=25pt,yshift=-25pt] at (d) {\footnotesize$\mathbf{X}_{3}$\normalsize};
\path (2.5,0) 
 node[matrix,anchor=west,draw,nodes={anchor=center},inner sep=2pt, thick]  {
 \pic{sample=solid}; & \node{\footnotesize$\leqslant_{X_3}$\normalsize}; \\
  \pic{sample=dashed}; & \node{\footnotesize$\alpha_3$\normalsize}; \\
  \pic{sample=dotted}; & \node{\footnotesize$\beta_3$\normalsize}; \\
  \pic{sample=dashdotdotted}; & \node{\footnotesize$E_3$ blocks\normalsize}; \\
 };
\end{scope}

\begin{scope}[xshift=6cm,yshift=0cm]
\node[draw,circle,inner sep=1.5pt,fill] (p) at (0,-1.1) {};
\node[draw,circle,inner sep=1.5pt,fill] (q) at (0,-0.4) {};
\node[draw,circle,inner sep=1.5pt,fill] (r) at (0,0.3) {};
\node[draw,circle,inner sep=1.5pt,fill] (s) at (0,1) {};
\draw[order] (p)--(q)--(r)--(s);
\node[label,anchor=west] at (s) {\footnotesize$\left\{s, t, u,v, w\right\}^2$\normalsize};
\node[label,anchor=west] at (r) {\footnotesize$\leqslant_{X_{3}}$\normalsize};
\node[label,anchor=west] at (q) {\footnotesize${\leqslant_{X_{3}}}{\setminus}\{(u,u)\}$\normalsize};
\node[label,anchor=west] at (p) {\footnotesize$\varnothing$\normalsize};
\node[label,anchor=north,yshift=-1pt] at (p) {\footnotesize$\mathbf{D}\hookrightarrow \mathbf{D}(E_3, \preccurlyeq_{E_3})$\normalsize};
\end{scope}

\end{tikzpicture}
\caption{The poset $\mathbf{X}_3$ used to represent $\mathbf{D}=\mathbf{S}_4$ (left) and the representation of $\D = \mathbf{S}_4$ (right).}
\label{Fig:S4Representation}
\end{figure}

\subsection{Frames}\label{subsec:frames}

Frames for DInFL-algebrasDInFL-algebras and DqRAs were first described in~\cite{CJR24}. These frames  are duals for complete perfect DInFL-algebras and DqRAs, respectively. The definitions given in~\cite{CJR24} use a ternary
relation $R$ to represent the monoid operation on the dual frame. In this  paper, we will use the definition of a DqRA-frame given in~\cite{CJR26}, where instead of $R$, a binary set operation $\circ$ is used. Since our focus is exclusively on DqRAs, we omit the prefix ``DqRA'' and simply refer to DqRA-frames as \emph{frames} for the remainder of the paper. In what follows, we state only those results that are needed for the present paper, and refer the reader to~\cite{CJR26} for a more detailed account of the duality results.

Given a binary set operation $\circ: F \times F \to \mathcal{P}(F)$, it can be extended to
a binary operation on $\mathcal P(F)$ via $U\circ V=\bigcup\{e\circ f\mid e\in U,f\in V\}$. We will  also make use of the notation $e\circ V=\{e\}\circ V$ and $U\circ f=U\circ \{f\}$.

\begin{definition}\label{def:frame}
A \emph{frame} is a tuple $\mathbf{F}= \left(F, I, \preccurlyeq_F, \circ, {^{\sim}}, {^{-}}, {^\neg}\right)$ with 
$I\subseteq F$, a partial order $\preccurlyeq_F$ on $F$, a binary set-operation $\circ:F\times F\to \mathcal{P}(F)$ 
and functions $^{\sim}, ^-, ^{\neg} : F \to F$ such that the following conditions hold for all $e, f, g, h \in F$:
\begin{multicols}{2}
\begin{enumerate}[
label=\textup{(\roman*)},
itemsep=0pt,
parsep=0pt,
topsep=0pt,
partopsep=0pt
]
\item $f
\in I \circ e$ iff $e \preccurlyeq_F f$ iff $f \in e\circ I$; 
\item If $e \preccurlyeq_F f$  and $e\in I$, then  $f\in I$;
\item If $e \preccurlyeq_F f$ and $e \in g \circ h$, then $f \in g \circ h$; 
\item $(e\circ f)\circ g = e \circ (f\circ g)$;
\item $g^\sim\in e\circ f$ iff $f^{-}\in g\circ e$;
\item $e^{\sim -} \preccurlyeq_F e$ and $  e^{-\sim} \preccurlyeq_F e$;
\item $e^{\neg\neg}=e$;
\item If $e\preccurlyeq_F f$, then $f^{\neg}\preccurlyeq_F e^{\neg}$;
\item $g^{-}\in e\circ f$ iff $g^{\neg}\in 
f^{\sim\neg}\circ e^{\sim\neg}$.
\item[] 
\end{enumerate}
\end{multicols}
\end{definition}

The following lemmas contain some useful results that will be used throughout this paper. 

\begin{lemma}{\cite[Lemma 3.3]{CJR26}}\label{lem:twiddle_minus_involutive}
Let  $\mathbf{F}= \left(F, I,\preccurlyeq_F, \circ, {^{\sim}}, {^{-}}, {^{\neg}}\right)$ be a frame. Then the following hold for all $e, f, g, h \in F$:
\begin{multicols}{2}
\begin{enumerate}[label=\textup{(\roman*)}]
\item $e^{\sim -} = e = e^{-\sim}$; 
\item If $e \preccurlyeq_F f$, then $f^- \preccurlyeq_F e^-$ and $f^\sim \preccurlyeq_F e^\sim$; 
\item If $e \preccurlyeq_F f$ and $g \in {f}\circ {h}$, then $g \in {e}\circ {h}$; 
\item If $e \preccurlyeq_F f$ and $g\in {h}\circ{f}$, then $g\in {h}\circ{e}$. 
\end{enumerate}
\end{multicols}
\end{lemma}

\begin{lemma}\label{lem:IcircI=I}
Let  $\mathbf{F}= \left(F, I,\preccurlyeq_F, \circ, {^{\sim}}, {^{-}}, ^{\neg}\right)$ be a frame. Then $I\circ I = I$. 
\end{lemma}

\begin{proof}
First, let $i \in I$. Now $i\preccurlyeq_F i$, so $i \in I \circ i$ by (i) of Definition~\ref{def:frame}. Hence, since $i \in I$, we get $i \in I \circ I$. 

Conversely, let $i \in I \circ I$. Then there exists some $j \in I$ such that $i \in j \circ I$. Hence, by (i) of Definition~\ref{def:frame}, $j \preccurlyeq_F i$. Since $j \in I$, we can apply (ii) of Definition~\ref{def:frame} to get $i \in I$. 
\end{proof}

Given a lattice $\mathbf{A}= \langle A, \wedge, \vee\rangle$, we denote by $J^\infty(\mathbf{A})$ and $M^\infty(\mathbf{A})$ the completely join-irreducible and completely meet-irreducible elements, respectively. 
We say that a lattice $\mathbf{A}$ is  \emph{perfect} if for every $a \in A$, we have 
$a = \bigvee \{\, j \in J^\infty(\mathbf{A}) \mid j \leqslant a \,\} = \bigwedge \{\, m \in M^\infty(\mathbf{A}) \mid a \leqslant m\,\}$, i.e., $J^\infty(\mathbf{A})$ 
 is join-dense and $M^\infty(\mathbf{A})$ is meet-dense.
Although some authors incorporate completeness into the definition of a perfect lattice, we treat completeness as a separate assumption.  We note that in a completely distributive lattice, every completely join-irreducible element is completely join-prime, and every completely meet-irreducible element is completely meet-prime.

The following proposition says that we can construct a complete perfect DqRA from a frame. 

\begin{proposition}{\cite[Proposition 3.13]{CJR26}}\label{prop:complex_algebra_frame}
Let $\mathbf{F}= \left(F, I,\preccurlyeq_F, \circ, {^{\sim}}, {^{-}}, {^{\neg}}\right)$ be a frame and let $\mathsf{Up}\left(F, \preccurlyeq_F\right)$ be  the set of all upsets
of $\left(F,\preccurlyeq_F\right)$.  For all $U \in \mathsf{Up}\left(F, \preccurlyeq_F\right)$, define
${\sim} U  = \left\{w \in F \mid w^- \notin U\right\}$,  
$-U  = \left\{w \in F \mid w^{\sim} \notin U\right\}$, and ${\neg} U  = \left\{w \in F \mid w^\neg \notin U\right\}$.
Then $\mathbf{F}^+ = \langle \mathsf{Up}\left(F, \preccurlyeq_F\right), \cap, \cup, \circ, \sim, -, \neg, I\rangle$ is a complete perfect DqRA. 
\end{proposition}

\noindent
The algebra $\mathbf{F}^+$ in Proposition~\ref{prop:complex_algebra_frame} is called the \emph{complex algebra} of $\mathbf{F}$.

Recall that if $\mathbf{A}$ is a complete perfect distributive lattice,   
then the  map $\kappa : J^{\infty}\left(\mathbf{A}\right) \to M^{\infty}\left(\mathbf{A}\right)$ defined by $\kappa\left(j\right) = \bigvee\left\{a \in A \mid j \not\leqslant a\right\}$ is an order isomorphism (\cite[Section II.5]{FJN1995}). 
The following lemma will be used in Section~\ref{sec:networks_representations}.

\begin{lemma}\label{lem:properties_kappa}
Let $\A=\langle A,\wedge,\vee\rangle$ be a complete perfect lattice.
\begin{enumerate}[label=\textup{(\roman*)}]
\item For each $j \in J^{\infty}(\A)$, we have $j \not\leqslant \kappa(j)$.
\item For all $a \in A$ and $j \in J^{\infty}(\A)$, we have $a \leqslant \kappa(j)$ iff $j \not\leqslant a$. 
\end{enumerate}
\end{lemma}

\begin{proof}
(i) Suppose $j \leqslant \kappa(j)$. Then, since $j$ is join-prime, there is some $a \in A$ such that $j \leqslant a$ and $j \not\leqslant a$, a contradiction. 

\noindent
(ii) First, if $j \not\leqslant a$, then $a \in \{d \in A \mid j \not\leqslant d\}$, and so $a \leqslant \bigvee \{d \in A \mid j \not\leqslant d\} = \kappa(j)$. 
For the converse implication, assume $a \leqslant \kappa(j)$. If $j \leqslant a$, then $j \leqslant \kappa(j)$, which contradicts (i). Hence, it must be the case that $j \not\leqslant a$. 
\end{proof}

For every completely join-irreducible $j$ of a complete perfect DqRA $\mathbf{A}$, define $j^{\sim} = {\sim}\kappa\left(j\right)$, $j^- = -\kappa\left(j\right)$, and $j^{\neg} = \neg\kappa(j)$. If $j \in J^{\infty}(\A)$, then $\kappa(j) \in M^{\infty}(\A)$, and so, since the unary operations $-, {\sim}$ and $\neg$ are dual
lattice isomorphisms, we have $j^{\sim}, j^{-}, j^{\neg} \in J^{\infty}(\A)$. 


\begin{proposition}{\cite[Proposition 3.16]{CJR26}}\label{prop:frame_from_perfect_DqRA}
Let $\mathbf{A}=\langle A,\wedge, \vee, \cdot, \sim, -, \neg, 1\rangle$ be a complete perfect DqRA. 
Set $A_+ = J^{\infty}(\A)$ and $I_{\A_+}=\left\{i \in A_+ \mid i \leqslant 1\right\}$. For all $i, j, k$ in $A_+$, define $\preccurlyeq_{\A_+} {=} \geqslant$, $k\in i \circ_{\A_+} j$ if, and only if, $k \leqslant i\cdot j$, $j^{\sim_{\A_+}} = {\sim}\kappa\left(j\right)$,  $j^{-_{\A_+}} = -\kappa\left(j\right)$, and $j^{\neg_{\A_+}} = \neg \kappa(j)$.
Then the structure  $\mathbf{A}_{+} =\left(A_+, I_{\A_+}, \preccurlyeq_{\A_+}, \circ_{\A_+}, {^{\sim_{\A_+}}}, {^{-_{\A_+}}}, ^{\neg_{\A_+}}\right)$ is a frame. 
\end{proposition}

If $\A$ is a complete perfect DqRA, we will sometimes omit the subscripts and write $\A_+ = (A_+, I, \preccurlyeq, \circ, ^{\sim}, ^{-}, ^{\neg})$ to denote its dual frame.  

\begin{example}\label{ex:dual_frames_ex}
Figure~\ref{fig:dual-frames} shows the dual frames of the DqRAs $\A$, $\B$, $\C$ and $\D$ in Figure~\ref{fig:small_examples_DqRAs}. 
Observe that $I_{\A_+}=\{1_\A\}$, $I_{\B_+}=\{1_\B\}$, $I_{\C_+}=\{1_\C\}$, and  $I_{\D_+}=\{1_\D, 0_\D\}$. In the tables for the binary circle, a sequence of elements such as $0_\D1_\D\top_\D$ denotes the set $\{0_\D,1_\D, \top_\D\}$.
The operation $^-$ is not shown since it is the inverse of $^\sim$, and for the frames in Figure~\ref{fig:dual-frames} we have $j^-= j^{\sim}$ for all join-irreducibles $j$.
\end{example}

\begin{figure}
\begin{tikzpicture}[xshift=-10cm, scale=0.7]
\footnotesize
\begin{scope}[xshift=-2cm]
\node[draw,circle,inner sep=2pt] (1) at (2,0.5) {};
\node[draw,circle,inner sep=2pt] (0) at (2,-0.5) {};
\draw[order] (1)--(0);
\node[label,anchor=west] at (0) {$\top_\A$};
\node[label,anchor=west] at (1) {$1_\A$};
\node[label,anchor=south,xshift=-30pt,yshift=-10pt] at (1) {$\A_+$:};
\node[anchor=west] at (3.5,0){
\footnotesize{$\begin{array}{|c|cc||c|c|}\hline
\circ_{\A_+} &1_\A&\ \top_\A & ^{\sim_{\A_+}} & ^{\neg_{\A_+}}\\\hline
1_{\A} & 1_\A & 1_\A\top_\A & \top_\A & \top_\A	\\
\top_\A & 1_\A\top_\A & 1_\A\top_\A & 1_\A & 1_\A	\\\hline
\end{array}$}\\[8pt]
};
\end{scope}

\begin{scope}[xshift=-1.1cm,yshift=-5.2cm]
\node[draw,circle,inner sep=2pt] (0) at (0,1) {};
\node[draw,circle,inner sep=2pt] (1) at (0,0) {};
\node[draw,circle,inner sep=2pt] (2) at (0,-1) {};
\draw[order] (2)--(1)--(0);
\node[label,anchor=west] at (0) {$0_\D$};
\node[label,anchor=west] at (1) {$1_\D$};
\node[label,anchor=west] at (2) {$\top_\D$};
\node[label,anchor=south,xshift=-30pt,yshift=-10pt] at (0) {$\D_+$:};
\node[anchor=west] at (1.5,0){
$\begin{array}{|c|ccc||c|c|}\hline
\circ_{\D_+} & 0_\D\ & 1_\D & \top_\D &^{\sim_{\D_+}} & ^{\neg_{\D_+}} \\\hline
0_\D & 0_\D & 0_\D & 0_\D1_\D\top_\D & \top_\D & \top_\D	\\
1_\D & 0_\D & 0_\D1_\D & 0_\D1_\D\top_\D & 1_\D & 1_\D \\
\top_\D & 0_\D1_\D\top_\D & 0_\D1_\D\top_\D & 0_\D1_\D\top_\D & 0_\D & 0_\D \\\hline
\end{array}$
};
\end{scope}





\begin{scope}[xshift=-5cm, yshift=-2.65cm]
\node[draw,circle,inner sep=2pt] (1) at (0,0) {};
\node[draw,circle,inner sep=2pt] (0) at (1.2,0) {};
\node[label,anchor=north] at (0) {$0_\B$};
\node[label,anchor=north] at (1) {$1_\B$};
\node[label,anchor=south,xshift=-10pt,yshift=8pt] at (1) {$\B_+$:};
\node[anchor=west] at (1.95,0.25){
$\begin{array}{|c|cc||c|c|}\hline
\circ_{\B_+} &1_\B & 0_\B &^{\sim_{\B_+}}& ^{\neg_{\B_+}}\\\hline
1_\B & 1_\B & 0_\B & 1_\B & 1_\B \\
0_\B & 0_\B & 1_\B & 0_\B & 0_\B	\\\hline
\end{array}$\\[8pt]
};
\end{scope}

\begin{scope}[xshift=6.3cm, yshift=-2.65cm]
\node[draw,circle,inner sep=2pt] (1) at (0,0) {};
\node[draw,circle,inner sep=2pt] (0) at (1.2,0) {};
\node[label,anchor=north] at (0) {$a$};
\node[label,anchor=north] at (1) {$1_\C$};
\node[label,anchor=south,xshift=-10pt,yshift=8pt] at (1) {$\C_+$:};
\node[anchor=west] at (1.95,0.25){
$\begin{array}{|c|cc||c|c|}\hline
\circ_{\C_+} &1_\C & a &^{\sim_{\C_+}}& ^{\neg_{\C_+}}\\\hline
1_\C & 1_\C & a & a & a \\
a & a & 1_\C & 1_\C & 1_\C	\\\hline
\end{array}$\\[8pt]
};
\end{scope}

\normalsize
\end{tikzpicture}

\caption{The dual frames of the DqRAs $\A$, $\B$, $\C$ and $\D$ in Figure~\ref{fig:small_examples_DqRAs}. 
}\label{fig:dual-frames}
\end{figure}
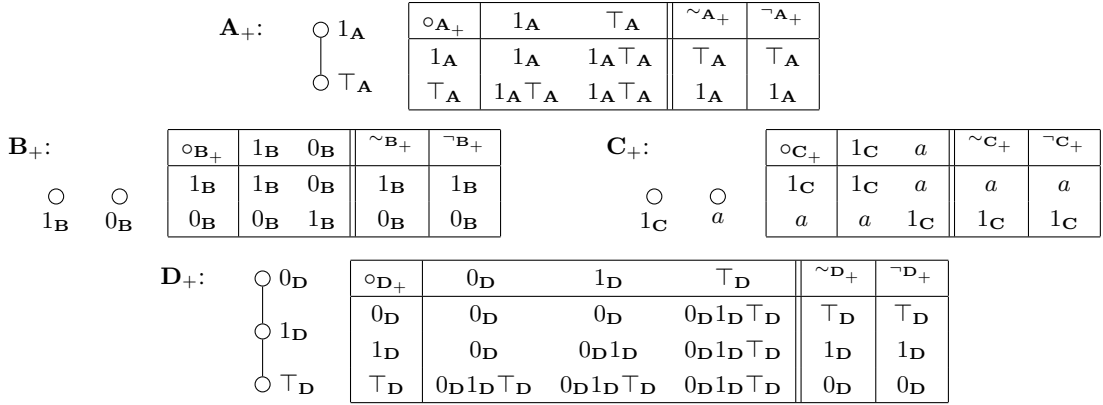



\section{Wreath products}\label{sec:wreath_products}

In this section, we adapt Comer's wreath-product construction in~\cite{Com83}, which he uses to build “big” polygroups from “small” ones.
Comer shows that important classes of polygroups, such as chromatic polygroups, are closed under the wreath-product construction. 
Here, we define the wreath product of two DqRA frames to obtain a larger DqRA frame. In the sequel, we use this construction to obtain representable DqRAs. For convenience, we omit subscripts from the relations and operations of the wreath product.   

Let $\A=\langle A,\wedge_\A,\vee_\A, \cdot_\A, {\sim}_\A, -_\A,\neg_\A, 1_\A\rangle$
and $\B=\langle B,\wedge_\B,\vee_\B,\cdot_\B, {\sim}_\B, -_\B,\neg_\B, 1_\B\rangle$ be
complete perfect DqRAs such that  $A \cap B=\varnothing$ and $1_\mathbf{A}\in \JA$, $1_{\B} \in \JB$ and $1_\B^{-_{\B_+}} = 1^{\sim_{\B_+}}_\B = 1_\B^{\neg_{\B_+}}=1_\B$.  Recall that $\A_+$ and $\mathbf{B}_+$ denote
their corresponding frames, respectively.

Now set $A_+[B_+]
= 
\JA\cup (\JB{\setminus}\{1_\B\})$ and $I = I_{\A_+} \cup (I_{\B_+}{\setminus}\{1_\B\})$. 
Define a binary relation $\preccurlyeq$ on $A_+[B_+]$ as follows:
\begin{align*}
i \preccurlyeq j \quad & \textnormal{ iff } \quad i, j \in \JA \textnormal{ and } i \preccurlyeq_{\A_+} j \notag \\
\quad & \textnormal{ or } \quad  i, j \in \JB{\setminus}\{1_\B\} \textnormal{ and } i \preccurlyeq_{\B_+} j \notag\\
\quad & \textnormal{ or } \quad  i \in \JA,  j \in \JB{\setminus} \{1_\B\} \textnormal{ and } 1_\B \prec_{\B_+} j \notag \\
\quad & \textnormal{ or } \quad  i \in \JB{\setminus} \{1_\B\}, j \in \JA \textnormal{ and } i \prec_{\B_+} 1_\B. 
\end{align*}

The proof of the following lemma is routine and left for the reader. 

\begin{lemma}\label{lem:wreat_product_poset}
Let $\A=\langle A,\wedge_\A,\vee_\A, \cdot_\A, {\sim}_\A, -_\A,\neg_{\A},1_\A\rangle$
and $\B=\langle B,\wedge_\B,\vee_\B,\cdot_\B,  {\sim}_\B, -_\B,\neg_{\B},1_\B\rangle$ be
complete perfect DqRAs such that $A \cap B=\varnothing$,  $1_\mathbf{A}\in \JA$, $1_{\B} \in \JB$ and $1_\B^{-_{\B_+}} = 1^{\sim_{\B_+}}_\B = 1_{\B}^{\neg_{\B_+}} =1_\B$.
Then $(A_+[B_+], \preccurlyeq)$ is a poset. 
\end{lemma}

Now define $\circ: A_+[B_+]\times A_+[B_+] \to \mathcal{P}(A_+[B_+])$  by setting, for all $i, j\in A_+[B_+]$,
\begin{equation*}
i\circ j = 
\begin{cases} 
(i \circ_{\mathbf{A}_+} j) \cup (I_{\B_+}{\setminus}\{1_\B\}) &  \textnormal{ if } i, j \in \JA\\
i \circ_{\B_+} j &  \textnormal{ if } i, j \in \JB {\setminus} \{1_\B\} \textnormal{ and } 1_\B \notin i\circ_{\B_+} j\\
((i \circ_{\B_+} j) {\setminus} \{1_\B\}) \cup \JA &  \textnormal{ if } i, j \in \JB {\setminus} \{1_\B\} \textnormal{ and } 1_\B \in i\circ_{\B_+} j\\
1_\B \circ_{\B_+} j &  \textnormal{ if } i \in \JA, j \in \JB {\setminus}\{1_\B\} \textnormal{ and } 1_\B \notin 1_\B\circ_{\B_+} j\\
((1_\B \circ_{\B_+} j) {\setminus} \{1_\B\}) \cup \JA &  \textnormal{ if } i \in \JA, j \in \JB {\setminus} \{1_\B\} \textnormal{ and } 1_\B \in 1_\B\circ_{\B_+} j\\
i \circ_{\B_+} 1_\B &  \textnormal{ if } i \in  \JB {\setminus}\{1_\B\}, j \in \JA \textnormal{ and } 1_\B \notin i\circ_{\B_+} 1_\B\\
((i \circ_{\B_+} 1_\B) {\setminus} \{1_\B\}) \cup \JA &  \textnormal{ if } i \in \JB {\setminus} \{1_\B\}, j \in \JA \textnormal{ and } 1_\B \in i\circ_{\B_+} 1_\B. 
\end{cases}
\raisetag{-3\baselineskip}
\end{equation*}

For each $i \in A_+[B_+]$, define $^\sim, ^-, ^{\neg}: A_+[B_+] \to A_+[B_+]$ 
by
\vspace{-0.5cm}
\begin{multicols}{2}
\begin{itemize}  
\item[]
\[ 
i^{\sim} = 
\begin{cases} 
i^{\sim_{\A_+}} & \textnormal{ if } i \in \JA\\
i^{\sim_{\B_+}} & \textnormal{ if } i \in \JB{\setminus} \{1_\B\},
\end{cases}
\]
\item[] 
\[ 
i^{-} = 
\begin{cases} 
i^{-_{\A_+}} & \textnormal{ if } i \in \JA\\
i^{-_{\B_+}} & \textnormal{ if } i \in \JB{\setminus} \{1_\B\}
\end{cases}
\]
\end{itemize}
\end{multicols}
\noindent
and
\[ 
i^{\neg} = 
\begin{cases} 
i^{\neg_{\A_+}} & \textnormal{ if } i \in \JA\\
i^{\neg_{\B_+}} & \textnormal{ if } i \in \JB{\setminus} \{1_\B\}.
\end{cases}
\]

\begin{theorem}\label{thm:wreat_product_DqRA-frames}Let $\A=\langle A,\wedge_\A,\vee_\A, \cdot_\A, {\sim}_\A, -_\A, \neg_\A, 1_\A\rangle$
and $\B=\langle B,\wedge_\B,\vee_\B,\cdot_\B,  {\sim}_\B, -_\B,\neg_\B, 1_\B\rangle$ be
complete perfect DqRAs such that $A \cap B=\varnothing$, $1_\mathbf{A}\in \JA$, $1_{\B} \in \JB$ and $1_\B^{-_{\B_+}} = 1^{\sim_{\B_+}}_\B = 1_\B^{\neg_{\B_+}}= 1_\B$. Then 
$\A_+[\B_+] = (A_+[B_+],I,\preccurlyeq, \circ, ^{\sim}, ^{-},^{\neg})$  
is a frame. 
\end{theorem}

\begin{proof}
We outline how conditions (i)--(ix) of Definition~\ref{def:frame} can be shown to hold.

(i)  Let $i, j \in A_+[B_+]$.  We only outline the proof of $j\in I\circ i$ iff $i\preccurlyeq j$.  The proof  of $j\in i\circ I$ iff $i\preccurlyeq j$ is similar.  Firstly, assume $i \preccurlyeq j$. We have to find some $k \in I$ such that $j \in k \circ i$.

If  $i, j \in A_+$, then $i \preccurlyeq_{\mathbf{A}_+} j$ and it follows from the fact that $\mathbf{A}_+$ satisfies (i) of Definition~\ref{def:frame}, that there exists some $k \in I_{\A_+}\subseteq I$ such that $j \in k \circ_{\A_+} i$.
If $i,j \in B_+ {\setminus} \{1_\B\}$, it follows easily that $j\in 1_\A\circ i$.  Then, since $1_\A\in I_{\A_+}\subseteq I$, we let $k=1_\A$.

If $j \in A_+$ and $i \in B_+{\setminus} \{1_\B\}$, we have $i \prec_{\B_+} 1_\B$, and hence $1_\B <_B 1_\B \cdot_\B i$. This implies $1_\B \in 1_\B \circ_{\B_+} i$, which means $1_\A \circ i = ((1_\B \circ_{\B_+} i){\setminus} \{1_\B\} \cup A_+$, and so $j \in 1_\A \circ i$. Moreover, $1_\A \in I$, so again we take $k = 1_\A$.

If $i \in A_+$ and $j \in B_+{\setminus} \{1_\B\}$, then it must be the case that $1_\B \prec_{\B_+} j$. Hence, $j <_\B 1_\B$, which means $j \in I_{\B_+} {\setminus}\{1_\B\} \subseteq I$. Now $j \leqslant_B j \cdot_\B 1_\B$, so $j \in j \circ_{\B_+} 1_\B$. Since $j <_\B 1_{\B}$, we have $1_\B \notin j \circ_{\B_+} 1_\B$, and thus $j \circ i = j \circ_{\B_+} 1_\B$. Therefore, $j \in j \circ i$, and in this case we take $k = j$. 



Conversely, assume $j \in I \circ i$. Then there is some $k \in I$ such that $j \in k \circ i$. If $i, j \in A_+$ and $k \in I_{\B_+}{\setminus}\{1_\B\} \subseteq B_+{\setminus} \{1_\B\}$, we have  $1_\B \in k \circ_{\B_+} 1_{\B}$, and so, $1_\B \leqslant_\B k \cdot_\B 1_\B = k <_\B 1_B$, which is a contradiction. 
The remaining seven cases each fall in one of the following categories.

If $i, j, k \in A_+$,  $k\in A_+$ and $i, j \in B_{+}{\setminus}\{1_\B\}$, or  
$i, j, k \in B_+{\setminus}\{1_\B\}$, it follows from the definition of $\circ$ and the fact that $\A_+$ and 
$\B_+$ satisfy (i) of Definition~\ref{def:frame} that $i \preccurlyeq_{\A_+} j$ or 
$i \preccurlyeq_{\B_+} j$, respectively. Hence $i \preccurlyeq j$.

If  $j, k \in A_+$ and $i \in B_+{\setminus} \{1_\B\}$, then $1_\B \in 1_\B \circ_{\B_+} i$, which means $i \prec_{\B_+} 1_\B$, and so $i \preccurlyeq j$.

If $j \in A_+$ and $i, k \in B_+{\setminus} \{1_\B\}$, then $1_\B \in k \circ_{\B_+} i$, so we can use the fact that $\B_+$ satisfies (i) of Definition~\ref{def:frame} to get $i \prec_{\B_+} 1_\B$, which implies $i \preccurlyeq j$. 

If $i, k \in A_+$ and $j \in B_+{\setminus}\{1_\B\}$, then it must be the case that $j \in I_{\B_+}{\setminus}\{1_\B\}$. Hence, $j <_\B 1_\B$, and so $1_B \prec_{\B_+} j$. This shows $i \preccurlyeq j$.

Finally, if $i \in A_+$ and $j, k \in B_+{\setminus} \{1_{\B}\}$, then $k \in I_{\B_+}$ and $j \in k\circ_{\B_+} 1_{\B}$. Since $\B_+$ satisfies (i) of Definition~\ref{def:frame}, it follows that $1_\B \preccurlyeq_{\B_+} j \neq 1_\B$, so $1_\B \prec_{\B_+} j$, and hence $i \preccurlyeq j$.

(ii)
Let $i \in I$
and assume $i \preccurlyeq j$. If $i \in I_{\B_+}{\setminus} \{1_\B\}$ and $j \in A_+$, then $1_\B \prec_{\B_+} i$ and $i \prec_{\B_+} 1_\B$, which is impossible. 
If $i \in I_{\A_+}$ and $j \in A_+$, or $i \in I_{\B_+}{\setminus} \{1_\B\}$ and $j \in B_+{\setminus}\{1_\B\}$, then $j\in I_{\A_+}$ or  $j \in I_{\B_+}{\setminus}\{1_\B\} \subseteq I$, respectively, since $\mathbf{A}_+$ and $\B_+$ satisfy (ii) of Definition~\ref{def:frame}.


If $i \in I_{\A_+}$ and $j \in B_+{\setminus}\{1_\B\}$, then we have $1_\B \prec_{\B_+} j$. Hence, $j <_\B 1_\B$, which implies $j \in I_{\B_+}{\setminus}\{1_\B\} \subseteq I$.



(iii) 
Assume 
$i \preccurlyeq j$ and $i \in k \circ \ell$. If $j, k, \ell \in A_+$ and $i \in B_+{\setminus}\{1_\B\}$ we have $1_\B <_{\B_+} i <_{\B_+} 1_\B$, which is a contradiction. The remaining fifteen cases can be shown in one of the following ways.

If $i,j, k, \ell \in A_+$, 
$i,j,k,\ell\in B_+{\setminus}\{1_\B\}$, 
$k \in A_+$ and $i, j, \ell \in B_+{\setminus}\{1_\B\}$,
or 
$i, j, k \in B_+{\setminus}\{1_\B\}$ and $\ell \in A_+$, it follows
directly from the fact that $\A_+$ and $\B_+$ satisfy (iii) of Definition~\ref{def:frame} that $j \in k \circ_{\A_+} \ell \subseteq k \circ \ell$, or
$j\in k\circ_{\B_+} \ell\subseteq k\circ \ell$, or
$j\in 1_\B\circ_{\B_+}\ell\subseteq k\circ \ell$
or $j\in k\circ_{\B_+}1_\B\subseteq k\circ \ell$, respectively.

If $i, j, k \in A_+$ and $\ell \in B_+{\setminus} \{1_\B\}$, $i, j \in A_+$ and $k, \ell \in B_+{\setminus}\{1_\B\}$, 
$i, k \in B_+{\setminus}\{1_\B\}$ and $j, \ell \in  A_+$, 
$j, k \in A_+$ and $i, \ell \in B_+{\setminus}\{1_\B\}$,
$j \in A_+$ and $i, k, \ell \in B_+{\setminus}\{1_\B\}$,
or 
$i, j, \ell \in A_+$ and $k \in B_+{\setminus}\{1_\B\}$, we have 
$A_+\subseteq k \circ \ell$.
Hence, since $j \in A_+$ in each of these cases, we have $j \in k \circ \ell$. 

If $i \in A_+$ and $j, k, \ell \in B_+{\setminus} \{1_\B\}$,
$i, k \in A_+$ and $j, \ell \in B_+{\setminus}\{1_\B\}$, or 
$i,\ell \in A_+$ and $j, k \in B_+{\setminus}\{1_\B\}$, then $1_\B \in k \circ_{\B_+} \ell$, $1_{\B} \in 1_\B \circ_{\B_+} \ell$, or $1_\B\in k \circ_{\B_+} 1_\B$, respectively.
Moreover, in all three cases, $1_\B\prec_{\B_+} j$
since $i\preccurlyeq j$. Hence, using the fact that 
$\B_+$ satisfies (iii) of Definition~\ref{def:frame}, it follows that $j \in k\circ \ell$ in each case.

If $k, \ell \in A_+$ and $i, j \in B_+{\setminus}\{1_\B\}$
or $i, k, \ell \in A_+$ and $j \in B_+{\setminus}\{1_\B\}$,
then $j <_{\B_+} 1_{\B}$, which implies that 
$j\in I_{\B_+}\backslash\{1_\B\}$.  Hence, $j\in k\circ \ell$.

(iv) We only prove one inclusion since the other inclusion is similar. Let $\ell \in (i\circ j) \circ k$. Then there exists some $m \in A_+[B_+]$ such that $m \in i \circ j$ and $\ell \in m \circ k$. Depending on whether each of $i, j, k, \ell$ and $m$ belongs to $A_+$ or $B_+$, there are thirty-two cases to consider. We highlight a few of the more involved cases here.  The other cases follow similar arguments and are left to the reader.
 

Assume $i, j, k, \ell \in A_+$ and $m \in B_+{\setminus}\{1_\B\}$. Then $m \in I_{\B_+}{\setminus}\{1_\B\}$ and $1_\B \in m \circ_{\B_+} 1_\B$. Since $j, k \in A_+$, we have $j \circ k = (j \circ_{\A_+} k) \cup (I_{\B_+}{\setminus}\{1_\B\})$, so $m \in j \circ k$. Now $1_\B \in m \circ_{\B_+} 1_\B$ implies $1_\B \leqslant_\B m \cdot_\B 1_\B = 1_\B \cdot_\B m$, which means $1_\B \in 1_\B \circ_{\B_+} m$. Hence, $i \circ m = ((1_\B \circ_{\B_+} m){\setminus} \{1_\B\}) \cup A_+$, and so $\ell \in i \circ m$. We thus obtain $\ell \in i \circ (j \circ k)$.  

Next, assume that $i, j, \ell \in A_+$ and $k, m \in B_{+}{\setminus}\{1_\B\}$. Then $m <_\B 1_\B$ and $1_\B \in m \circ_{\B_+} k$.  Since $m <_\B 1_\B$, we get $m \cdot_\B k \leqslant_\B 1_\B \cdot_\B k$. From $1_\B \in m \circ_{\B_+} k$, we get $1_\B \leqslant_\B m \cdot_\B k$, and so $1_\B \leqslant_\B m \cdot_\B k$, which implies $1_\B \in 1_\B \circ_{\B_+} k$. It follows that $i \circ k = ((1_\B \circ_{\B_+} k){\setminus}\{1_\B\}) \cup A_+$, and hence $\ell \in i \circ k$. Since $1_\B \in 1_\B \circ_{\B_+} k$, we also know $j \circ k = ((1_\B \circ_{\B_+} k){\setminus}\{1_\B\}) \cup A_+$. Now $k \leqslant_\B k = 1_\B \cdot_\B k$ and $k \neq 1_\B$, so $k \in (1_\B \circ_{\B_+} k){\setminus}\{1_\B\} \subseteq j \circ k$. This shows $\ell \in i \circ (j\circ k)$.

Now assume that $k, m \in A_+$ and $i, j, \ell \in B_+{\setminus} \{1_\B\}$. Here $1_\B \in i \circ_{\B_+} j$ and $\ell <_{\B} 1_\B$. The first part implies $1_\B \leqslant_\B i \cdot_\B j$, and so $\ell <_\B i \cdot_\B j$. This shows $\ell \in ((i \circ_{\B_+} j){\setminus} \{1_\B\}) \cup A_+ =i \circ j $. Now $j \leqslant_\B j \cdot_\B 1_\B$,  so $j \in j \circ_{\B_+} 1_\B$. Hence, since $j \neq 1_\B$, we get $j \in (j \circ_{\B_+} 1_{\B}){\setminus}\{1_\B\} \subseteq j \circ k$. It follows that $\ell \in i \circ (j \circ k)$. 

Considering one final case, assume that $\ell \in A_+$ and $i, j, k, m \in B_+ {\setminus} \{1_\B\}$. Then we have $m \in i \circ_{\B_+} j$ and $1_\B \in m \circ_{\B_+} k$, which means $1_\B \in (i \circ_{\B_+} j) \circ_{\B_+} k = i \circ_{\B_+} (j \circ_{\B_+} k)$. Hence, there exists some $n \in B_+ {\setminus} \{1_\B\}$ such that $n \in j \circ_{\B_+} k$ and $1_\B \in i \circ_{\B_+} n$. The second part implies that $i \circ n = ((i \circ_{\B_+} n) {\setminus} \{1_\B\}) \cup A_+$, so $\ell \in i \circ n$. Since $n \in j \circ_{\B_+} k$ and $n \neq 1_\B$, it must be the case that $n \in j \circ k$. Consequently, $\ell \in i \circ (j \circ k)$. 



(v) 
To prove the forward implication, assume that $k^{\sim} \in i \circ j$. Then there are eight cases to consider.

If $i, j, k \in A_+$ or $i,j,k \in B_+{\setminus}\{1_\B\}$, then, since $\A_+$ and $\B_+$ both satisfy (v) of Definition~\ref{def:frame}, we have
$j^{-}=j^{{-}_{\A_+}}\in k\circ_{\A_+} i$ or $j^{-}=j^{{-}_{\B_+}}\in k\circ_{\B_+} i$, respectively. In the first subcase, we get $j^- \in k\circ i$. In the second subcase, since $j \neq 1_\B$ and $1_\B^{{\sim}_{\B_+}} = 1_\B$, we have $j^{-_{\B_+}} \neq 1_\B$, and so  $j^{-}=j^{{-}_{\B_+}}\in (k\circ_{\B_+} i){\setminus}\{1_\B\} \subseteq k\circ i$.
For the three cases that have $j\in A_+$, we show that $A_+\subseteq k\circ i$
and therefore $j^{-}=j^{{-}_{\A_+}}\in k\circ i$.  We consider the case with $i, j \in A_+$ and $k \in B_+{\setminus} \{1_\B\}$ and leave the other two cases for the reader to check. In this case 
$k^{\sim_{\B_+}} <_\B 1_\B$. Hence, $1_\B \prec_{\B_+} k^{\sim_{\B_+}}$, and so, since $1_\B^{-_{\B_+}} = 1_\B$, we get $k \prec_{\B_+} 1_\B$. Consequently, $1_\B <_\B k = k \cdot_{\B} 1_\B$, which means $1_\B \in k \circ_{\B_+} 1_\B$. It follows that $k \circ i = ((k \circ_{\B_+} 1_\B){\setminus} \{1_\B\}) \cup A_+$, and therefore $j^- = j^{-_{\A_+}} \in k \circ i$. 

In the remaining three cases $j\in B_+{\setminus}\{1_\B\}$. In these cases, we use the fact that $\B_+$ satisfies (v) of Definition~\ref{def:frame} to show that 
$j^- = j^{{-}_{\B_+}}\in k\circ i$. As an example, consider the case with $k \in A_+$ and $i, j \in B_+{\setminus}\{1_\B\}$. Then $1_\B \in i\circ_{\B_+} j$ so that $i\circ j = ((i\circ_{\B_+} j) {\setminus} \{1_\B\}) \cup A_+$. Since $1_\B^{\sim_{\B_+}} = 1_\B$, we have $1^{\sim_{\B_+}}_\B \in i\circ_{\B_+} j$. Applying Definition~\ref{def:frame}(v) to this gives $j^{-_{\B_+}} \in 1_\B \circ_{\B_+} i$. 
Now $j\neq 1_\B$ implies $j^{-_{\B_+}}\neq 1_\B$, so 
$j^- = j^{-_{\B_+}} \in (1_\B \circ_{\B_+} i){\setminus} \{1_\B\} \subseteq k \circ i$.

Conversely, assume $j^- \in k \circ i$. Again, there are eight cases to consider.

If $i, j, k \in A_+$ or $i, j, k \in B_+{\setminus}\{1_\B\}$, the result follows from the fact that $\A_+$ and $\B_+$, respectively, satisfy (v) of Definition~\ref{def:frame}.  If $k\in A_+$,  we show that $A_+\subseteq i\circ j$,
which implies that $k^\sim=k^{\sim_{\A_+}}\in i\circ j$.  In the last three cases we have $k\in B_+{\setminus}\{1_\B\}$. Here we use the fact that $\B_+$ satisfies (v) of Definition~\ref{def:frame} and $k^{\sim_{\B_+}}\neq 1_\B$ to show that 
$k^{\sim} = k^{\sim_{\B_+}}\in i\circ j$. 

Items (vi) and (vii) follow directly from the definition of $^{\neg}$
and the fact that both $\A_+$ and $\B_+$ are frames,
so $^{\neg_{\A_+}}$ and $^{\neg_{\B_+}}$ satisfy these
conditions.

(viii) Let $i,j\in \JA[\JB]$ and
assume $i\preccurlyeq j$.  If $i,j\in A_+$ or $i,j\in B_+{\setminus}\{1_\B\}$, respectively, it follows that $j^\neg\preccurlyeq i^\neg$
since both $\A_+$ and $\B_+$ satisfy Definition~\ref{def:frame}(viii).
If $i\in A_+$ and $j\in B_+{\setminus}\{1_\B\}$, then $i^\neg=i^{\neg_{\A_+}}\in A_+$,
$j^\neg=j^{\neg_{\B_+}}\in B_+{\setminus}\{1_\B\}$ and
$i\preccurlyeq j$ implies $1_\B\prec_{\B_+} j$.  Since $^{\neg_{\B_+}}$
is order-reversing, we have $j^{\neg_{\B_+}}\prec 1^{\neg_{\B_+}}_\B = 1_\B$, and hence 
$j^{\neg} =j^{\neg_{\B_+}}\preccurlyeq i^{\neg_{\A_+}} = i^{\neg}$. The case
with $j\in A_+$ and $i\in B_+{\setminus}\{1_\B\}$ follows similarly.

(ix) Let $i,j,k\in \JA[\JB]$. 
There are eight cases to consider.

First, if $i, j, k \in A_+$, then $k^{-} \in i\circ j$ iff $k^{{-}_{\A_+}}\in i\circ_{\A_+} j$ iff
$k^{\neg_{\A_+}}\in j^{\sim_{\A_+}\neg_{\A_+}}\circ_{\A_+} i^{\sim_{\A_+}\neg_{\A_+}}$ iff $k^{\neg} \in j^{\sim\neg}\circ i^{\sim\neg}$  since $\A_+$ satisfies (ix) of Definition~\ref{def:frame}. 

If $i,j,k \in B_+{\setminus}\{1_\B\}$, then $k^{-_{\B_+}}, k^{\neg_{\B_+}} \neq 1_\B$, and hence $k^{-}=k^{{-}_{\B_+}}\in i\circ j$ iff $k^{{-}_{\B_+}}\in  i\circ_{\B_+} j$ iff $k^{\neg}=k^{\neg_{\B_+}}\in j^{\sim_{\B_+}\neg_{\B_+}}\circ_{\B_+} i^{\sim_{\B_+}\neg_{\B_+}}$ iff $k^{\neg} \in j^{\sim\neg}\circ i^{\sim\neg}$ since $\B_+$ satisfies (ix) of Definition~\ref{def:frame}.

If $i, k \in A_+$ and $j \in B_+{\setminus}\{1_\B\}$, then $k^-= k^{-_{\A_+}} \in i \circ j$ iff $1_\B \in 1_\B \circ_{\B_+} j$ iff $1^{-_{\B_+}} \in 1_\B \circ_{\B_+} j$ iff $1_\B^{\neg_{\B_+}}\in j^{{\sim}_{\B_+}\neg_{\B_+}}\circ_{\B_+} 1_\B^{{\sim}_{\B_+}\neg_{\B_+}}$ iff $1_\B\in j^{{\sim}_{\B_+}\neg_{\B_+}}\circ_{\B_+} 1_\B$ iff $k^{\neg_{\A_+}} \in j^{{\sim}_{\B_+}\neg_{\B_+}}\circ i^{{\sim}_{\A_+}\neg_{\A_+}}$ iff $k^{\neg} \in j^{\sim\neg}\circ i^{\sim\neg}$.

The other two mixed cases involving $k \in A_+$ can be proved similarly. 
In the last three cases we have $k\in \B_+{\setminus}\{1_\B\}$. As an example, we consider the case with $i, j \in A_+$ 
Here we have $k^- = k^{-_{\B_+}} \in i \circ j$ iff $k^{-_{\B_+}} <_\B 1_\B$ iff $1_\B \prec_{\B_+} k^{-_{\B_+}}$ iff $k \prec_{\B_+} 1_\B^{\sim_{\B_+}} = 1_\B$ iff $1_\B = 1_\B^{\neg_{\B_+}} \prec_{\B_+} k^{\neg_{\B_+}}$ iff $k^{\neg_{\B_+}} <_\B 1_\B$ iff $k^{\neg}= k^{\neg_{\B_+}} \in j^{\sim_{\A_+}\neg_{\A_+}}\circ i^{{\sim}_{\A_+}\neg_{\A_+}} = j^{\sim\neg}\circ i^{{\sim}\neg}$.  
\end{proof}
The structure $\mathbf{A}_+[\mathbf{B}_+]=\left(A_+[B_+],I,\preccurlyeq,\circ,^{\sim},^{-},^{\neg}\right)$ is called the {\em wreath product} of $\mathbf{A}_+$ and $\mathbf{B}_+$.

\begin{example}\label{ex:wreath_product_example_1}
Consider the  dual frames $\A_+$ and $\B_+$ of the DqRAs $\A$ and $\B$  in Figure~\ref{fig:small_examples_DqRAs}. 
Both $\A_+$ and $\B_+$ are depicted 
in Figure~\ref{fig:dual-frames}.
Since $1_\A \in A_+$, $1_\B \in B_+$ and $1_\B^{-_{\B_+}} = 1^{\sim_{\B_+}}_\B = 1^{\neg_{\B_+}}_\B = 1_\B$, it follows from Theorem~\ref{thm:wreat_product_DqRA-frames} that the wreath product $\A_+[\B_+]$ is a frame. The wreath product of $\A_+$ and $\B_+$ is shown in Figure~\ref{fig:wreath_product_example_1}. As before, in the tables for $\circ$, a sequence of elements such as $1_\A\top_\A$ denotes the set $\{1_\A, \top_\A\}$. Since $I_{\B_+}{\setminus}\{1_\B\} = \varnothing$, we have $I = \{1_\A\}$.  
%


\begin{figure}
\centering
\begin{tikzpicture}[scale=0.8]
\footnotesize
\node[draw,circle,inner sep=1.75pt] (top) at (0,0) {};
\node[draw,circle,inner sep=1.75pt] (id) at (0,1) {};
\node[draw,circle,inner sep=1.75pt] (0) at (1.2,0.5) {};
\draw[order] (top)--(id);
\node[label,anchor=east,xshift=1pt] at (id) {$1_\A$};
\node[label,anchor=east,xshift=1pt] at (top) {$\top_\A$};
\node[label,anchor=west,xshift=-1pt] at (0) {$0_\B$};
\node[label,anchor=north,yshift=-10pt] at (top) {$\A_+[\B_+]$};
\node[anchor=west] at (2.7,0.2){
$\begin{array}{|c|ccc||c|c|}\hline
\circ & 0_\B\ &1_\A& \top_\A &^\sim&^\neg\\\hline
0_\B & 1_\A\top_\A & 0_\B & 0_\B & 0_\B & 0_\B	\\
1_\A & 0_\B & 1_\A & 1_\A\top_\A  & \top_\A & \top_\A	\\
\top_\A & 0_\B  & 1_\A\top_\A & 1_\A\top_\A & 1_\A & 1_\A \\\hline
\end{array}$
};
\normalsize
\end{tikzpicture}
\caption{The wreath product $\A_{+}[\B_+]$ of $\A_+$ and $\B_+$ in Figure~\ref{fig:dual-frames}.}\label{fig:wreath_product_example_1}
\end{figure}
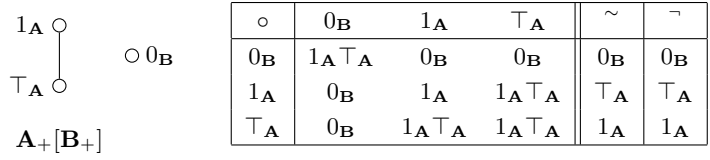

\end{example}

\begin{example}\label{ex:wreath_product_example_2}
Consider the dual frames $\C_+$ 
and $\D_+$ of the DqRAs $\C$ and $\D$ in Figure~\ref{fig:small_examples_DqRAs}. Both $\C_+$ and $\D_+$ are also shown in Figure~\ref{fig:dual-frames}. 
Since $1_\C \in C_+$, $1_\D\in D_+$ and $1_\D^{-_{\D_+}} = 1_\D^{\sim_{\D_+}} = 1_\D^{\neg_{\D_+}} = 1_\D$, we can apply Theorem~\ref{thm:wreat_product_DqRA-frames} to conclude that $\C_+[\D_+]$ is a frame. In this case $I = \{1_\C, 0_\D\}$. See Figure~\ref{fig:wreath_product_example_2} for more details. 

\begin{figure}
\centering
\begin{tikzpicture}[scale=0.8]
\footnotesize
\node[draw,circle,inner sep=2pt] (top) at (1,0) {};
\node[draw,circle,inner sep=2pt] (0) at (1,2) {};
\node[draw,circle,inner sep=2pt] (1) at (0,1) {};
\node[draw,circle,inner sep=2pt] (a) at (2,1) {};
\draw[order] (top)--(1);
\draw[order] (top)--(a);
\draw[order] (1)--(0);
\draw[order] (a)--(0);
\node[label,anchor=north] at (top) {$\top_\D$};
\node[label,anchor=south,xshift=1pt] at (0) {$0_\D$};
\node[label,anchor=east,xshift=1pt] at (1) {$1_\C$};
\node[label,anchor=west,xshift=-1pt] at (a) {$a$};
\node[label,anchor=north,yshift=-20pt] at (top) {$\C_+[\D_+]$};
\node[anchor=west] at (3.3,0.65){
$\begin{array}{|c|cccc||c|c|}\hline
\circ & 0_\D\ &1_\C& a & \top_\D &^{\sim}& ^{\neg}\\\hline
0_\D & 0_\D	& 0_\D & 0_\D & 0_\D1_\C a\top_\D & \top_\D &\top_\D	\\
1_\C& 0_\D & 0_\D1_\C & 0_\D a & 0_\D1_\C a\top_\D& a& a	\\
a & 0_\D & 0_\D a & 0_\D 1_\C & 0_\D1_\C a\top_\D & 1_\C & 1_\C \\
\top_\D & 0_\D1_\C a\top_\D &0_\D1_\C a\top_\D & 0_\D1_\C a\top_\D &  0_\D1_\C a\top_\D &0_\D&0_\D\\\hline
\end{array}$
};
\normalsize
\end{tikzpicture}
\caption{The wreath product $\mathbf{C}_{+}[\mathbf{D}_+]$ of $\mathbf C_+$ and $\mathbf D_+$ in Figure~\ref{fig:dual-frames}.}\label{fig:wreath_product_example_2}
\end{figure}
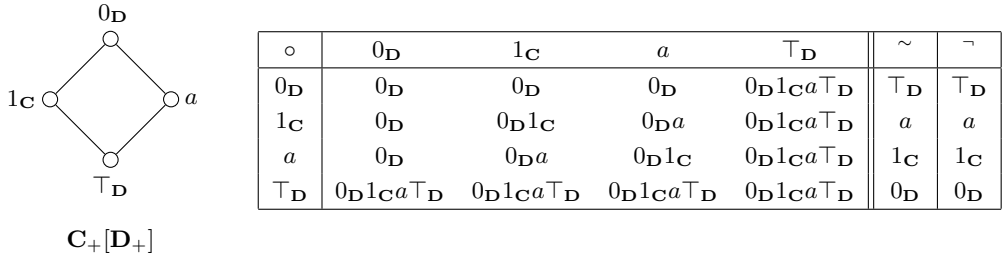

\end{example}

\begin{remark}
In \cite{CMR2026}, we introduced a nested sum construction $\mathbf K[\mathbf L]$ as a method for constructing a new qRA from two qRAs $\mathbf K$ and $\mathbf L$. The construction is subject to several restrictions on the choice of $\mathbf{K}$, which are described in detail in that paper. In  particular, $\mathbf{K}$ is required to be odd (i.e., $-1={\sim}1 = \neg 1 = 1$) and $1$ must be a join-irreducible of $\mathbf{K}$.
Now let $\A$ be a DqRA whose identity element $1$ is a join-irreducible.
If $\mathbf{S}_{n}$ is an even Sugihara chain, then its dual frame satisfies the requirements outlined above, and so $\A_+[(\mathbf{S}_n)_+]$ is a frame by Theorem~\ref{thm:wreat_product_DqRA-frames}. It is also straightforward to check that $\mathbf{S}_{n-1}$ satisfies the requirements imposed on $\mathbf{K}$ in \cite{CMR2026}, so $S_{n-1}[\mathbf A]$ is a DqRA. Moreover, the reader can verify that $\mathbf S_{n-1}[\mathbf A]$ is isomorphic to the complex algebra of $\A_+[(\mathbf{S}_n)_+]$. 

\end{remark}

\section{Frame networks and representability}\label{sec:networks_representations}

Networks provide the underlying structures on which representation games for relation algebras~\cite{HH2002book} and weakening relation algebras~\cite{JS23} are played. In these games, players successively extend labelled networks while preserving consistency conditions that reflect the algebraic operations, and they  
are used to characterise representability for both relation algebras and weakening relation algebras in terms of a winning strategy for the builder player.
In this section, we define a network for a frame and establish a correspondence between the representability of a finite DqRA and the existence of consistent saturated frame networks for its dual frame. This will enable us to show that the complex algebra of the wreath product of the dual frames of two finite normally representable DqRAs satisfying additional conditions is itself representable.

\begin{definition}\label{def:network_for_frame}
A \emph{network} \emph{for a frame} $\mathbf{F} = \left(F, I,\preccurlyeq_F, \circ, {^{\sim}}, {^{-}}, ^\neg\right)$ is a tuple $\mathcal{N} = (N, \lambda, \alpha, \beta)$ where $N$ is a set of nodes, $\lambda: N^2 \to F$ is a \emph{labelling function}, and $\alpha, \beta: N \rightharpoonup N$ are partial functions \footnote{We require $\alpha$ and $\beta$ to be partial functions so that, in a representation game, the builder player can extend these maps incrementally as the game progresses. This allows her to determine their values in response to the moves made during each round, rather than having to specify their values in advance at the outset of the game. However, throughout this paper, $\alpha$ and $\beta$ will always be total functions.}.
\end{definition}

\begin{definition}\label{def:consistent_network_for_frame}
A network $\mathcal{N} = (N, \lambda, \alpha, \beta)$ for a frame $\mathbf{F} = \left(F, I,\preccurlyeq_F, \circ, {^{\sim}}, {^{-}}, ^\neg\right)$ is said to be  \emph{consistent} if the following conditions hold for all $x, y, z \in N$:
\begin{enumerate}[label=\textup{(\roman*)}]
\item $\lambda(x, x) \in I$ 
\item $\lambda(x, z) \in \lambda(x,y) \circ \lambda(y, z)$
\item If $\lambda(x,y) \in I$ and $\lambda(y, x) \in I$, then $x = y$.
\item If $\alpha(x)$ and $\alpha(y)$ are defined, then  $\lambda(x, y) \in I$ iff $\lambda(\alpha(x), \alpha(y)) \in I$. 
\item If $\beta(x)$ and $\beta(y)$ are defined, then $\lambda(x, y) \in I$ iff $\lambda(\beta(y), \beta(x)) \in I$.
\item If $\beta(\beta(x))$ is defined, then $\beta(\beta(x))=x$.
\item If $\alpha( \beta(\alpha(x)))$ is defined, then $\alpha(\beta(\alpha(x))) = \beta(x)$.
\item If $\alpha(x)$ is defined, then $\lambda(y, \alpha(x)) = \lambda(x, y)^{\sim}$.
\item If $\beta(\alpha(x))$ and $\beta(y)$ are defined, then  $\lambda(\beta(\alpha(x)), \beta(y)) = \lambda(x, y)^{\neg}$.
\end{enumerate}
\end{definition}

We need to define the following refinement of $\circ$ which will be used to define a saturated network. Define an operation $\cmax : F \times F \to \mathcal{P}(F)$ by saying that $f \in g \cmax h$ if, and only if,

\begin{enumerate}
\item $f \in g \circ h$, and
\item for all $g', h' \in F$, if $g \preccurlyeq g'$, $h\preccurlyeq h'$ and $f \in g' \circ h'$, then $g'=g$ and $h'=h$.
\end{enumerate}

A network for a frame $\mathbf{F}$ is said to be \emph{saturated} if $\alpha$ and $\beta$ are both surjective total functions, and if $\lambda(x, y) \in f \cmax g$, then there exists some $z \in N$ such that $\lambda(x, z) = f$ and $\lambda(z, y) = g$. 

We say that $f \in F$ is \emph{network-realised} if there exists a consistent saturated network $\mathcal{N}^f = (N^f, \lambda^f, \alpha^f, \beta^{f})$ for $\mathbf{F}$ and nodes $u, v \in N^f$ such that $\lambda^f(u, v) = f$. We say that a frame $\mathbf F$ has the \emph{network realisation property} if every $f \in F$ is network realised.

In what follows, we will show that if a finite frame $\mathbf F$ has the network realisation property, then its complex algebra is representable. 
To prove this, we construct a full DqRA for each frame element $f \in F$ from its corresponding consistent saturated network $\mathcal{N}^f$, and then show that the complex algebra of the frame embeds into the product of all the full DqRAs constructed in this way, thereby establishing that it is isomorphic to a subalgebra of a product of full DqRAs.

Fix a finite frame $\mathbf{F} = \left(F, I,\preccurlyeq_F, \circ, {^{\sim}}, {^{-}}, ^\neg\right)$ and an $f \in F$. Let $\mathcal{N}^f = (N^f, \lambda^f, \alpha^f, \beta^f)$ be a consistent saturated network for $\mathbf{F}$ such that $\lambda^f(u, v) = f$ for some $(u, v) \in N^f\times N^f$. 

Define ${\leqslant_{N^f}} \subseteq N^f \times N^f$ as follows:
$$
x \leqslant_{N^f} y \quad \textnormal{ iff } \quad \lambda^{f}(x, y) \in I
$$
Then $(N^f, \leqslant_{N^f})$ is a poset. It is straightforward to check that $\leqslant_{N^f}$ is reflexive and antisymmetric.  To prove that it is transitive, we use Lemma~\ref{lem:IcircI=I}.

\begin{lemma}\label{lem:alpha_beta}
\begin{enumerate}[label=\textup{(\roman*)}]
\item The map $\alpha^f: N^f \to N^f$ is an order automorphism  of $(N^f, \leqslant_{N^f})$. 
\item The map $\beta^{f}: N^f \to N^f$ is a  self-inverse dual order automorphism of $(N^f, \leqslant_{N^f})$. 
\item $\alpha^f\mathbin{;}\beta^{f}\mathbin{;}\alpha^{f} = \beta^{f}$
\end{enumerate}
\end{lemma}

\begin{proof}
(i) Firstly, since $\mathcal{N}^f$ is saturated,  $\alpha^f$ is surjective. 
To see that $\alpha^{f}$ is an order automorphism of $\left(N^f, \leqslant_{N^f}\right)$, let $x, y\in N^f$. Since $\alpha^f$ is total, both $\alpha(x)$ and $\alpha(y)$ are defined and we can apply (iv) of Definition~\ref{def:consistent_network_for_frame}: 
$x \leqslant_{N^f} y \textnormal{ iff }  \lambda^f(x, y) \in I \textnormal{ iff } \lambda^f(\alpha^f(x), \alpha^f(y)) \in I 
\textnormal{ iff } \alpha^f(x)\leqslant_{N^f} \alpha^f(y)$.

(ii) The map $\beta^f$ is surjective since $\mathcal{N}^f$ is saturated. It follows from  Definition~\ref{def:consistent_network_for_frame}(vi) that $\beta^f$ is self-inverse. Using  
(v), we get:
$x \leqslant_{N^f} y \textnormal{ iff }  \lambda^f(x, y) \in I
\textnormal{ iff } \lambda^f(\beta^f(y), \beta^f(x)) \in I
\textnormal{ iff } 
\beta^f(y)\leqslant_{N^f} \beta^f(x)$. This shows that $\beta^f: N^f \to N^f$ is a dual order automorphism.  

(iii) This follows from (vii) of Definition~\ref{def:consistent_network_for_frame}.
\end{proof}

Recall that $N^f \times N^f$ can be partially ordered as follows: for $(u, v), (x, y) \in N^f \times N^f$,
\[
(x, y) \preccurlyeq_{N^f} (u, v) \quad \textnormal{ iff } \quad u \leqslant_{N^f} x \textnormal{ and } y \leqslant_{N^f} v.
\]
Since $\alpha^f$ and $\beta^f$ satisfy the conditions of Theorem~\ref{Theorem:Dq(E)} by Lemma~\ref{lem:alpha_beta}, we obtain the following result. 

\begin{theorem}\label{thm:new_DqRA}
The algebra $\mathbf{D}(N^f\times N^f, \preccurlyeq_{N^f}) =\left\langle \mathsf{Up}\left(N^f \times N^f, \preccurlyeq_{N^f}\right), \cap, \cup, \mathbin{;}, {\sim}, -, \neg, \leqslant_{N^f}\right\rangle$ is a full DqRA. 
\end{theorem}

Now define $\psi^f: \mathsf{Up}\left(F, \preccurlyeq_F\right) \to \mathcal{P}\left(N^f\times N^f\right)$ by setting, for all $U \in \mathsf{Up}\left(F, \preccurlyeq_F\right)$,
\[
\psi^f(U) := \{(x, y) \mid \lambda^f(x, y) \in U\}.
\]

\begin{lemma}\label{lem:psi(U)_upset}
For all $U \in \mathsf{Up}\left(F, \preccurlyeq_F\right)$, we have $\psi^f(U) \in \mathsf{Up}\left(N^f \times N^f, \preccurlyeq_{N^f}\right)$. 
\end{lemma}

\begin{proof}
Let $(x, y)\in \psi_f(U)$ and assume $(x, y) \preccurlyeq_{N^f} (u, v)$. Since $(x, y)\in \psi^f(U)$, we get $\lambda_f(x, y) \in U$. We also have $u \leqslant_{N^f} x$ and $y\leqslant_{N^f} v$, and so $\lambda^f(u,x) \in I$ and $\lambda^f(y, v) \in I$. 
Using Definition~\ref{def:consistent_network_for_frame}(ii), 
we get $\lambda^f(u, y) \in \lambda^f(u, x) \circ \lambda^f(x, y)$. Consequently, $\lambda^f(u, y) \in I \circ U = U$. 
Similarly we have $\lambda^f(u, v) \in \lambda^f(u, y) \circ \lambda^f(y, v)$, which shows $\lambda^f(u, v) \in U \circ I = U$. It follows that $(u, v) \in \psi^f(U)$, as required. 
\end{proof}



In the next lemma we will show that $\psi^f$ is a homomorphism from the complex algebra $\mathbf{F}^+$ of $\mathbf{F}$ to the distributive quasi relation algebra $\mathbf{D}\left(N^f\times N^f, \preccurlyeq_{N^f}\right)$. We will need the fact that $\alpha^f$ has an inverse $(\alpha^f)^{-1}$. This follows from Lemma~\ref{lem:alpha_beta}(i). 


\begin{lemma}\label{lem:psi_f_homomorphism}
The map $\psi^f: \mathsf{Up}\left(F, \preccurlyeq_F\right) \to \mathsf{Up}\left(N^f\times N^f, \preccurlyeq_{N^f}\right)$ is a homomorphism from $\mathbf{F}^+$ to $\mathbf{D}(N^f\times N^f, \preccurlyeq_{N^f})$. 
\end{lemma}

\begin{proof}
Firstly, the fact that $\psi^f(U\cap V) = \psi^f(U)\cap\psi^f(V)$ and $\psi^f(U\cup V) = \psi^f(U)\cup\psi^f(V)$ for all $U, V\in \mathsf{Up}\left(F, \preccurlyeq_F\right)$
follows from basic set theory. 
%
It is also immediate from the definitions of $\psi^f$ and $\leqslant_{N^f}$ that $\psi^f(I) = {\leqslant_{N^f}}$. To see that $\psi^f$  preserves 
the monoid operation, let $(x, y) \in \psi^f(U)\mathbin{;} \psi^f(V)$. Then there exists some $z \in N^f$ such that $(x, z) \in \psi^f(U)$ and $(z, y) \in \psi^f(V)$. Hence,  $\lambda^f(x, z) \in U$ and $\lambda^f(z, y) \in V$. Applying (ii) of Definition~\ref{def:consistent_network_for_frame}, we get $\lambda^f(x,y) \in \lambda^f(x, z) \circ \lambda^f(z, y)$. It follows that $\lambda^f(x, y) \in U \circ V$.  We can thus conclude that $(x, y) \in \psi^f(U \circ V)$. 

For the other inclusion, let $(x, y) \in \psi^f(U\circ V)$. Then $\lambda^f(x, y) \in U \circ V$, and so there exist $g \in U$  and $h \in V$ such that $\lambda^f(x, y) \in g \circ h$. Since $F$ is finite, there exist $g', h' \in F$ such that $g \preccurlyeq g'$, $h \preccurlyeq h'$ and $\lambda^{f}(x, y) \in g' \cmax h'$. Since $\mathcal{N}^f$ is saturated, there exists some $z \in N^f$ such that $\lambda^f(x, z) = g'$ and $\lambda^f(z, y) = h'$.  Using the fact that $U$ and $V$ are upsets of $(F, \preccurlyeq_F)$, we can conclude that $\lambda^f(x, z) \in U$ and $\lambda^f(z, y) \in V$. Consequently, $(x, z) \in \psi^f(U)$ and $(z, y) \in \psi^{f}(V)$, and thus $(x, y) \in \psi^f(U)\mathbin{;}\psi^f(V)$. 



Next we show that $\psi^f$ preserves the linear negations. We will show that $\psi^f$ preserves $\sim$. The proof that it preserves $-$ follows
similarly.
%
Let $x, y \in N^f$. Then, using Definition~\ref{def:consistent_network_for_frame}(viii), we obtain
\begin{align*}
(x, y) \in \psi^f({\sim}U) 
& \textnormal{ iff } \lambda^f(x, y) \in {\sim}U
\textnormal{ iff } \lambda^f(x, y)^{-} \notin U 
\textnormal{ iff }  \lambda^f(x, \alpha^f((\alpha^f)^{-1}(y)))^{-} \notin U \\
&\textnormal{ iff }  \lambda^f((\alpha^f)^{-1}(y), x)^{\sim -} \notin U 
\textnormal{ iff }  \lambda^f((\alpha^f)^{-1}(y), x) \notin U 
\textnormal{ iff } ((\alpha^f)^{-1}(y), x) \notin \psi_f(U)\\
&\textnormal{ iff } (x, (\alpha^f)^{-1}(y)) \in \psi^f(U)^{c\smallsmile}
\textnormal{ iff } (x, y) \in \psi^f(U)^{c\smallsmile}\mathbin{;} \alpha^{f} =
{\sim}\psi^f(U).
\end{align*}


Finally we show that $\psi^f$ preserves the involution:
%
%
\begin{align*}
(x, y) \in \psi^f(\neg U)
& \textnormal{ iff }  \lambda^f(x, y) \in \neg U
 \textnormal{ iff } \lambda^f(x, y)^{\neg} \notin U
\textnormal{ iff } \lambda^f(\beta^f(\alpha^f(x)), \beta^f(y)) \notin U \\
&\textnormal{ iff } (\beta^f(\alpha^f(x)), \beta^f(y)) \notin \psi^f(U) 
\textnormal{ iff } (\beta^f(\alpha^f(x)), \beta^f(y)) \in \psi^f(U)^c\\
& \textnormal{ iff } (x,y) \in \alpha^{f}\mathbin{;} \beta^{f} \mathbin{;} \psi^f(U)^c\mathbin{;} \beta^{f}= \neg\psi^f(U).
\end{align*}
Here the third equivalence follows from Definition~\ref{def:consistent_network_for_frame}(ix). 
\end{proof}

We are now ready to prove one of the main results of this section.

\begin{proposition}\label{prop:consistent_sat_network_imply_representable} 
Let $\mathbf{F} = \left(F, I,\preccurlyeq_F, \circ, {^{\sim}}, {^{-}}, {^{\neg}}\right)$ be a \emph{finite} frame. If $\mathbf F$ has the network realisation property, then the complex algebra $\mathbf{F}^+$ of $\mathbf{F}$ is representable.
\end{proposition}

\begin{proof}
Since $\mathbf F$ has the network realisation property, for every $f \in F$, there is a consistent saturated network $\mathcal{N}^f = (N^f, \lambda^f, \alpha^f, \beta^{f})$ and nodes $u, v \in N^f$ such that $\lambda^f(u, v) = f$. Let $f \in F$ and define ${\leqslant_{N^f}}$ by 
$x \leqslant_{N^f} y$ iff $\lambda^{f}(x, y) \in I$.
Then $(N^f, \leqslant_{N^f})$ is a poset. It also follows from Lemma~\ref{lem:alpha_beta} that $\alpha^f$ is an order automorphism of $(N^f, \leqslant_{N^f})$, $\beta^f$ is a dual order-automorphism of $(N^f, \leqslant_{N^f})$, and $\alpha^f\mathbin{;} \beta^f\mathbin{;}\alpha^f = \beta^f$. For $(u, v), (x, y) \in N^f \times N^f$,
define $\preccurlyeq_{N^f}$ by 
$(x, y) \preccurlyeq_{N^f} (u, v)$ iff $u \leqslant_{N^f} x$ and $y \leqslant_{N^f} v$. Hence, by Theorem~\ref{thm:new_DqRA}, the algebra $\mathbf{D}(N^f\times N^f, \preccurlyeq_{N^f}) =\left\langle \mathsf{Up}\left(N^f \times N^f, \preccurlyeq_{N^f}\right), \cap, \cup, \mathbin{;}, {\sim}, -, \neg, \leqslant_{N^f}\right\rangle$ is a full DqRA. Moreover, by Lemma~\ref{lem:psi_f_homomorphism}, the map $\psi^f: \mathsf{Up}\left(F, \preccurlyeq_F\right) \to \mathsf{Up}(N^f\times N^f, \preccurlyeq_{N^f})$ defined by $$\psi^f(U) := \{(x, y) \mid \lambda^f(x, y) \in U\},$$ for all $U \in \mathsf{Up}\left(F, \preccurlyeq_F\right)$, is a homomorphism from $\mathbf{F}^+$ to $\mathbf{D}(N^f\times N^f, \preccurlyeq_{N^f})$. 
Now define \break $\psi: \mathsf{Up}(F, \preccurlyeq_F) \to \prod_{f \in F}\mathsf{Up}(N^f\times N^f, \preccurlyeq_{N^f})$ by setting, for all $U \in \mathsf{Up}(F, \preccurlyeq_F)$, 
$
\psi(U)(f) = \psi^f(U). 
$
Then $\psi$ is a homomorphism from $\mathbf{F}^+$ to $\prod_{f\in F}\mathbf{D}(N^f\times N^f, \preccurlyeq_{N^f})$, so all that remains to check is that $\psi$ is injective. Let $U, V \in \mathsf{Up}(F, \preccurlyeq_F)$ such that $U \neq V$. Then $U \not\subseteq V$ or $V \not\subseteq U$. Assume the first. Then there exists some $g \in F$ such that $g \in U$ while $g \notin V$. By assumption there exists a network $\mathcal{N}^g$ for $\mathbf{F}$ and $s, t \in N^g$ such that $\lambda^g(s, t) = g$. Hence,  $\lambda^g(s, t) \in U$ while $\lambda^g(s, t) \notin V$, which means $(s, t) \in \psi^g(U)$ while $(s, t) \notin \psi^g(V)$. This shows $\psi(U)(g) = \psi^g(U) \neq \psi^g(V) = \psi(V)(g)$. The other case is similar. 
\end{proof}

In the remainder of this section, we show that the dual frame of a finite normally representable DqRA has the network realisation property. 

Let $\A=\langle A,\wedge,\vee, \cdot, {\sim}, -, \neg, 1\rangle$ be a \emph{finite} normally representable DqRA. Then 
there exists a poset $\mathbf X = \left(X, \leqslant_X\right)$,
an equivalence relation $E\subseteq X\times X$ with
${\leqslant_X} \subseteq E$, an
order automorphism $\alpha:X\to X$
and a self-inverse dual order automorphism $\beta:X\to X$ such that 
$\alpha,\beta\subseteq E$ and
$\alpha\mathbin{;} \beta \mathbin{;} \alpha = \beta$.  Moreover, there exists an embedding
$\varphi: \mathbf{A} \hookrightarrow \mathbf{D}\left(E, \preccurlyeq_{E}\right)$. Since $\mathbf A$ is normally representable, it must be the case that $\varphi(\top) = E$ and $\varphi(\bot) = \varnothing$. 

For each $(x, y) \in E$, let $F_{(x, y)} = \{d \in A \mid (x, y) \in \varphi(d)\}$. In the next lemma, we show that this set is a prime filter of the lattice reduct of $\A$, and in the following lemma that the meet of the elements in $F_{(x,y)}$ is a join-irreducible. 
Note that when we write $\leqslant$ we mean the lattice order on $A$.  The partial order on $X$ is denoted by $\leqslant_X$. 

\begin{lemma}\label{lem:prime_filter_and_uparrow=F}
For each $(x, y) \in E$, the set $F_{(x, y)}$ is a prime filter of the lattice reduct of $\mathbf{A}$, and moreover ${\uparrow}\bigwedge F_{(x, y)} = F_{(x, y)}$.
\end{lemma}

\begin{proof}
First, $(x, y) \in E = \varphi(\top)$, so $\top \in F_{(x, y)}$, which means $F_{(x, y)} \neq \varnothing$. Also, $(x, y) \notin \varnothing = \varphi(\bot)$, so $F_{(x, y)} \neq A$. 

To show that $F_{(x, y)}$ is upward closed, let $c\in F_{(x, y)}$ and assume $c \leqslant d$. The first part implies $(x, y) \in \varphi(c)$, and the second part implies $\varphi(c) \subseteq \varphi(d)$. Hence, $(x, y) \in \varphi(d)$, and so $d \in F_{(x, y)}$. 

Now let $c, d \in F_{(x, y)}$. Then we have $(x, y) \in \varphi(c)$ and $(x, y) \in \varphi(d)$. Therefore, $(x, y) \in \varphi(c) \cap \varphi(d) = \varphi(c \wedge d)$. This shows $c \wedge d \in F_{(x, y)}$. 

Finally, let $c \vee d \in  F_{(x, y)}$. Then $(x, y) \in \varphi(c \vee d) = \varphi(c) \cup \varphi(d)$, and so $(x, y) \in \varphi(c)$ or $(x, y) \in \varphi(d)$. We thus get $c \in F_{(x, y)}$ or $d \in F_{(x, y)}$. 



To see that ${\uparrow}\bigwedge F_{(x, y)} = F_{(x, y)}$, let $d \in F_{(x, y)}$. Then $\bigwedge F_{(x, y)} \leqslant d$, and so $d \in {\uparrow}\bigwedge F_{(x, y)}$. This shows that $F_{(x, y)} \subseteq {\uparrow}\bigwedge F_{(x, y)}$.  

For the other inclusion, let $d \in {\uparrow}\bigwedge F_{(x, y)}$. Then $\bigwedge F_{(x, y)} \leqslant d$. Since $\mathbf{A}$ is finite and $\varphi$ preserves finite meets, we have $(x, y) \in \varphi\left(\bigwedge F_{(x, y)}\right)$. Hence, $\bigwedge F_{(x, y)} \in F_{(x, y)}$, and so, since $F_{(x, y)}$ is upward closed, it follows that $d \in F_{(x, y)}$. 
\end{proof}

\begin{lemma}\label{lem:join_irreducible}
For each $(x, y) \in E$, we have $\bigwedge F_{(x, y)} \in A_+$.  
\end{lemma}

\begin{proof}
Since $\A$ is finite and $\langle A, \wedge, \vee\rangle$ is distributive, an element of $A$ is join-irreducible iff it is join-prime. We will thus show that $\bigwedge F_{(x, y)}$ is join-prime. Let $c, d \in A$ and assume $\bigwedge F_{(x, y)} \leqslant c \vee d$. Then we have $c \vee d \in {\uparrow}\bigwedge F_{(x, y)}$, and thus $c \vee d \in F_{(x, y)}$ by Lemma~\ref{lem:prime_filter_and_uparrow=F}. Since $F_{(x, y)}$ is a prime filter by Lemma~\ref{lem:prime_filter_and_uparrow=F}, we get $c \in F_{(x, y)}$ or $d \in F_{(x, y)}$. Therefore, $\bigwedge F_{(x, y)} \leqslant c$ or $\bigwedge F_{(x, y)} \leqslant d$. 
\end{proof}

Now define $\lambda: E \to A_+$ by setting, for each $(x, y) \in E$, 
$$
\lambda(x, y) := \bigwedge \{d \in A \mid (x, y) \in \varphi(d)\}.
$$
The next lemma guarantees that the function $\lambda$ is surjective. 

\begin{lemma}\label{lem:lambda_surjective}
For each join-irreducible $i \in A_+$, there exists some $(s, t) \in E$ such that $\lambda(s, t) = i$.  
\end{lemma}

\begin{proof}
Let $i \in A_+$. Since $i \not\leqslant \kappa(i)$ and $\varphi$ is injective, we have $\varphi(i)\not\subseteq \varphi(\kappa(i))$. Hence, there exists  some $(s, t) \in E$ such that $(s, t) \in \varphi(i)$ while $(s, t) \notin \varphi(\kappa(i))$. Since $(s, t) \in \varphi(i)$, we have $i \in F_{(s, t)}$, and so $\lambda(s, t) = \bigwedge F_{(s, t)} \leqslant i$. 

To see that $i \leqslant \bigwedge F_{(s, t)}$, suppose for the sake of a contradiction that this is not the case. Then $\bigwedge F_{(s, t)} \in \{d \in A \mid i \not\leqslant d\}$, and so $\bigwedge F_{(s, t)} \leqslant \kappa(i)$. This shows that $\kappa(i) \in {\uparrow}\bigwedge F_{(s, t)}$, which implies $\kappa(i) \in F_{(s, t)}$ by Lemma~\ref{lem:prime_filter_and_uparrow=F}. Hence, $(s, t) \in \varphi(\kappa(i))$, a contradiction. We can thus conclude that $\lambda(s, t) = i$. 
\end{proof}

The following lemma will also be useful. 

\begin{lemma}\label{lem:interaction_lambda_varphi}
For each $a \in A$, we have $(x, y) \in \varphi(a)$ iff $\lambda(x, y) \leqslant a$.
\end{lemma}

\begin{proof}
Let $a \in A$ and $(x, y) \in \varphi(a)$. Then $a \in F_{(x, y)}$, and so $\lambda(x, y) = \bigwedge F_{(x, y)} \leqslant a$. Conversely, assume $\lambda(x, y) \leqslant a$. Then $\varphi(\lambda(x, y)) \subseteq \varphi(a)$. Now $\mathbf{A}$ is finite and $\varphi$ preserves finite meets, so $(x, y) \in \varphi\left(\bigwedge F_{(x, y)}\right)$. Hence, $(x, y) \in \varphi(a)$. 
\end{proof}

We are now ready to prove that the dual frame of a finite normally representable DqRA has the network realisation property. In what follows, we will denote the equivalence class of $x \in X$ with respect to $E$ by $[x]$. 

\begin{proposition}\label{prop:SatConNetworkExists}
Let $\A=\langle A,\wedge,\vee, \cdot, {\sim}, -, \neg, 1\rangle$ be a \emph{finite} normally representable DqRA. Then $\A_+ = (A_+, I, \preccurlyeq, \circ, ^{\sim}, ^{-}, ^{\neg})$ has the network realisation property. 
\end{proposition}

\begin{proof}
Let $i$ be an arbitrary join-irreducible of $\A$. By Lemma~\ref{lem:lambda_surjective}, there is some $(s, t) \in E$ such that $\lambda(s, t) = i$. Notice that $\lambda\vert_{[t]\times [t]}(s, t) = i$, so we show that $([t], \lambda\vert_{[t]\times [t]}, \alpha\vert_{[t]}, \beta\vert_{[t]})$ is a consistent saturated network for $\mathbf{A}_+$. 

To prove (i) of Definition~\ref{def:consistent_network_for_frame}, let $u \in [t]$. We have $(u, u) \in {\leqslant_X} = \varphi(1)$, so $1 \in F_{(u, u)}$. Hence, $\lambda\vert_{[t]\times [t]}(u, u) =\lambda(u, u) = \bigwedge F_{(u, u)} \leqslant 1$, which means $\lambda\vert_{[t]\times [t]}(u, u) \in \{k \in A_+ \mid k \leqslant 1\}=I$.

For item (ii), let $u, v, w \in [t]$. Since $\mathbf{A}$ is finite and $\varphi$ preserves finite meets, it follows that $(u, v) \in \varphi\left(\bigwedge F_{(u, v)}\right)$ and $(v, w) \in \varphi\left(\bigwedge F_{(v, w)}\right)$. Hence, $(u, w) \in \varphi\left(\bigwedge F_{(u, v)}\right) \mathbin{;} \varphi\left(\bigwedge F_{(v, w)}\right)$, and so $\bigwedge F_{(u, v)} \cdot \bigwedge F_{(v, w)} \in F_{(u, w)}$. This shows that 
\begin{align*}
\lambda\vert_{[t]\times [t]}(u, w) = \lambda(u, w)  = \bigwedge F_{(u, w)} & \leqslant \bigwedge F_{(u, v)} \cdot \bigwedge F_{(v, w)}\\
& = \lambda(u, v) \cdot \lambda(v, w) \\
& = \lambda\vert_{[t]\times [t]}(u, v) \cdot \lambda\vert_{[t]\times [t]}(v, w).
\end{align*}
We can thus conclude that $\lambda\vert_{[t]\times [t]}(u, w) \in \lambda\vert_{[t]\times [t]}(u, v) \circ \lambda\vert_{[t]\times [t]}(v,w)$, as required.

Next we prove (iii).  Let $u, v \in [t]$ and assume $\lambda\vert_{[t] \times [t]}(u, v) \in I$ and $\lambda\vert_{[t]\times [t]}(v, u) \in I$. Then $\lambda(u, v) = \bigwedge F_{(u, v)} \leqslant 1$ and $\lambda(v, u) = \bigwedge F_{(v, u)} \leqslant 1$, and therefore, by Lemma~\ref{lem:prime_filter_and_uparrow=F}, $1 \in F_{(u, v)}$ and $1 \in F_{(v, u)}$. This shows that $(u, v) \in \varphi(1) = {\leqslant_X}$ and $(v, u) \in \varphi(1) = {\leqslant_X}$. Hence, since $\leqslant_X$ is antisymmetric, we get $u= v$.

For (iv), let $u, v \in [t]$. Using the definition of $I$ and Lemma~\ref{lem:prime_filter_and_uparrow=F}, we obtain the following equivalences:
\begin{align*}
\lambda\vert_{[t]\times [t]}(u, v) = \lambda(u, v) \in I
 \textnormal{ iff } \bigwedge F_{(u, v)} \leqslant 1 
 \textnormal{ iff } 1 \in {\uparrow}\bigwedge F_{(u, v)}
 \textnormal{ iff } 1 \in F_{(u, v)} \textnormal{ iff } (u, v) \in \varphi(1).
\end{align*}
Since $\varphi(1) = {\leqslant_X}$ and $\alpha$ is an order automorphism of $(X, \leqslant_X)$, it follows that 
\begin{align*}
\lambda\vert_{[t]\times [t]}(u, v)\in I
 \textnormal{ iff } (u, v) \in {\leqslant_X}
\textnormal{ iff } (\alpha(u), \alpha(v)) \in  {\leqslant_X} = \varphi(1)   \textnormal{ iff } 1 \in F_{(\alpha(u), \alpha(v))}.
\end{align*}
 Hence, applying Lemma~\ref{lem:prime_filter_and_uparrow=F} and the definition of $I$ again, we get
 
\begin{align*}
\lambda\vert_{[t]\times [t]}(u, v) &
\textnormal{ iff } 1 \in {\uparrow}\bigwedge F_{(\alpha(u), \alpha(v))} 
\textnormal{ iff } \bigwedge F_{(\alpha(u), \alpha(v))} \leqslant 1 
 \textnormal{ iff } \lambda(\alpha(u), \alpha(v)) \leqslant 1 \\
& \textnormal{ iff } \lambda(\alpha(u), \alpha(v)) \in I  
\textnormal { iff } \lambda\vert_{[t]\times [t]}(\alpha\vert_{[t]}(u), \alpha\vert_{[t]}(v)) \in I.
\end{align*}

Using the fact that $\beta$ is a dual-order automorphism of $(X, \leqslant_X)$, we can prove item (v). 
Item (vi) follows from the fact that $\beta: X \to X$ is self-inverse, while (vii) follows from the fact that $\alpha\mathbin{;} \beta\mathbin{;} \alpha = \beta$. 

For (viii), let $u, v \in [t]$ and $a$ an arbitrary element of $A$. Then the equivalences below follow from the definition of $^{\sim}$ and Lemma~\ref{lem:properties_kappa}(ii):
\begin{align*}
\lambda\vert_{[t]\times[t]}(u, v)^{\sim} = \lambda(u, v)^{\sim} \leqslant a = {\sim}{-} a 
\textnormal{ iff } {-}a \leqslant \kappa(\lambda(u, v)) 
 \textnormal{ iff } \lambda(u, v) \not\leqslant {-}a 
 \textnormal{ iff } (u, v) \notin \varphi(-a).
\end{align*}
Using the fact that $\varphi(-a) = -\varphi(a)$ and  Lemmas~\ref{lem:important_eq_injective_map} and~\ref{lem:interaction_lambda_varphi}, we obtain the following:
\begin{align*}
\lambda\vert_{[t]\times[t]}(u, v)^{\sim} \leqslant a
& \textnormal{ iff } (u, v) \notin -\varphi(a) = \alpha\mathbin{;}\varphi(a)^{c\smallsmile} 
 \textnormal{ iff } (u, v) \notin  (\alpha\mathbin{;}\varphi(a)^{\smallsmile})^c \\
& \textnormal{ iff } (u, v) \in \alpha\mathbin{;}\varphi(a)^{\smallsmile} 
 \textnormal{ iff } (v, \alpha(u)) \in \varphi(a) 
\textnormal{ iff } \lambda(v, \alpha(u)) \leqslant a.
\end{align*}
Since $u \in [t]$ and $\alpha\subseteq E$, we get $\alpha(u) \in [t]$, so
$
\lambda\vert_{[t]\times[t]}(u, v)^{\sim} \leqslant a \textnormal{ iff } 
\lambda\vert_{[t]\times[t]}(v, \alpha\vert_{[t]}(u)) \leqslant a.
$
Hence, as $a \in A$ was arbitrary, it follows that $\lambda\vert_{[t]\times[t]}(u, v)^{\sim} = \lambda\vert_{[t]\times[t]}(v, \alpha\vert_{[t]}(u))$.

To prove (ix), let $u, v \in [t]$ and $a$ an arbitrary element of $A$. Then:
\begin{align*}
\lambda\vert_{[t]\times[t]}(u, v)^{\neg} \leqslant a 
& \textnormal{ iff } \lambda(u, v)^{\neg} \leqslant {\neg}{\neg}a 
\textnormal{ iff } {\neg}a \leqslant \kappa(\lambda(u, v)) 
\textnormal{ iff } \lambda(u, v) \not\leqslant {\neg }a 
\textnormal{ iff } (u, v) \notin \varphi(\neg a) \\
& \textnormal{ iff } (u, v) \notin \neg\varphi(a) = \alpha\mathbin{;} \beta\mathbin{;}\varphi(a)^{c}\mathbin{;}\beta 
 \textnormal{ iff } (u, v) \notin  (\alpha\mathbin{;}\beta\mathbin{;}\varphi(a)\mathbin{;}\beta)^c \\
& \textnormal{ iff } (u, v) \in \alpha\mathbin{;}\beta\mathbin{;}\varphi(a)\mathbin{;}\beta 
 \textnormal{ iff } (\beta(\alpha(u)), \beta(v)) \in \varphi(a) 
 \textnormal{ iff } \lambda(\beta(\alpha(u)), \beta(v)) \leqslant a \\
& \textnormal{ iff } \lambda\vert_{[t]\times[t]}(\beta\vert_{[t]}(\alpha\vert_{[t]}(u)), \beta\vert_{[t]}(v)) \leqslant a.
\end{align*}
Since $a \in A$ was arbitrary, this shows that $\lambda\vert_{[t]\times[t]}(u, v)^{\neg} = \lambda\vert_{[t]\times[t]}(\beta\vert_{[t]}(\alpha\vert_{[t]}(u)), \beta\vert_{[t]}(v))$.

Finally, we show that $([t], \lambda\vert_{[t]\times [t]}, \alpha\vert_{[t]}, \beta\vert_{[t]})$ is saturated. First, $\alpha\vert_{[t]}$ and $\beta\vert_{[t]}$ are both total. Moreover, since $\alpha$ and $\beta$ are both surjective, $\alpha\vert_{[t]}$ and $\beta\vert_{[t]}$ are also surjective. 

Now  let $u, v \in [t]$ and assume $\lambda\vert_{[t]\times[t]}(u, v) \in j \circ^{\text{max}} k$. Then $\lambda(u, v) = \lambda\vert_{[t]\times[t]}(u,v) \in j \circ k$, and for all $j', k' \in A_+$, if $j' \leqslant j$, $k' \leqslant k$ and $\lambda(u, v) \in j'\circ k'$, then $j = j'$ and $k = k'$. The first part implies $\lambda(u, v) \leqslant j \cdot k$, and so $j\cdot k \in {\uparrow}\bigwedge F_{(u, v)}$. Hence, by Lemma~\ref{lem:prime_filter_and_uparrow=F}, we have $j\cdot k \in F_{(u, v)}$, which means $(u, v) \in \varphi(j\cdot k) = \varphi(j) \mathbin{;} \varphi(k)$. This implies there exists some $w \in X$ such that $(u, w) \in \varphi(j)$ and $(w, v) \in \varphi(k)$. Consequently, $j \in F_{(u, w)}$ and $k \in F_{(w, v)}$, and therefore $\lambda(u, w) = \bigwedge F_{(u, w)} \leqslant j$ and $\lambda(w, v) = \bigwedge F_{(w, v)} \leqslant k$. Moreover, since $\varphi(j) \subseteq E$ and $u\in [t]$, we obtain $w \in [t]$. Using item (ii) above, we get $\lambda(u, v) \in \lambda(u, w) \circ \lambda(w, v)$. It thus follows that $\lambda(u, w) = j$ and $\lambda(w, v) = k$.  Hence, $\lambda\vert_{[t]\times[t]}(u, w) = j$ and $\lambda\vert_{[t]\times[t]}(w, v) = k$. 
\end{proof}

In the following examples, we use the representations in Examples~\ref{ex:AandCRepresentations},~\ref{ex:BRepresentation} and \ref{ex:DRepresentation} of the DqRAs $\A$, $\B$, $\C$ and $\D$ in Figure~\ref{fig:small_examples_DqRAs} to construct consistent saturated networks for their dual frames in Figure~\ref{fig:dual-frames}. 

\begin{example}\label{ex:network_for_(S3)_+}
The representation of the $3$-element Sugihara chain $\A= \mathbf{S}_3$ yields the network $\mathcal{N}_{\A_+} = (X_1, \lambda_{\A_+}, \alpha_1, \beta_1)$ for $\A_+$, where $\lambda_{\A_+}: X_1\times X_1\to A_+$ is given by $\lambda_{\A_+}(x, x) = \lambda_{\A_+}(y, y) = 1_\A$ and $\lambda_{\A_+}(x, y) = \lambda_{\A_+}(y, x) = \top_\A$. One may readily verify that $\mathcal{N}_{\A_+}$ is a consistent saturated network for $\A_+$.  Moreover, since $\lambda_{\A_+}$ is surjective, $\A_+$ has the network realisation property. Figure~\ref{fig:network_for_S3} depicts this network. 

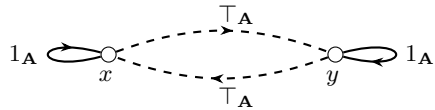
\begin{figure}[ht]
\centering
\begin{tikzpicture}[scale=1.5,pics/sample/.style={code={\draw[#1] (0,0) --(0.6,0) ;}},
Dotted/.style={
dash pattern=on 0.1\pgflinewidth off #1\pgflinewidth,line cap=round,
shorten >=#1\pgflinewidth/2,shorten <=#1\pgflinewidth/2},
Dotted/.default=3]
\footnotesize
\node[draw,circle,inner sep=2pt] (x) at (-2.5,0.7) {};
\node[draw,circle,inner sep=2pt] (y) at (-0.5,0.7) {};
\path (x) edge [bend left=20, dashed, thick, postaction={decorate},
decoration={
markings,
mark=at position 0.55 with {\arrow{Stealth[scale=0.7]}}
}] node {} (y);
\path (y) edge [bend left=20, dashed, thick, postaction={decorate},
decoration={
markings,
mark=at position 0.55 with {\arrow{Stealth[scale=0.7]}}
}] node {} (x);
\draw [thick] (x) edge[loop left,looseness=45,postaction={decorate},
decoration={
markings,
mark=at position 0.75 with {\arrow{Stealth[scale=0.7]}}
}]node{} (x); 
\draw [thick] (y) edge[loop right,looseness=45,postaction={decorate},
decoration={
markings,
mark=at position 0.75 with {\arrow{Stealth[scale=0.7]}}
}]node{} (y); 
\node[label,anchor=north,xshift=-1pt] at (x) {$x$};
\node[label,anchor=north,xshift=-1pt] at (y) {$y$};
\node[label,anchor=east,xshift=-20pt] at (x) {$1_\A$};
\node[label,anchor=west,xshift=20pt] at (y) {$1_\A$};
\node[label,anchor=west,xshift=35pt,yshift=15pt] at (x) {$\top_\A$};
\node[label,anchor=west,xshift=35pt,yshift=-16pt] at (x) {$\top_\A$};
\normalsize
\end{tikzpicture}
\caption{A consistent saturated network for the dual frame $\A_+$ of $\mathbf{A} = \mathbf{S}_3$.}\label{fig:network_for_S3}
\end{figure}
\end{example}

\begin{example}\label{ex:network_for_diamond_1=0^2}
Using the representation of $\mathbf{B} = D^4_{2,2}$ given in Example~\ref{ex:BRepresentation}, we obtain the consistent saturated network
$\mathcal{N}_{\B_+} = (X_2, \lambda_{\B_+}, \alpha_2, \beta_2)$ for $\B_+$, where $\lambda_{\B_+}: X_2 \times X_2 \to B_+$ is given by $\lambda_{\B_+}(u, u) = \lambda_{\B_+}(v, v) = 1_\B$ and $\lambda_{\B_+}(u, v) = \lambda_{\B_+}(v, u) = 0_\B$. Clearly then $\B_+$ has the network realisation property. Figure~\ref{fig:network_for_diamond_1=0^2} shows this network.  

\begin{figure}[ht]
\centering
\begin{tikzpicture}[scale=1.5,pics/sample/.style={code={\draw[#1] (0,0) --(0.6,0) ;}},
Dotted/.style={
dash pattern=on 0.1\pgflinewidth off #1\pgflinewidth,line cap=round,
shorten >=#1\pgflinewidth/2,shorten <=#1\pgflinewidth/2},
Dotted/.default=3]
\footnotesize
\node[draw,circle,inner sep=2pt] (u) at (-2.5,0.7) {};
\node[draw,circle,inner sep=2pt] (v) at (-0.5,0.7) {};
\path (u) edge [bend left=20, dashed, thick, postaction={decorate},
decoration={
markings,
mark=at position 0.55 with {\arrow{Stealth[scale=0.7]}}
}] node {} (v);
\path (v) edge [bend left=20, dashed, thick, postaction={decorate},
decoration={
markings,
mark=at position 0.55 with {\arrow{Stealth[scale=0.7]}}
}] node {} (u);
\draw [thick] (u) edge[loop left,looseness=45,postaction={decorate},
decoration={
markings,
mark=at position 0.75 with {\arrow{Stealth[scale=0.7]}}
}]node{} (u); 
\draw [thick] (v) edge[loop right,looseness=45,postaction={decorate},
decoration={
markings,
mark=at position 0.75 with {\arrow{Stealth[scale=0.7]}}
}]node{} (v); 
\node[label,anchor=north,xshift=-1pt] at (u) {$u$};
\node[label,anchor=north,xshift=-1pt] at (v) {$v$};
\node[label,anchor=east,xshift=-20pt] at (u) {$1_\B$};
\node[label,anchor=west,xshift=20pt] at (v) {$1_\B$};
\node[label,anchor=west,xshift=35pt,yshift=15pt] at (x) {$0_\B$};
\node[label,anchor=west,xshift=35pt,yshift=-16pt] at (x) {$0_\B$};
\normalsize
\end{tikzpicture}
\caption{A consistent saturated network for the dual frame $\B_+$ of the $4$-element DqRA $\mathbf{B}$ in Figure~\ref{fig:small_examples_DqRAs}.}\label{fig:network_for_diamond_1=0^2}
\end{figure}
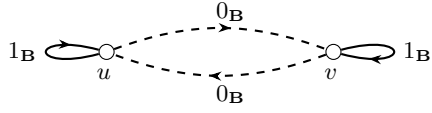
\end{example}

\begin{example}\label{ex:network_for_diamond_with_a}
The representation in Example~\ref{ex:AandCRepresentations} of the DqRA $\C = D^4_{3,2}$ yields the consistent saturated network $\mathcal{N}_{\C_+} = (X_1, \lambda_{\C_+}, \alpha_1, \beta_1)$ for $\C_+$, where $\lambda_{\C_+}: X_1\times X_1\to C_+$ is given by $\lambda_{\C_+}(x, x) = \lambda_{\C_+}(y, y) = 1_\C$ and $\lambda_{\C_+}(x, y) = \lambda_{\C_+}(y, x) = a$.   Figure~\ref{fig:network_for_Diamond} depicts this network. As in the previous examples, $\lambda_{\C_+}$ is surjective, so $\C_+$ has the network realisation property.

\begin{figure}[ht]
\centering
\begin{tikzpicture}[scale=1.5,pics/sample/.style={code={\draw[#1] (0,0) --(0.6,0) ;}},
Dotted/.style={
dash pattern=on 0.1\pgflinewidth off #1\pgflinewidth,line cap=round,
shorten >=#1\pgflinewidth/2,shorten <=#1\pgflinewidth/2},
Dotted/.default=3]
\footnotesize
\node[draw,circle,inner sep=2pt] (x) at (-2.5,0.7) {};
\node[draw,circle,inner sep=2pt] (y) at (-0.5,0.7) {};
\path (x) edge [bend left=20, dashed, thick, postaction={decorate},
decoration={
markings,
mark=at position 0.55 with {\arrow{Stealth[scale=0.7]}}
}] node {} (y);
\path (y) edge [bend left=20, dashed, thick, postaction={decorate},
decoration={
markings,
mark=at position 0.55 with {\arrow{Stealth[scale=0.7]}}
}] node {} (x);
\draw [thick] (x) edge[loop left,looseness=45,postaction={decorate},
decoration={
markings,
mark=at position 0.75 with {\arrow{Stealth[scale=0.7]}}
}]node{} (x); 
\draw [thick] (y) edge[loop right,looseness=45,postaction={decorate},
decoration={
markings,
mark=at position 0.75 with {\arrow{Stealth[scale=0.7]}}
}]node{} (y); 
\node[label,anchor=north,xshift=-1pt] at (x) {$x$};
\node[label,anchor=north,xshift=-1pt] at (y) {$y$};
\node[label,anchor=east,xshift=-20pt] at (x) {$1_\C$};
\node[label,anchor=west,xshift=20pt] at (y) {$1_\C$};
\node[label,anchor=west,xshift=35pt,yshift=15pt] at (x) {$a$};
\node[label,anchor=west,xshift=35pt,yshift=-16pt] at (x) {$a$};
\normalsize
\end{tikzpicture}
\caption{A consistent saturated network for the dual frame $\C_+$ of the $4$-element DqRA $\C$ in Figure~\ref{fig:small_examples_DqRAs}.}\label{fig:network_for_Diamond}
\end{figure}
\end{example}

\begin{example}
\label{ex:network_for_(S4)_+}
Using the representation of $\mathbf{D} = \mathbf{S}_4$ given in Example~\ref{ex:DRepresentation}, we obtain the consistent saturated network $\mathcal{N}_{\D_+} = (X_3, \lambda_{\D_+}, \alpha_3, \beta_3)$ for $\D_+$, where $\lambda_{\D_+}: X_3 \times X_3 \to D_+$ is given by the table in Figure~\ref{fig:network_for_S4}. Clearly $\D_+$ has the network realisation property. 
The network is also shown in Figure~\ref{fig:network_for_S4}. 
\end{example}

\begin{figure}[ht]
\centering
\begin{tikzpicture}[scale=1.5,pics/sample/.style={code={\draw[#1] (0,0) --(0.6,0) ;}},
Dotted/.style={
dash pattern=on 0.1\pgflinewidth off #1\pgflinewidth,line cap=round,
shorten >=#1\pgflinewidth/2,shorten <=#1\pgflinewidth/2},
Dotted/.default=3]
\footnotesize
\node[draw,circle,inner sep=2pt] (y) at (-1,1) {};
\node[draw,circle,inner sep=2pt] (z) at (1,1) {};
\node[draw,circle,inner sep=2pt] (x) at (0,0) {};
\node[draw,circle,inner sep=2pt] (v) at (-1,-1) {};
\node[draw,circle,inner sep=2pt] (w) at (1,-1) {};
\node[label,anchor=south,xshift=-8pt,yshift=-10pt] at (y) {$v$};
\node[label,anchor=south,yshift=-8pt, xshift=8pt] at (z) {$w$};
\node[label,anchor=north,xshift=-8pt,yshift=8pt] at (v) {$s$};
\node[label,anchor=north,xshift=10pt,yshift=8pt] at (w) {$t$};
\node[label,anchor=east,xshift=-2pt] at (x) {$u$};
\draw [dashed,thick] (x) edge[loop above,looseness=45,postaction={decorate},
decoration={
markings,
mark=at position 0.75 with {\arrow{Stealth[scale=1]}}
}]node{} (x); 
\draw [thick] (y) edge[loop above,looseness=45, postaction={decorate},
decoration={
markings,
mark=at position 0.75 with {\arrow{Stealth[scale=1]}}
}]node{} (y); 
\draw [thick] (z) edge[loop above,looseness=45, postaction={decorate},
decoration={
markings,
mark=at position 0.75 with {\arrow{Stealth[scale=1]}}
}]node{} (z); 
\draw [thick] (v) edge[loop below,looseness=45, postaction={decorate},
decoration={
markings,
mark=at position 0.75 with {\arrow{Stealth[scale=1]}}
}]node{} (v);
\draw [thick] (w) edge[loop below,looseness=45, postaction={decorate},
decoration={
markings,
mark=at position 0.75 with {\arrow{Stealth[scale=1]}}
}]node{} (w);
\path (y) edge [bend left=20, dotted, thick, postaction={decorate},
decoration={
markings,
mark=at position 0.55 with {\arrow{Stealth[scale=1]}}
}] node {} (v);
\path (z) edge [bend left=20, dotted, thick, postaction={decorate},
decoration={
markings,
mark=at position 0.55 with {\arrow{Stealth[scale=1]}}
}] node {} (w);
\path (y) edge [bend left=20, dotted, thick, postaction={decorate},
decoration={
markings,
mark=at position 0.55 with {\arrow{Stealth[scale=1]}}
}] node {} (x);
\path (y) edge [bend left=20, dotted, thick, postaction={decorate},
decoration={
markings,
mark=at position 0.55 with {\arrow{Stealth[scale=1]}}
}] node {} (z);
\path (z) edge [bend left=20, dotted, thick, postaction={decorate},
decoration={
markings,
mark=at position 0.55 with {\arrow{Stealth[scale=1]}}
}] node {} (y);
\path (v) edge [bend left=20, dotted, thick, postaction={decorate},
decoration={
markings,
mark=at position 0.55 with {\arrow{Stealth[scale=1]}}
}] node {} (w);
\path (w) edge [bend left=20, dotted, thick, postaction={decorate},
decoration={
markings,
mark=at position 0.55 with {\arrow{Stealth[scale=1]}}
}] node {} (v);
\path (x) edge [bend left=20, dotted, thick, postaction={decorate},
decoration={
markings,
mark=at position 0.55 with {\arrow{Stealth[scale=1]}}
}] node {} (w);
\path (z) edge [bend left=23, dotted, thick, postaction={decorate},
decoration={
markings,
mark=at position 0.525 with {\arrow{Stealth[scale=1]}}
}] node {} (v);
\path (y) edge [bend left=23, dotted, thick, postaction={decorate},
decoration={
markings,
mark=at position 0.525 with {\arrow{Stealth[scale=1]}}
}] node {} (w);
\path (z) edge [bend left=20, dotted, thick, postaction={decorate},
decoration={
markings,
mark=at position 0.55 with {\arrow{Stealth[scale=1]}}
}] node {} (x);
\path (x) edge [bend left=20, dotted, thick, postaction={decorate},
decoration={
markings,
mark=at position 0.55 with {\arrow{Stealth[scale=1]}}
}] node {} (v);
\path (v) edge [bend left=20, thick, postaction={decorate},
decoration={
markings,
mark=at position 0.55 with {\arrow{Stealth[scale=1]}}
}] node {} (y);
\path (w) edge [bend left=20, thick, postaction={decorate},
decoration={
markings,
mark=at position 0.55 with {\arrow{Stealth[scale=1]}}
}] node {} (z);
\path (w) edge [bend left=20, thick, postaction={decorate},
decoration={
markings,
mark=at position 0.55 with {\arrow{Stealth[scale=1]}}
}] node {} (x);
\path (x) edge [bend left=20, thick, postaction={decorate},
decoration={
markings,
mark=at position 0.55 with {\arrow{Stealth[scale=1]}}
}] node {} (y);
\path (w) edge [bend left=23, thick, postaction={decorate},
decoration={
markings,
mark=at position 0.525 with {\arrow{Stealth[scale=1]}}
}] node {} (y);
\path (v) edge [bend left=23, thick, postaction={decorate},
decoration={
markings,
mark=at position 0.525 with {\arrow{Stealth[scale=1]}}
}] node {} (z);
\path (v) edge [bend left=20, thick, postaction={decorate},
decoration={
markings,
mark=at position 0.55 with {\arrow{Stealth[scale=1]}}
}] node {} (x);
\path (x) edge [bend left=20, thick, postaction={decorate},
decoration={
markings,
mark=at position 0.55 with {\arrow{Stealth[scale=1]}}
}] node {} (z);
\path (2.5,0.4) 
node[matrix,anchor=north east,draw,nodes={anchor=center},inner sep=2pt, thick]  {
\pic{sample=solid}; & \node{$0_\D$}; \\
\pic{sample=dashed}; & \node{$1_\D$}; \\
\pic{sample=dotted}; & \node{$\top_\D$}; \\
};
\node[anchor=west] at (-5.5,-0.15){
$\begin{array}{|c|ccccc|}\hline
\lambda_{\D_+} & s &t& u & v& w\\\hline
s & 0_\D & \top_\D & 0_\D & 0_\D & 0_\D	\\
t & \top_\D & 0_\D & 0_\D  & 0_\D & 0_\D	\\
u & \top_\D  & \top_\D & 1_\D & 0_\D & 0_\D \\
v & \top_\D & \top_\D & \top_\D & 0_\D & \top_\D \\
w & \top_\D & \top_\D & \top_\D & \top_\D & 0_\D \\ \hline
\end{array}$
};
\normalsize
\end{tikzpicture}
\caption{A consistent saturated network for the dual frame $\D_+$ of $\D = \mathbf{S}_4$ in Figure~\ref{fig:small_examples_DqRAs}.}
\label{fig:network_for_S4}
\end{figure}
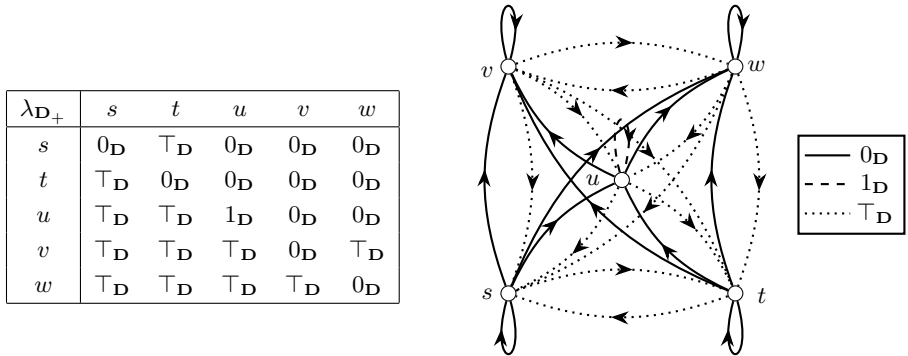

\section{Representability via wreath products} \label{sec:NetworkForWreathPRoduct}

In this section, we demonstrate how the wreath product construction introduced in Section~\ref{sec:wreath_products} can be used to obtain representable DqRAs. More precisely, we show that, under certain conditions, the complex algebra of the wreath product of the dual frames of two finite normally representable DqRAs is itself representable. The key step in the proof is to show that the wreath product has the network realisation property so that we can apply Proposition~\ref{prop:consistent_sat_network_imply_representable}. 

The following lemma will be very useful. 

\begin{lemma}\label{lem:LambdaEqualToOne}
Let 
$\B=\langle B,\wedge,\vee,\cdot,  {\sim}, -,\neg, 1\rangle$ be a complete perfect DqRA with $1 \in B_+$ and $1^{-} = 1^{\sim} = 1^{\neg} = 1$, and let $\mathcal{N}=\left(N,\lambda,\alpha,\beta\right)$
be a consistent network for $\B_+ = (B_+, I, \preccurlyeq, \circ, ^{\sim}, ^{-}, ^{\neg})$. 
\begin{enumerate}[label=\textup{(\roman*)}]
\item If $\alpha(x)$ and $\alpha(y)$ are defined, then $\lambda(x,y)=1$ iff 
$\lambda(y,\alpha(x))=1$
iff $\lambda(\alpha(x),\alpha(y))=1$.  
\item If $\alpha(x)$, $\beta(x)$ and $\beta(y)$ are defined, then $\lambda(x,y)=1$ iff
$\lambda(\beta(\alpha(x)),\beta(y))=1$
iff
$\left.\lambda(\beta(y),\beta(x))=1\right.$.
\end{enumerate}
\end{lemma}
\begin{proof}
    To prove item (i), we use the fact that $\mathcal{N}$ 
    satisfies condition (viii) from Definition~\ref{def:consistent_network_for_frame} to get
    \begin{align*}
        \lambda(x,y)=1&\textnormal{ iff } 
        \lambda(y,\alpha(x))=\lambda(x,y)^\sim=1^{\sim}=1\\
        &\textnormal{ iff } \lambda(\alpha(x),\alpha(y))=\lambda(y,\alpha(x))^{\sim}=1^{\sim}=1.
    \end{align*}
    For item (ii), we use conditions (vii), (viii) and (ix) of Definition~\ref{def:consistent_network_for_frame}, satisfied by $\mathcal{N}$, to get
    \begin{align*}
        \lambda(x,y)=1
        \textnormal{ iff } &\lambda(\beta(\alpha(x)),\beta(y))=\lambda(x,y)^\neg=1^\neg=1\\
        \textnormal{ iff } &
        \lambda(\beta(y),\beta(x))=
        \lambda(\beta(y),\alpha(\beta(\alpha(x))))
        = \lambda(\beta(\alpha(x)),\beta(y))^\sim=1^\sim=1. 
    \end{align*}    
\end{proof}

Now let $\A=\langle A,\wedge_\A,\vee_\A, \cdot_\A, {\sim}_\A, -_\A, \neg_\A, 1_\A\rangle$
and $\B=\langle B,\wedge_\B,\vee_\B,\cdot_\B,  {\sim}_\B, -_\B,\neg_\B, 1_\B\rangle$ be finite DqRAs such that $A \cap B=\varnothing$ and $1_\mathbf{A}\in \JA$, $1_{\B} \in \JB$ and $1_\B^{-_{\B_+}} = 1^{\sim_{\B_+}}_\B = 1_\B^{\neg_{\B_+}} = 1_\B$. 
Assume there exists a consistent saturated network  
$\mathcal{N}_{\A_+}=\left(N_{\A_+},\lambda_{\A_+},\alpha_{\A_+},\beta_{\A_+}\right)$ for $\A_+ = \left(A_+,I_{\A_+},\preccurlyeq_{\A_+},\circ_{\A_+},^{\sim_{\A_+}},^{-_{\A_+}},^{\neg_{\A_+}}\right)$. 
Likewise, assume there exists a consistent saturated network
$\mathcal{N}_{\B_+}=\left(N_{\B_+},\lambda_{\B_+},\alpha_{\B_+},\beta_{\B_+}\right)$ for $\B_+ = \left(B_+,I_{\B_+},\preccurlyeq_{\B_+},\circ_{\B_+},^{\sim_{\B_+}},^{-_{\B_+}},^{\neg_{\B_+}}\right)$. 
Recall the definition of the wreath product of $\A_+$ and $\B_+$, described above and below Lemma~\ref{lem:wreat_product_poset}.

For each $s\in N_{\B_+}$, let $_sN_{\A_+}$ be a copy of $N_{\A_+}$ and
let $x^s$ denote the copy of $x$ in $_sN_{\A_+}$.
Set $$N_{\A_+[\B_+]}=\displaystyle\bigcup_{s\in N_{\B_+} }{}_sN_{\A_+}.$$  That is, $N_{\A_+[\B_+]}$ contains a copy of $N_{\A_+}$ for each element of
$N_{\B_+}$.
Define an interim function $\lambda_{\A_+[\B_+]}:N_{\A_+[\B_+]}\times N_{\A_+[\B_+]}\to A_+[B_+]$ as follows:  for all $x^s,y^t \in N_{\A_+[\B_+]}$,
\[
\lambda_{\A_+[\B_+]}\left(x^s,y^t\right)=
\begin{cases}
   \lambda_{\A_+}(x,y)&\text{ if }\lambda_{\B_+}(s,t)=1_{\B}\\
   \lambda_{\B_+}(s,t)&\text{ if }\lambda_{\B_+}(s,t)\neq 1_{\B}.
\end{cases}
\]

\begin{lemma}\label{Lem:InterimLemmaOneAndTwo}
    The map $\lambda_{\A_+[\B_+]}: N_{\A_+[\B_+]}\times N_{\A_+[\B_+]}\to A_+[B_+]$ defined above satisfies 
    conditions (i) and (ii) of Definition~\ref{def:consistent_network_for_frame}.
\end{lemma}
\begin{proof}
    Recall that $I = I_{\A_+} \cup (I_{\B_+}{\setminus} \{1_\B\})$. Let $x^s\in N_{\A_+[\B_+]}$, then it
    follows from the definition of $\lambda_{\A_+[\B_+]}$  that $\lambda_{\A_+[\B_+]}\left(x^s,x^s\right)\in I$ and hence (i) is satisfied.

    For item (ii) of Definition~\ref{def:consistent_network_for_frame}, let $x^s,y^t,z^r\in N_{\A_+[\B_+]}$ and consider the following cases:
    
    Case 1: $\lambda_{\B_+}(s, t) = 1_\B$ and $\lambda_{\B_+}(t, r) = 1_\B$. Then \begin{align*}
    \lambda_{\A_+[\B_+]}(x^s,y^t)\circ 
    \lambda_{\A_+[\B_+]}(y^t,z^r) & = \lambda_{\A_+}(x, y) \circ \lambda_{\A_+}(y, z) \\
    & = (\lambda_{\A_+}(x, y) \circ_{\A_+} \lambda_{\A_+}(y, z)) \cup (I_{\B_+}{\setminus}\{1_\B\}).
    \end{align*}
    Since $\mathcal{N}_{\B_+}$ is a consistent network, we have $\lambda_{\B_+}(s, r) \in \lambda_{\B_+}(s,t) \circ_{\B_+} \lambda_{\B_+}(t, r)$, which means $\lambda_{\B_+}(s, r) \leqslant_\B \lambda_{\B_+}(s,t) \cdot_{\B} \lambda_{\B_+}(t, r) = 1_\B \cdot_{\B} 1_\B = 1_\B$.
    If $\lambda_{\B_+}(s,r)=1_\B$, then $\lambda_{\A_+[\B_+]}(x^s,z^r)=\lambda_{\A_+}(x,z)$ and so, since $\mathcal{N}_{\A_+}$ is
    consistent, 
    \begin{align*}
   \lambda_{\A_+[\B_+]}(x^s, z^r) = \lambda_{\A_+}(x,z)\in \lambda_{\A_+}(x,y)\circ_{\A_+} 
    \lambda_{\A_+}(y,z) & \subseteq \lambda_{\A_+[\B_+]}(x^s,y^t)\circ 
    \lambda_{\A_+[\B_+]}(y^t,z^r).
    \end{align*}
    If 
    $\lambda_{\B_+}(s,r) <_\B 1_\B$, then 
    $\lambda_{\A_+[\B_+]}(x^s,z^r)=\lambda_{\B_+}(s,r)   \in I_{\B_+}\backslash\{1_\B\} \subseteq \lambda_{\A_+[\B_+]}(x^s,y^t)\circ 
    \lambda_{\A_+[\B_+]}(y^t,z^r)$.
        
    Case 2: $\lambda_{\B_+}(s, t) = 1_\B$ and $\lambda_{\B_+}(t, r) \neq 1_\B$ 
    with $1_\B\notin 1_\B\circ_{\B_+}\lambda_{\B_+}(t,r)$. Then  
    \begin{align*}
    \lambda_{\A_+[\B_+]}(x^s,y^t)\circ 
    \lambda_{\A_+[\B_+]}(y^t,z^r) & = \lambda_{\A_+}(x, y)\circ
    \lambda_{\B_+}(t,r)  
    =  1_\B \circ_{\B_+} \lambda_{\B_+}(t, r).
    \end{align*}
    Since $\mathcal{N}_{\B_+}$ is consistent, we have  
    $\lambda_{\B_+}(s,r)\in \lambda_{\B_+}(s,t)\circ_{\B_+} 
    \lambda_{\B_+}(t,r) = 1_\B \circ_{\B_+} \lambda_{\B_+}(t, r)$, and as $1_\B\notin \lambda_{\B_+}(s,t)\circ_{\B_+}\lambda_{\B_+}(t,r)$ it follows that $\lambda_{\B_+}(s,r)\neq 1_{\B}$. Hence,
    $\lambda_{\A_+[\B_+]}(x^s,z^r)=\lambda_{\B_+}(s,r)$,
    and the result follows.
    
    Case 3: $\lambda_{\B_+}(s, t) = 1_\B$ and $\lambda_{\B_+}(t, r) \neq 1_\B$ 
   with $1_\B\in 1_\B\circ_{\B_+}\lambda_{\B_+}(t,r)$.
    Here 
    \begin{align*}
    \lambda_{\A_+[\B_+]}(x^s,y^t)\circ 
    \lambda_{\A_+[\B_+]}(y^t,z^r) & = \lambda_{\A_+}(x, y)\circ 
    \lambda_{\B_+}(t,r) 
    =((1_\B\circ_{\B_+} 
    \lambda_{\B_+}(t,r))\backslash\{1_{\B}\})\cup A_+.
    \end{align*}
    If $\lambda_{\B_+}(s,r)=1_\B$
    , then 
    $\lambda_{\A_+[\B_+]}(x^s,z^r)=\lambda_{\A_+}(x,z)\in A_+$.  If $\lambda_{\B_+}(s,r)\neq 1_{\B}$, then we have $\lambda_{\A_+[\B_+]}(x^s,z^r)=\lambda_{\B_+}(s,r)\in
    \lambda_{\B_+}(s,t)\circ_{\B_+} 
    \lambda_{\B_+}(t,r) = 1_\B \circ_{\B_+} \lambda_{\B_+}(t, r)$ since $\mathcal{N}_{\B_+}$ is a consistent network.  In either case, the result follows.
    
    Case 4: $\lambda_{\B_+}(s, t) \neq 1_\B$ and $\lambda_{\B_+}(t, r) = 1_\B$ 
    with $1_\B\notin \lambda_{\B_+}(s,t)\circ_{\B_+} 1_\B$. This case is similar to Case 2. 
    
    Case 5: $\lambda_{\B_+}(s, t) \neq 1_\B$ and $\lambda_{\B_+}(t, r) = 1_\B$ 
    with $1_\B\in \lambda_{\B_+}(s,t)\circ_{\B_+}1_\B$. This case is similar to Case 3.
    
    Case 6: $\lambda_{\B_+}(s, t) \neq 1_\B$ and $\lambda_{\B_+}(t, r) \neq 1_\B$ 
    with $1_\B\notin \lambda_{\B_+}(s,t)\circ_{\B_+}\lambda_{\B_+}(t,r)$. In this case 
    \begin{align*}
    \lambda_{\A_+[\B_+]}(x^s,y^t)\circ \lambda_{\A_+[\B_+]}(y^t,z^r) & = \lambda_{\B_+}(s,t)\circ\lambda_{\B_+}(t, r) = \lambda_{\B_+}(s,t)\circ_{\B_+}\lambda_{\B_+}(t, r). 
    \end{align*}
    Since $\mathcal{N}_{\B_+}$ is consistent, we have $\lambda_{\B_+}(s, r) \in \lambda_{\B_+}(s, t)\circ_{\B_+} \lambda_{\B_+}(t, r)$, and so it must be the case that $\lambda_{\B_+}(s, r) \neq 1_\B$. Hence, $\lambda_{\A_+[\B_+]}(x^s, z^r) = \lambda_{\B_+}(s, r) \in \lambda_{\A_+[\B_+]}(x^s,y^t)\circ \lambda_{\A_+[\B_+]}(y^t,z^r)$. 
    
    Case 7: $\lambda_{\B_+}(s, t) \neq 1_\B$ and $\lambda_{\B_+}(t, r) \neq 1_\B$ 
    with $1_\B \in \lambda_{\B_+}(s,t)\circ_{\B_+}\lambda_{\B_+}(t,r)$. Then 
    \begin{align*}
    \lambda_{\A_+[\B_+]}(x^s,y^t)\circ \lambda_{\A_+[\B_+]}(y^t,z^r) & = \lambda_{\B_+}(s,t)\circ\lambda_{\B_+}(t, r)  \\
    & = ((\lambda_{\B_+}(s,t)\circ_{\B_+}\lambda_{\B_+}(t, r)){\setminus}\{1_\B\}) \cup A_+. 
    \end{align*}
    If $\lambda_{\B_+}(s,r) = 1_\B$, then $\lambda_{\A_+[\B_+]}(x^s, z^r) = \lambda_{\A_+}(s, r) \in A_+$, and the result follows. On the other hand, if $\lambda_{\B_+}(s,r) \neq 1_\B$, then $\lambda_{\A_+[\B_+]}(x^s,z^r)=\lambda_{\B_+}(s,r)\in
    \lambda_{\B_+}(s,t)\circ_{\B_+} 
    \lambda_{\B_+}(t,r)$ since $\mathcal{N}_{\B_+}$ is consistent. Thus, $\lambda_{\A_+[\B_+]}(x^s,z^r) \in \lambda_{\A_+[\B_+]}(x^s,y^t)\circ \lambda_{\A_+[\B_+]}(y^t,z^r)$.
\end{proof}

Define a binary relation ${\equiv} \subseteq \left(\bigcup_{s\in N_{\B_+}} {}_sN_{\A_+}\right)\times \left(\bigcup_{s\in N_{\B_+}} {}_sN_{\A_+}\right)$ by  
\[x^s\equiv y^t\quad \textnormal{iff} \quad 
\lambda_{\A_+[\B_+]}(x^s,y^t)\in I\quad\text{and}\quad
\lambda_{\A_+[\B_+]}(y^t,x^s)\in I.\]

\begin{lemma}\label{lem:EquivRel}
The binary relation $\equiv$ defined above is an equivalence relation.
\end{lemma}
\begin{proof}
It is easy to see that $\equiv$ is symmetric. That $\equiv$ is reflexive follows from the fact that $\lambda_{\A_+[\B_+]}$ satisfies item (i) from Definition~\ref{def:consistent_network_for_frame} as was shown in Lemma~\ref{Lem:InterimLemmaOneAndTwo}.  

To
see that $\equiv$ is transitive, let $x^s,y^t,z^r\in N_{\A_+[\B_+]}$ such that $x^s\equiv y^t$ and $y^t\equiv z^r$. Then 
$\lambda_{\A_+[\B_+]}(x^s,y^t)\in I$ and $\lambda_{\A_+[\B_+]}(y^t,x^s)\in I$, and $\lambda_{\A_+[\B_+]}(y^t,z^r)\in I$ and $\lambda_{\A_+[\B_+]}(z^r,y^t)\in I$.
We show that $\lambda_{\A_+[\B_+]}(x^s, z^r) \in I$. The proof that $\lambda_{\A_+[\B_+]}(z^r, x^s) \in I$ is similar. There are four cases to consider.

Case 1: $\lambda_{\A_+[\B_+]}(x^s,y^t)\in I_{\A_+}$ and $\lambda_{\A_+[\B_+]}(y^t,z^r)\in I_{\A_+}$. This implies $\lambda_{\B_+}(s,t)=1_\B$ and $\lambda_{\B_+}(t,r)=1_\B$, so that
$\lambda_{\A_+[\B_+]}(x^s,y^t)=\lambda_{\A_+}(x,y)$
and $\lambda_{\A_+[\B_+]}(y^t,z^r)=\lambda_{\A_+}(y,z)$.
If $\lambda_{\B_+}(s,r)=1_{\B}$,
then $\lambda_{\A_+[\B_+]}(x^s,z^r)=\lambda_{\A_+}(x,z)\in
\lambda_{\A_+}(x,y)\circ_{\A_+}\lambda_{\A_+}(y,z)$ since $\mathcal{N}_{\A_+}$ is consistent. Hence, by Lemma~\ref{lem:IcircI=I}, $\lambda_{\A_+[\B_+]}(x^s,z^r)=\lambda_{\A_+}(x,z)\in I_{\A_+}\circ_{\A_+} I_{\A_+}=I_{\A_+}\subseteq I$.
On the other hand, if $\lambda_{\B_+}(s,r)\neq 1_{\B}$, then $\lambda_{\A_+[\B_+]}(x^s,z^r)=\lambda_{\B_+}(s,r)\in \lambda_{\B_+}(s,t)\circ_{\B_+}\lambda_{\B_+}(t,r)$, which shows $\lambda_{\A_+[\B_+]}(x^s,z^r)=\lambda_{\B_+}(s,r) \in I_{\B_+}\circ_{\B_+} I_{\B_+}=I_{\B_+}$, where the final equality follows from 
Lemma~\ref{lem:IcircI=I}. It follows that $\lambda_{\A_+[\B_+]}(x^s,z^r) \in I_{\B_+}{\setminus}\{1_\B\}\subseteq I$.

Case 2: $\lambda_{\A_+[\B_+]}(x^s,y^t)\in I_{\A_+}$ and $\lambda_{\A_+[\B_+]}(y^t,z^r)\in I_{\B_+}{\setminus}\{1_\B\}$.
Then $\lambda_{\A_+[\B_+]}(x^s,y^t)=\lambda_{\A_+}(x,y)$
with $\lambda_{\B_+}(s,t)=1_\B$ and 
$\lambda_{\A_+[\B_+]}(y^t,z^r)=\lambda_{\B_+}(t,r) <_\B 1_\B$. Since $\mathcal{N}_{\B_+}$ is consistent, we get $\lambda_{\B_+}(s, r) \in \lambda_{\B_+}(s,t) \circ_{\B_+} \lambda_{\B_+}(t, r)$, and so $\lambda_{\B_+}(s, r) \leqslant_\B 1_\B \cdot_\B \lambda_{\B_+}(t, r) <_\B 1_\B$. It follows that $\lambda_{\A_+[\B_+]}(x^s, z^r) = \lambda_{\B_+}(s, r) \in I_{\B_+}{\setminus}\{1_\B\} \subseteq I$. 

Case 3: $\lambda_{\A_+[\B_+]}(x^s,y^t)\in I_{\B_+}{\setminus}\{1_\B\}$ and $\lambda_{\A_+[\B_+]}(y^t,z^r)\in I_{\A_+}$. This case is similar to the previous case.

Case 4: $\lambda_{\A_+[\B_+]}(x^s,y^t)\in I_{\B_+}{\setminus}\{1_\B\}$ and $\lambda_{\A_+[\B_+]}(y^t,z^r)\in I_{\B_+}{\setminus}\{1_\B\}$. In this case we obtain  
$\lambda_{\A_+[\B_+]}(x^s,y^t)=\lambda_{\B_+}(s,t) <_\B 1_\B$ and $\lambda_{\A_+[\B_+]}(y^t,z^r)=\lambda_{\B_+}(t,r) <_\B 1_\B$.  
Since $\mathcal{N}_{\B_+}$ is consistent, we have $\lambda_{\B_+}(s,r)\in \lambda_{\B_+}(s,t)\circ_{\B_+}\lambda_{\B_+}(t,r)$, 
and consequently, $\lambda_{\B_+}(s,r)\leqslant_\B \lambda_{\B_+}(s,t)\cdot_\B\lambda_{\B_+}(t,r)<_\B 1_{\B}$. 
This shows that $\lambda_{\A_+[\B_+]}(x^s,z^r)=\lambda_{\B_+}(s,r)\in I_{\B_+}{\setminus}\{1_\B\}\subseteq I$.  

Since $\lambda_{\A_+[\B_+]}(x^s, z^r) \in I$ and $\lambda_{\A_+[\B_+]}(z^r, x^s) \in I$, we can conclude that $x^s \equiv z^r$. 
\end{proof}

Set
$\displaystyle N=N_{\A_+[\B_+]}/{\equiv}$ and
define $\lambda:N\times N\to A_+[B_+]$ by setting, for all $[x^s],[y^t]\in N$, \begin{equation}\label{eq:lambda_final}
\lambda([x^s],[y^t])=\lambda_{\A_+[\B_+]}(x^s,y^t).
\end{equation}

\begin{lemma}\label{lem:lambdaWellDefined}
The map $\lambda$, defined above, is well-defined.  That is, for
all $[w^r],[x^s],[y^t],[z^u]\in N$, we have that
$\lambda([w^r],[y^t])=\lambda([x^s],[z^u])$ whenever $[w^r]=[x^s]$ and $[y^t] =[z^u]$.
\end{lemma}

\begin{proof}
    It is convenient to prove this in two steps, namely, $[w^r]=[x^s]$ implies $\lambda([w^r],[y^t])=\lambda([x^s],[y^t])$ and $[y^t] =[z^u]$ implies $\lambda([w^r],[y^t])=\lambda([w^r],[z^u])$. We prove the first implication and leave the second implication for the reader. 
    Let $[w^r], [x^s], [y^t]\in N$ be such that $[w^r]=[x^s]$, i.e.,
    $w^r\equiv x^s$.  Then $\lambda_{\A_+[\B_+]}(w^r, x^s)\in I$ and $\lambda_{\A_+[\B_+]}(x^s, w^r)\in I$ by definition
    and we consider four cases.

    Case 1: $\lambda_{\A_+[\B_+]}(w^r, x^s)\in I_{\A_+}$ and $\lambda_{\A_+[\B_+]}(x^s, w^r)\in I_{\A_+}$. In this case, $\lambda_{\B_+}(r, s)=1_\B$ and $\lambda_{\B_+}(s, r)=1_\B$ so that  $\lambda_{\A_+[\B_+]}(w^r, x^s)=\lambda_{\A_+}(w, x)\in I_{\A_+}$ and
    $\lambda_{\A_+[\B_+]}(x^s, w^r)=\lambda_{\A_+}(x, w)\in I_{\A_+}$  Then $x=w$ by (iii) of Definition~\ref{def:consistent_network_for_frame} since $\mathcal{N}_{\A_+}$ is a consistent network.  As $1_\B\in I_{\B_+}$
    we conclude that $r=s$  since $\mathcal{N}_{\B_+}$ is also a consistent network.  Hence, $w^r=x^s$ and the claim follows.

    Case 2: $\lambda_{\A_+[\B_+]}(w^r, x^s)\in I_{\A_+}$ and $\lambda_{\A_+[\B_+]}(x^s, w^r)\in I_{\B_+}$. The first part implies $\lambda_{\B_+}(r, s)=1_\B$ so that
    $\lambda_{\A_+[\B_+]}(w^r, x^s)=\lambda_{\A_+}(w, x)\in I_{\A_+}$, while the second part gives $\lambda_{\A_+[\B_+]}(x^s, w^r)=\lambda_{\B_+}(s, r)<_\B 1_\B$.  Then $r=s$ by Definition~\ref{def:consistent_network_for_frame}(iii) since $\mathcal{N}_{\B_+}$ is a consistent network. 
    Hence, if $\lambda_{\B_+}(r, t)=1_\B$, then $\lambda_{\B_+}(s, t)=1_{\B}$. Since $\mathcal{N}_{\B_+}$ is
    consistent, it follows from Definition~\ref{def:consistent_network_for_frame}(ii) that
    $\lambda_{\B_+}(s, t)\in\lambda_{\B_+}(s, r)\circ_{\B_+}\lambda_{\B_+}(r, t)$.  But then  
    $1_\B=\lambda_{\B_+}(s, t)\leqslant_{\B}\lambda_{\B_+}(s, r)\cdot_{\B}\lambda_{\B_+}(r, t)=\lambda_{\B_+}(s,r)\cdot_\B 1_\B=\lambda_{\B_+}(s,r)<_\B 1_{\B}$, 
    a contradiction.
    It must therefore  be the case that $\lambda_{\B_+}(r, t)\neq 1_{\B}$.
    Thus, $\lambda_{\B_+}(s, t)\neq 1_\B$ since $s=r$ and 
    $\lambda([w^r],[y^t])=\lambda_{\A_+[\B_+]}(w^r,y^t)
    =\lambda_{\B_+}(r, t)=\lambda_{\B_+}(s, t)=
    \lambda_{\A_+[\B_+]}(x^s,y^t)=\lambda([x^s],[y^t])$.
    
    Case 3: $\lambda_{\A_+[\B_+]}(w^r, x^s)\in I_{\B_+}$ and $\lambda_{\A_+[\B_+]}(x^s, w^r)\in I_{\A_+}$. In this case, the proof of the claim is similar to that of the previous case. 

    Case 4: $\lambda_{\A_+[\B_+]}(w^r, x^s)\in I_{\B_+}$ and $\lambda_{\A_+[\B_+]}(x^s, w^r)\in I_{\B_+}$. Here we have  $\lambda_{\A_+[\B_+]}(w^r, x^s)=\lambda_{\B_+}(r, s)<_\B 1_\B$ and $\lambda_{\A_+[\B_+]}(x^s, w^r)=\lambda_{\B_+}(s, r)<_\B 1_\B$.  Since $\lambda_{\B_+}(r, s)$, $\lambda_{\B_+}(s, r)\in I_{\B_+}$ and
    $\mathcal{N}_{\B_+}$ is
    consistent, it follows from (iii) of Definition~\ref{def:consistent_network_for_frame} that $s=r$.
    As in Case 2, if $\lambda_{\B_+}(r,t)=1_\B$, then $\lambda_{\B_+}(s, t)=1_{\B}$ which leads to a contradiction.
    Hence, $\lambda_{\B_+}(r, t)\neq 1_{\B}$, which implies 
    $\lambda_{\B_+}(s, t)\neq 1_{\B}$, and the claim follows as in Case 2.
\end{proof}

Since $\alpha_{\A_+}$, $\beta_{\A_+}$, $\alpha_{\B_+}$ and $\beta_{\B_+}$ are total functions, we define $\alpha:N\to N$ and $\beta:N\to N$ by setting, for all $[x^s]\in N$,
\begin{equation}\label{eq:alpha_beta_final}
\alpha([x^s])=\left[\left(\alpha_{\A_+}(x)\right)^{\alpha_{\B_+}(s)}\right]\qquad\text{ and }\qquad
\beta([x^s])=\left[\left(\beta_{\A_+}(x)\right)^{\beta_{\B_+}(s)}\right].
\end{equation}

\begin{lemma}\label{lem:alpha_and_beta_well-defined}
The maps $\alpha$ and $\beta$, defined above, are well-defined. That is, for all $[x^s], [y^t] \in N$, we have $\alpha([x^s]) = \alpha([y^t])$ and $\beta([x^s]) = \beta([y^t])$ whenever $[x^s] = [y^t]$. 
\end{lemma}

\begin{proof}
We will prove the claim for $\alpha$. The proof that $\beta$ is well-defined follows similarly. Let $[x^s], [y^t] \in N$ such that $[x^s] = [y^t]$. Then we have $x^s \equiv y^t$, and so $\lambda_{\A_+[\B_+]}(x^s, y^t) \in I$ and $\lambda_{\A_+[B_+]}(y^t, x^s) \in I$. There are four cases, but we consider only one case and leave the other three cases for the reader to check.  

%
%
Assume
$\lambda_{\A_+[\B_+]}(x^s, y^t) \in I_{\A_+}$ and $\lambda_{\A_+[\B_+]}(y^t, x^s) \in I_{\B_+}{\setminus}\{1_\B\}$. Then $\lambda_{\B_+}(s, t) =1_\B$, $\lambda_{\A_+[\B_+]}(x^s, y^t) = \lambda_{\A_+}(x, y)$ and $\lambda_{\B_+}(t, s) <_\B 1_\B$. Since $\mathcal{N}_{\B_+}$ satisfies Definition~\ref{def:consistent_network_for_frame} (iv), the latter implies that $\lambda_{\B_+}(\alpha_{\B_+}(t), \alpha_{\B_+}(s)) \leqslant_\B 1_\B$. Hence, by Lemma~\ref{lem:LambdaEqualToOne} (i), $\lambda_{\B_+}(\alpha_{\B_+}(s), \alpha_{\B_+}(t)) = 1_\B$ and $\lambda_{\B_+}(\alpha_{\B_+}(t), \alpha_{\B_+}(s)) <_\B 1_\B$, which means $$\lambda_{\A_+[\B_+]}\left(\left(\alpha_{\A_+}(x)\right)^{\alpha_{\B_+}(s)}, \left(\alpha_{\A_+}(y)\right)^{\alpha_{\B_+}(t)}\right) = \lambda_{\A_+}(\alpha_{\A_+}(x), \alpha_{\A_+}(y))$$ and $$\lambda_{\A_+[\B_+]}\left(\left(\alpha_{\A_+}(y)\right)^{\alpha_{\B_+}(t)}, \left(\alpha_{\A_+}(x)\right)^{\alpha_{\B_+}(s)}\right) = \lambda_{\B_+}(\alpha_{\B_+}(t), \alpha_{\B_+}(s)) <_\B 1_\B.$$ Since $\mathcal{N}_{\A_+}$ satisfies Definition~\ref{def:consistent_network_for_frame} (iv) and $\lambda_{\A_+}(x, y) \in I_{\A_+}$, we have $\lambda_{\A_+}(\alpha_{\A_+}(x), \alpha_{\A_+}(y)) \in I_{\A_+}$. It follows that $$\lambda_{\A_+[\B_+]}\left(\left(\alpha_{\A_+}(x)\right)^{\alpha_{\B_+}(s)}, \left(\alpha_{\A_+}(y)\right)^{\alpha_{\B_+}(t)}\right) \in I_{\A_+}\subseteq I$$ and $$\lambda_{\A_+[\B_+]}\left(\left(\alpha_{\A_+}(y)\right)^{\alpha_{\B_+}(t)}, \left(\alpha_{\A_+}(x)\right)^{\alpha_{\B_+}(s)}\right) \in I_{\B_+}{\setminus}\{1_\B\} \subseteq I.$$ Hence, $\left(\alpha_{\A_+}(x)\right)^{\alpha_{\B_+}(s)} \equiv \left(\alpha_{\A_+}(y)\right)^{\alpha_{\B_+}(t)}$, and therefore $$\alpha([x^s]) = \left[\left(\alpha_{\A_+}(x)\right)^{\alpha_{\B_+}(s)}\right] = \left[\left(\alpha_{\A_+}(y)\right)^{\alpha_{\B_+}(t)}\right] = \alpha([y^t]),$$
as required. 
\end{proof}


\begin{lemma}\label{lem:WreathProductNetworkConsistent}
Let $N$, $\lambda$, $\alpha$, $\beta$ be defined as above. Then $\mathcal{N} = (N, \lambda, \alpha, \beta)$ is a consistent network for $\A_+[\B_+]$. 
\end{lemma}
\begin{proof}	
    First, we show that $\mathcal{N}$ satisfies items (i) to (ix) of Definition~\ref{def:consistent_network_for_frame}.
    As the interim function $\lambda_{\A_+[\B_+]}$ satisfies (i) by Lemma~\ref{Lem:InterimLemmaOneAndTwo},
    it follows from the definition of $\lambda$
    that it satisfies (i). 
    Similarly, for (ii), $\lambda([x^s],[z^r])=
    \lambda_{\A_+[\B_+]}(x^s,z^r)\in 
    \lambda_{\A_+[\B_+]}(x^s,y^t)\circ \lambda_{\A_+[\B_+]}(y^t,z^r)=\lambda([x^s],[y^t]) \circ \lambda([y^t], [z^r])$ by
    Lemma~\ref{Lem:InterimLemmaOneAndTwo}.

    The use of the quotient structure is needed to prove item (iii): let $[x^s],[y^t]\in N$ and suppose $\lambda([x^s],[y^t])=\lambda_{\A_+[\B_+]}(x^s,y^t)\in I$ and
    $\lambda([y^t],[x^s])=\lambda_{\A_+[\B_+]}(y^t,x^s)\in I$.  Then, $x^s\equiv y^t$
    and $[x^s]=[y^t]$.
    
    For (iv), 
    first let $\lambda([x^s],[y^t])\in I$. 
    Then $\lambda_{\A_+[\B_+]}(x^s,y^t)\in I_{\A_+}$ or $\lambda_{\A_+[\B_+]}(x^s,y^t)\in I_{\B_+}\backslash\{1_\B\}$.  We consider the two cases. 

    Case 1: $\lambda([x^s],[y^t])=\lambda_{\A_+[\B_+]}(x^s,y^t)\in I_{\A_+}$. In this case, $\lambda_{\B_+}(s,t)=1_\B$ so that $\lambda([x^s],[y^t])=\lambda_{\A_+[\B_+]}(x^s,y^t)=\lambda_{\A_+}(x,y)$. Using Lemma~\ref{lem:LambdaEqualToOne}(i), we obtain $\lambda_{\B_+}(\alpha_{\B_+}(s),\alpha_{\B_+}(t))=1_\B$. Since $\mathcal{N}_{\A_+}$ is consistent, we also have $\lambda_{\A_+}(\alpha_{\A_+}(x),\alpha_{\A_+}(y))\in I_{\A_+}$. Hence,
     \begin{align*}
     \lambda(\alpha([x^s]),\alpha([y^t])) & = \lambda\left(\left[\left(\alpha_{\A_+}(x)\right)^{\alpha_{\B_+}(s)}\right], \left[\left(\alpha_{\A_+}(y)\right)^{\alpha_{\B_+}(t)}\right]\right) \\
       & = \lambda_{\A_+[\B_+]}\left(\left(\alpha_{\A_+}(x)\right)^{\alpha_{\B_+}(s)},\left(\alpha_{\A_+}(y)\right)^{\alpha_{\B_+}(t)}\right)\\
       & = \lambda_{\A_+}(\alpha_{\A_+}(x),\alpha_{\A_+}(y))\in I_{\A_+} \subseteq I.
      \end{align*}
     
    Case 2: $\lambda([x^s],[y^t])=\lambda_{\A_+[\B_+]}(x^s,y^t)\in I_{\B_+}{\setminus}\{1_\B\}$. In this case, $\lambda_{\B_+}(s, t) \neq 1_\B$, and so $\lambda_{\B_+}(\alpha_{\B_+}(s), \alpha_{\B_+}(t)) \neq 1_\B$ by Lemma~\ref{lem:LambdaEqualToOne}(i). Moreover, since $\mathcal{N}_{\B_+}$ is a consistent network, we have $\lambda_{\B_+}(\alpha_{\B_+}(s), \alpha_{\B_+}(t)) \in I_{\B_+}$. Thus, 
    \begin{align*}
       \lambda(\alpha([x^s]),\alpha([y^t])) &=
        \lambda_{\A_+[\B_+]}\left(\left(\alpha_{\A_+}(x)\right)^{\alpha_{\B_+}(s)},\left(\alpha_{\A_+}(y)\right)^{\alpha_{\B_+}(t)}\right) \\
        & =        \lambda_{\B_+}(\alpha_{\B_+}(s),\alpha_{\B_+}(t))\in I_{\B_+}{\setminus}\{1_\B\} \subseteq I.
    \end{align*}

    Now let $\lambda(\alpha([x^s]),\alpha([y^t])) \in I$. Then $\lambda_{\A_+[\B_+]}\left(\left(\alpha_{\A_+}(x)\right)^{\alpha_{\B_+}(s)},\left(\alpha_{\A_+}(y)\right)^{\alpha_{\B_+}(t)}\right) \in I_{\A_+}$ or $\lambda_{\A_+[\B_+]}\left(\left(\alpha_{\A_+}(x)\right)^{\alpha_{\B_+}(s)},\left(\alpha_{\A_+}(y)\right)^{\alpha_{\B_+}(t)}\right) \in I_{\B_+}{\setminus}\{1_\B\}$.

    Case 1: $\lambda_{\A_+[\B_+]}\left(\left(\alpha_{\A_+}(x)\right)^{\alpha_{\B_+}(s)},\left(\alpha_{\A_+}(y)\right)^{\alpha_{\B_+}(t)}\right) \in I_{\A_+}$. Here $\lambda_{\B_+}(\alpha_{\B_+}(s), \alpha_{\B_+}(t)) = 1_\B$ so that $\lambda_{\A_+[\B_+]}\left(\left(\alpha_{\A_+}(x)\right)^{\alpha_{\B_+}(s)},\left(\alpha_{\A_+}(y)\right)^{\alpha_{\B_+}(t)}\right) = \lambda_{\A_+}(\alpha_{\A_+}(x), \alpha_{\A_+}(y)) \in I_{\A_+}$. Since $\mathcal{N}_{\A_+}$ is consistent, we have $\lambda_{\A_+}(x, y) \in I_{\A_+}$. Moreover, by Lemma~\ref{lem:LambdaEqualToOne}(i), we have $\lambda_{\B_+}(s, t) = 1_\B$. Hence, $\lambda([x^s], [y^t]) = \lambda_{\A_+}(x, y) \in I_{\A_+} \subseteq I$. 

    Case 2: $\lambda_{\A_+[\B_+]}\left(\left(\alpha_{\A_+}(x)\right)^{\alpha_{\B_+}(s)},\left(\alpha_{\A_+}(y)\right)^{\alpha_{\B_+}(t)}\right) \in I_{\B_+}{\setminus}\{1_\B\}$. In this case, we have $\lambda_{\B_+}(\alpha_{\B_+}(s), \alpha_{\B_+}(t)) \neq 1_\B$ and  $\lambda_{\B_+}(\alpha_{\B_+}(s), \alpha_{\B_+}(t)) \in I_{\B_+}{\setminus}\{1_\B\}$. Since $\mathcal{N}_{\B_+}$ is consistent, it follows that $\lambda_{\B_+}(s, t) \in I_{\B_+}{\setminus}\{1_\B\}$. Moreover, by Lemma~\ref{lem:LambdaEqualToOne}(i), $\lambda_{\B_+}(s, t)\neq 1_\B$. Consequently, $\lambda([x^s], [y^t]) = \lambda_{\B_+}(s, t) \in I_{\B_+}{\setminus}\{1_\B\} \subseteq I$. 

    The proof of (v) is similar to that of (iv). 
    The proof uses Lemma~\ref{lem:LambdaEqualToOne}(ii) and the fact that $\mathcal{N}_{\A_+}$ and $\mathcal{N}_{\B_+}$ satisfy condition (v) of Definition~\ref{def:consistent_network_for_frame}.  The details are left to the reader.



    
    To prove (vi), let $[x^s]\in N$. Then
    $
    	\beta(\beta([x^s]))=\left[
    	\left(\beta_{\A_+}(\beta_{\A_+}(x))\right)^{\beta_{\B_+}(\beta_{\B_+}(s))}\right]=[x^s]
    $    
    since 
    $\mathcal{N}_{\A_+}$ and $\mathcal{N}_{\B_+}$
    are consistent which implies that 
    $\beta_{\A_+}(\beta_{\A_+}(x))=x$ 
    and $\beta_{\B_+}(\beta_{\B_+}(s))=s$.
    
    For (vii) let $[x^s]\in N$.  Then,
    \begin{align*}        \alpha(\beta(\alpha([x^s])))&=
    \left[\left(\alpha_{\A_+}(\beta_{\A_+}(\alpha_{\A_+}(x)))\right)^{\alpha_{\B_+}(\beta_{\B_+}(\alpha_{\B_+}(x)))}\right]
    =\left[\left(\beta_{\A_+}(x)\right)^{\beta_{\B_+}(s)}\right]
    =\beta([x^s])
    \end{align*}   	    
     since 
    $\mathcal{N}_{\A_+}$ and $\mathcal{N}_{\B_+}$
    are consistent 
    and so 
    $\alpha_{\A_+}(\beta_{\A_+}(\alpha_{\A_+}(x)))=\beta_{\A_+}(x)$ 
    and $\alpha_{\B_+}(\beta_{\B_+}(\alpha_{\B_+}(s)))=\beta_{\B_+}(s)$.

    To see that (viii) holds, let $[x^s],[y^t]\in N$.
    If $\lambda_{\B_+}(t,\alpha_{\B_+}(s))=1_\B$, then
    $$\lambda([y^t],\alpha([x^s]))= \lambda_{\A_+[\B_+]}\left(y^t,\left(\alpha_{\A_+}(x)\right)^{\alpha_{\B_+}(s)}\right)=
    \lambda_{\A_+}(y,\alpha_{\A_+}(x))=
    \lambda_{\A_+}(x,y)^{\sim_{\A_+}}$$ 
    since $\mathcal{N}_{\A_+}$ is consistent. Furthermore, from Lemma~\ref{lem:LambdaEqualToOne} (i),
    it follows that $\lambda_{\B_+}(s,t)=1_\B$, and
    hence $\lambda([x^s],[y^t])^\sim=
    \left(\lambda_{\A_+[\B_+]}(x^s,y^t)\right)^\sim=\lambda_{\A_+}(x,y)^{\sim_{\A_+}}$. Hence, $\lambda([y^t],\alpha([x^s])) = \lambda([x^s],[y^t])^\sim$. 
    On the other hand, if $\lambda_{\B_+}(t,\alpha_{\B_+}(s))\neq 1_\B$,
    then
    $$\lambda([y^t],\alpha([x^s])) = \lambda_{\A_+[\B_+]}\left(y^t,\left(\alpha_{\A_+}(x)\right)^{\alpha_{\B_+}(s)}\right)=
    \lambda_{\B_+}(t,\alpha_{\B_+}(s))=
    \lambda_{\B_+}(s,t)^{\sim_{\B_+}}$$
    since $\mathcal{N}_{\B_+}$ is consistent network. By Lemma~\ref{lem:LambdaEqualToOne} (i), 
    $\lambda_{\B_+}(s,t)\neq 1_{\B}$, and
    thus $\lambda([x^s],[y^t])^\sim=
    \left(\lambda_{\A_+[\B_+]}(x^s,y^t)\right)^\sim=\lambda_{\B_+}(s,t)^{\sim_{\B_+}}$. Hence, $\lambda([y^t],\alpha([x^s])) = \lambda([x^s],[y^t])^\sim$. 

    The proof of (ix) is similar to that of (viii), using item (ii) of Lemma~\ref{lem:LambdaEqualToOne}. 
    \end{proof}

To show that the network constructed above is saturated, we need to make an additional assumption about the dual frames of $\A$ and $\B$. We first introduce some terminology. 

We say that an element $j \in B_+{\setminus}\{1_\B\}$ is \emph{domain-deficient} if there exists some $k \in B_+$ such that $k <_\B 1_\B$ and $j = k \cdot_\B j$. It is said to be \emph{range-deficient} if there exists some $\ell \in B_+$ such that $\ell <_\B 1_\B$ and $j = j \cdot_\B \ell$. If $j \in B_+{\setminus}\{1_\B\}$ is both domain- and range-deficient, then we say that it is \emph{doubly deficient}. 

\begin{lemma}\label{lem:domain_range-deficiency}
Let $i \in A_+$ and $j \in \B_+{\setminus}\{1_\B\}$.
\begin{enumerate}[label=\textup{(\roman*)}]
\item If $j \in B_+$ is domain-deficient, then $i \cmax j = \varnothing$. 
\item If $j \in B_+$ is range-deficient, then $j \cmax i = \varnothing$. 
\end{enumerate}
\end{lemma}

\begin{proof}
We prove item (i). The proof of item (ii) is similar. Assume $j \in B_+$ is domain-deficient. Then there is some $\ell \in B_+$ such that $\ell <_\B 1_\B$ and $j = \ell \cdot_\B j$. Suppose, for the sake of a contradiction, there exists some $k \in A_+[B_+]$ such that $k \in i \cmax j$. Then $k \in i \circ j$, and for all $i', j' \in A_+[B_+]$, if $i \preccurlyeq i'$, $j \preccurlyeq j'$ and $k \in i'\circ j'$, then $i = i'$ and $j = j'$. We consider two cases:

Case 1: $k \in A_+$. Then $1_\B \in 1_\B \circ_{\B_+} j$, which means $1_\B\leqslant_\B 1_\B \cdot_B j = j$. Since $j \neq 1_\B$, we have $j \prec_{\B_+} 1_\B$. Hence, $1_\B = 1_{\B}^{-_{\B_+}} \prec_{\B_+} j^{-_{\B_+}}$, and so $i \preccurlyeq j^{-_{\B_+}}$. Now $j^{-_{\B_+}\sim_{\B_+}} = j \in j \circ_{\B_+} 1_\B$, and therefore $1^{-_{\B_+}}_\B \in j^{-_{\B_+}}\circ_{\B_+} j$ by Definition~\ref{def:frame}(v). Since $1_\B^{-_{\B_+}} = 1_\B$, we get  $1_\B \in j^{-_{\B_+}}\circ_{\B_+} j$. Thus, $j^{-_{\B_+}}\circ j = ((j^{-_{\B_+}}\circ_{\B_+} j){\setminus}\{1_\B\}) \cup A_+$, which means $k \in j^{-_{\B_+}}\circ j$. Hence, since $i \preccurlyeq j^{-_{\B_+}}$ and $j \preccurlyeq j$, it must be the case that $i = j^{-_{\B_+}}$, contradicting $A\cap B = \varnothing$. 

Case 2: $k \in B_+{\setminus}\{1_\B\}$. Then $k \in 1_\B \circ_{\B_+} j$, and so $k \leqslant_\B 1_\B \cdot_\B j = j$. Since $1_\B \prec_{\B_+} \ell$, we have $i \preccurlyeq \ell$. Furthermore, $j \leqslant_\B \ell \cdot_\B j$, so $k \leqslant \ell \cdot_\B j$. This shows that $k \in \ell \circ_{\B_+} j$, and moreover $k \in \ell \circ j$. Hence, since $i \preccurlyeq \ell$ and $j \preccurlyeq j$, it must be the case that $i = \ell$, which is a contradiction since $A\cap B = \varnothing$. 
\end{proof}

We are now ready to prove that the network we have constructed for the wreath product  $\A_+[\B_+]$ of $\A_+$ and $\B_+$ is saturated. 

\begin{lemma}\label{lem:WreathProductNetworkSaturated}
If $1_\A$ is the top element of $(A_+, \preccurlyeq_{\A_+})$ or every element of $B_+{\setminus}\{1_\B\}$ is doubly deficient, then $\mathcal{N}= (N, \lambda, \alpha, \beta)$ is saturated. 
\end{lemma}

\begin{proof}
Let $\lambda([x^s], [y^t]) \in i \cmax j$. We have to find some $[z^r] \in N$ such that $\lambda([x^s], [z^r]) = i$ and $\lambda([z^r], [y^t]) = j$. It follows from the assumption that $\lambda([x^s], [y^t]) \in i \circ j$, and for all $i', j' \in A_+[B_+]$, if $i \preccurlyeq i'$, $j\preccurlyeq j'$ and $\lambda([x^s], [y^t]) \in i'\circ j'$, then $i = i'$  and $j=j'$. We consider the following cases:

Case 1: $i, j \in A_+$. In this case, $\lambda([x^s], [y^t]) \in (i\circ_{\A_+} j) \cup (I_{\B_+}{\setminus}\{1_\B\})$. Suppose  $\lambda([x^s], [y^t]) \in I_{\B_+}{\setminus}\{1_\B\}$. Then $\lambda([x^s], [y^t]) = \lambda_{\A_+[\B_+]}(x^s,y^t) = \lambda_{\B_+}(s, t) <_\B 1_\B$, and so $1_\B \prec_{\B_+} \lambda_{\B_+}(s, t)$.  This implies $i \preccurlyeq \lambda_{\B_+}(s, t)$. Since $\lambda_{\B_+}(s, t) \leqslant_\B \lambda_{\B_+}(s, t)\cdot_\B 1_\B$, we also have $\lambda([x^s], [y^t]) = \lambda_{\B_+}(s, t) \in \lambda_{\B_+}(s, t)\circ_{\B_+} 1_\B$. Hence, $\lambda([x^s], [y^t]) \in \lambda_{\B_+}(s, t) \circ j$, and so, since $j \preccurlyeq j$, we can conclude that $i =\lambda_{\B_+}(s, t)$, which is a contradiction since $A\cap B = \varnothing$. Consequently, it must be the case that $\lambda([x^s], [y^t]) \in i \circ_{\A_+} j$. Hence, $\lambda_{\B_+}(s, t) = 1_\B$ so that $\lambda([x^s], [y^t]) = \lambda_{\A_+[\B_+]}(x^s, y^t) = \lambda_{\A_+}(x, y)$. 
    
We claim that $\lambda_{\A_+}(x, y) \in i \circ^{\text{max}}_{\A_+} j$. First, note that $\lambda_{\A_+}(x, y) \in i \circ_{\A_+} j$. Now let $k, \ell \in A_+$ such that $i\preccurlyeq_{\A_+} k$, $j \preccurlyeq_{\A_+} \ell$ and $\lambda_{\A_+}(x, y) \in k \circ_{\A_+} \ell$. Then we have $i\preccurlyeq k$, $j \preccurlyeq \ell$ and $\lambda([x^s], [y^t]) \in k \circ \ell$, and therefore $i = k$ and $j = \ell$. This shows $\lambda_{\A_+}(x, y) \in i \circ^{\text{max}}_{\A_+} j$.

Since $\A_+$ is saturated, there is some $z \in N_{\A_+}$ such that $\lambda_{\A_+}(x, z) = i$ and $\lambda_{\A_+}(z, y) = j$. Recall that $\lambda_{\B_+}(s, t) = 1_\B$, so $\lambda([x^s], [z^t]) = \lambda_{\A_+}(x, z) = i$.  

We now show that $\lambda([z^t], [y^t]) = j$. Recall that $\lambda_{\B_+}(t, t) \leqslant_\B 1_\B$. Since $\lambda_{\B_+}(s, t) \in \lambda_{\B_+}(s, t)\circ_{\B_+} \lambda_{\B_+}(t, t)$, if $\lambda_{\B_+}(t, t) <_\B 1_\B$, we have $1_\B = \lambda_{\B_+}(s, t) \leqslant_\B \lambda_{\B_+}(s, t)\cdot_\B \lambda_{\B_+}(t, t) = 1_\B \cdot_\B \lambda_{\B_+}(t, t) =\lambda_{\B_+}(t, t) <_\B 1_\B$, a contradiction. Hence, it must be the case that $\lambda_{\B_+}(t, t) = 1_\B$. This shows that $\lambda([z^t], [y^t]) = \lambda_{\A_+[\B_+]}(z^t, y^t) = \lambda_{\A_+}(z, y) = j$. 

Case 2: $i \in A_+$ and $j \in B_+{\setminus}\{1_\B\}$. Assume first that $\lambda_{\B_+}(s, t) = 1_\B$. Then $\lambda([x^s], [y^t]) = \lambda_{\A_+}(x, y)$. Consequently, it must be the case that $1_\B \in 1_\B \circ_{\B_+} j$ so that $i \circ j = \left.((1_\B \circ_{\B_+} j){\setminus}\{1_\B\}) \cup A_+\right.$. Then $1_\B \in 1_\B \circ_{\B_+} j$, and so $1_\B <_\B j$, which means $j \prec_{\B_+} 1_\B$. Hence, as in Case 1 of the proof of Lemma~\ref{lem:domain_range-deficiency}, we can show that $i = j^{-_{\B_+}}$, which is a contradiction. We can thus conclude that $\lambda_{\B_+}(s, t) \neq 1_\B$. In that case $\lambda([x^s], [y^t]) = \lambda_{\B_+}(s, t) \in 1_\B \circ_{\B_+} j$, and thus $\lambda_{\B_+}(s, t) \leqslant_\B 1_\B \cdot_\B j = j$ or, equivalently, $j \preccurlyeq_{\B_+} \lambda_{\B_+}(s, t)$.

Since $i \cmax j\neq \varnothing$, by Lemma~\ref{lem:domain_range-deficiency}, $j$ is not domain-deficient, and therefore not doubly-deficient. This implies $1_\A$ must be the top element of $(A_+, \preccurlyeq_{\A_+})$, which means $i \preccurlyeq_{\A_+} 1_\A$. 

Now $\lambda_{\B_+}(s, t) \leqslant_\B 1_\B \cdot_\B \lambda_{\B_+}(s, t)$, so we have $\lambda_{\B_+}(s, t) \in 1_\B \circ_{\B_+}\lambda_{\B_+}(s, t)$, and this shows $\lambda([x^s], [y^t]) \in 1_\A \circ \lambda_{\B_+}(s, t)$. Hence, since $i \preccurlyeq 1_\A$ and $j \preccurlyeq \lambda_{\B_+}(s, t)$, it follows that $i = 1_\A$ and $j = \lambda_{\B_+}(s, t)$. The latter implies that $\lambda([x^s], [y^t]) = j$, so if we can show that $\lambda([x^s], [x^s]) = i$, we would be done. To prove this, we have to show that $\lambda_{\B_+}(s, s)=1_\B$ and $\lambda_{\A_+}(x, x) = i$. 

We know that $\lambda_{\B_+}(s, s) \leqslant_{\B} 1_\B$. Now suppose for the sake of a contradiction that $\lambda_{\B_+}(s, s) <_\B 1_\B$. Then $1_\B \prec_{\B_+} \lambda_{\B_+}(s, s)$, and so $i \preccurlyeq \lambda_{\B_+}(s, s)$. Moreover, $\lambda_{\B_+}(s, t) \in \lambda_{\B_+}(s, s) \circ_{\B_+} \lambda_{\B_+}(s, t) = \lambda_{\B_+}(s, s) \circ_{\B_+} j$, which means $\lambda([x^s], [y^t]) \in \lambda_{\B_+}(s, s) \circ j$. Hence, it must be the case that $i = \lambda_{\B_+}(s, s)$, a contradiction since $A \cap B = \varnothing$. 

We also know that $\lambda_{\A_+}(x, x) \leqslant_\A 1_\A$, and thus $1_\A \preccurlyeq_{\A_+} \lambda_{\A_+}(x, x)$. But $\lambda_{\A_+}(x, x) \preccurlyeq_{\A_+} 1_\A$, so $\lambda_{\A_+}(x, x) = 1_\A = i$, as required. 

Case 3: $i \in \B_+{\setminus}\{1_\B\}$ and $j \in A_+$. The proof of this case is similar to that of Case 2. Here we use the fact that $i \cmax j \neq \varnothing$ implies that $i$ is not range-deficient. 

Case 4: $i, j \in B_+{\setminus}\{1_\B\}$. Assume first that $\lambda_{\B_+}(s, t) = 1_\B$. Then $\lambda([x^s], [y^t]) = \lambda_{\A_+}(x, y)$. Hence, it must be the case that $1_\B \in i \circ_{\B_+} j$ so that $i \circ j = ((i \circ_{\B_+} j){\setminus}\{1_\B\}) \cup A_+$. Hence, $\lambda_{\B_+}(s, t) \in i \circ_{\B_+} j$. Now let $k, \ell \in B_+$ such that $i \preccurlyeq_{\B_+} k$, $j \preccurlyeq_{\B_+} \ell$ and $\lambda_{\B_+}(s, t) \in k \circ_{\B_+} \ell$. Thus, $1_\B \in k \circ_{\B_+} \ell$, and so $k \circ \ell = ((k \circ_{\B_+} \ell){\setminus}\{1_\B\}) \cup A_+$, which means $\lambda([x^s], [y^t]) = \lambda_{\A_+}(x, y) \in k \circ \ell$. Moreover, $i \preccurlyeq_{\B_+} k$ and  $j \preccurlyeq_{\B_+} \ell$ imply $i \preccurlyeq k$ and $j \preccurlyeq \ell$, so we can conclude that $i = k$ and $j = \ell$. This shows that $\lambda_{\B_+}(s, t) \in i \circ^{\text{max}}_{\B_+} j$. Using the fact that $\B_+$ is saturated, we can conclude that there exists some $r \in N_{\B_+}$ such that $\lambda_{\B_+}(s, r) = i$ and $\lambda_{\B_+}(r, t) = j$. Since $i\neq 1_\B$ and $j \neq 1_\B$, we have $\lambda([x^s], [x^r]) = \lambda_{\B_+}(s, r) = i$ and $\lambda([x^r], [y^t]) = \lambda_{\B_+}(r, t) = j$.   

Now assume $\lambda_{\B_+}(s, t) \neq 1_\B$. Then $\lambda([x^s], [y^t]) = \lambda_{\B_+}(s, t)$, which means $\lambda_{\B_+}(s, t) \in i \circ_{\B_+} j$. Now let $k, \ell \in B_+$ such that $i \preccurlyeq_{\B_+} k$, $j \preccurlyeq_{\B_+} \ell$ and $\lambda_{\B_+}(s, t) \in k \circ_{\B_+} \ell$. Hence, $i \preccurlyeq k$, $j \preccurlyeq \ell$ and $\lambda([x^s], [y^t]) \in k \circ \ell$, so we can conclude that $i = k$ and $j = \ell$. This shows that $\lambda_{\B_+}(s, t) \in i \circ^{\text{max}}_{\B_+} j$. The proof now proceeds as above. 
\end{proof}

\begin{proposition}\label{prop:existence_network_wreath_product}
Let $\A$ and $\B$ be finite DqRAs such that $A\cap B=\varnothing$, $1_\A\in A_+$, $1_\B\in B_+$ and $1_{\B}^{-_{\B_+}}=1_{\B}^{\sim_{\B_+}}= 1_\B^{\neg_{\B_+}} = 1_{\B}$. Suppose $\A_+$ and $\B_+$ have the network realisation property, and that either $1_\A$ is the top element of $(A_+, \preccurlyeq_{\A_+})$ or every element of $B_+{\setminus}\{1_\B\}$ is doubly deficient. Then $\A_+[\B_+]$ has the network realisation property.
\end{proposition}

\begin{proof}
Assume $\A_+$ and $\B_+$ have the network realisation property. Let $k\in A_+[B_+]$. We consider two cases. 

Case 1: $k \in A_+$. Since $\A_+$ has the network realisation property, there exists a consistent saturated network $\mathcal{N}_{\A_+} = (N_{\A_+}, \lambda_{\A_+}, \alpha_{\A_+}, \beta_{\A_+})$ for $\A_+$ and $w,z \in N_{\A_+}$ such that $\lambda_{\A_+}(w,z) = k$. Similarly, since $\B_+$ has the network realisation property, there exists a consistent saturated network $\mathcal{N}_{\B_+} = (N_{\B_+}, \lambda_{\B_+}, \alpha_{\B_+}, \beta_{\B_+})$ for $\B_+$  and $u, v \in N_{\B_+}$ such that $\lambda_{\B_+}(u, v) = 1_\B$. 

For each $s\in N_{\B_+}$, let $_sN_{\A_+}$ be a copy of $N_{\A_+}$ and
let $x^s$ denote the copy of $x$ in $_sN_{\A_+}$.
Set $$N_{\A_+[\B_+]}=\displaystyle\bigcup_{s\in N_{\B_+} }{}_sN_{\A_+}.$$  Define  $\lambda_{\A_+[\B_+]}:N_{\A_+[\B_+]}\times N_{\A_+[\B_+]}\to A_+[B_+]$ by setting, for  all $x^s,y^t \in N_{\A_+[\B_+]}$,
\[
\lambda_{\A_+[\B_+]}\left(x^s,y^t\right)=
\begin{cases}
   \lambda_{\A_+}(x,y)&\text{ if }\lambda_{\B_+}(s,t)=1_{\B}\\
   \lambda_{\B_+}(s,t)&\text{ if }\lambda_{\B_+}(s,t)\neq 1_{\B}.
\end{cases}
\]

Next define a binary relation ${\equiv} \subseteq \left(\bigcup_{s\in N_{\B_+}} {}_sN_{\A_+}\right)\times \left(\bigcup_{s\in N_{\B_+}} {}_sN_{\A_+}\right)$ by  
\[x^s\equiv y^t\quad \textnormal{iff} \quad 
\lambda_{\A_+[\B_+]}(x^s,y^t)\in I\quad\text{and}\quad
\lambda_{\A_+[\B_+]}(y^t,x^s)\in I.\]
By Lemma~\ref{lem:EquivRel}, $\equiv$ is an equivalence relation on $N_{\A_+[\B_+]}$. 
Now set
$\displaystyle N=N_{\A_+[\B_+]}/{\equiv}$ and
define $\lambda:N\times N\to A_+[B_+]$ by setting, for all $[x^s],[y^t]\in N$, \[\lambda([x^s],[y^t])=\lambda_{\A_+[\B_+]}(x^s,y^t).\]
In Lemma~\ref{lem:lambdaWellDefined} we showed that $\lambda$ is well-defined. 
Finally, define $\alpha:N\to N$ and $\beta:N\to N$ by setting, for all $[x^s]\in N$,
\[\alpha([x^s])=\left[\left(\alpha_{\A_+}(x)\right)^{\alpha_{\B_+}(s)}\right]\qquad\text{ and }\qquad
\beta([x^s])=\left[\left(\beta_{\A_+}(x)\right)^{\beta_{\B_+}(s)}\right].\]
It follows from Lemma~\ref{lem:WreathProductNetworkConsistent} that $\mathcal{N} = (N, \lambda, \alpha, \beta)$ is a consistent network for $\A_+[\B_+]$, and by Lemma~\ref{lem:WreathProductNetworkSaturated}, $\mathcal{N}$ is saturated. Moreover, $\lambda_{\B_+}(u, v) = 1_\B$, so $\lambda([w^u], [z^v]) = \lambda_{\A_+[\B_+]}(w^u, z^v) = \lambda_{\A_+}(w, z) = k$.

Case 2: $k \in B_+{\setminus}\{1_\B\}$. The proof of this case is similar to that of Case 1. Here we use the fact that there exist consistent saturated networks
$\mathcal{N}_{\A_+}=(N_{\A_+},\lambda_{\A_+},\alpha_{\A_+},\beta_{\A_+})$ and $\mathcal{N}_{\B_+}=(N_{\B_+},\lambda_{\B_+},\alpha_{\B_+},\beta_{\B_+})$ for $\A_+$ and $\B_+$, respectively, together with $w,z\in N_{\A_+}$ and $u,v\in N_{\B_+}$ such that
$\lambda_{\A_+}(w,z)=1_{\A}$ and $\lambda_{\B_+}(u,v)=k$.
\end{proof}

Using the construction outlined above, we can construct consistent saturated networks for the wreath products in Examples~\ref{ex:wreath_product_example_1} and~\ref{ex:wreath_product_example_2}.

\begin{example}\label{ex:wreath_products_networks_ex1}
Consider the consistent saturated network $\mathcal{N}_{\A_+}$ in Example~\ref{ex:network_for_(S3)_+} for the dual frame $\A_+$ of $\A$ in Figure~\ref{fig:small_examples_DqRAs}, and  the network $\mathcal{N}_{\B_+}$ in Example~\ref{ex:network_for_diamond_1=0^2} for the dual frame $\B_+$ of $\B$.
Then $N_{\A_+[\B_+]} = \{x^u, y^u, x^v, y^v\}$ and the map $\lambda_{\A_+[\B_+]}: N_{A_+[B_+]} \times N_{A_+[B_+]} \to A_+[B_+]$ is given by $\lambda_{\A_+[\B_+]}(x^u, x^u) = \lambda_{\A_+[\B_+]}(y^u, y^u) = \lambda_{\A_+[\B_+]}(x^v, x^v) = \lambda_{\A_+[\B_+]}(y^v, y^v) = 1_\A$, $\lambda_{\A_+[\B_+]}(x^u, y^u) =  \lambda_{\A_+[\B_+]}(y^u, x^u) = \lambda_{\A_+[\B_+]}(x^v, y^v) = \lambda_{\A_+[\B_+]}(y^v, x^v) =  \top_\A$ and $\lambda_{\A_+[\B_+]}(x^u, x^v) = \lambda_{\A_+[\B_+]}(x^u, y^v) = \lambda_{\A_+[\B_+]}(y^u, y^v) = \lambda_{\A_+[\B_+]}(y^u, x^v) = \lambda_{\A_+[\B_+]}(y^v, y^u) = \lambda_{\A_+[\B_+]}(y^v, x^u) = \lambda_{\A_+[\B_+]}(x^v, x^u) = \lambda_{\A_+[\B_+]}(x^v, y^u) = 0_\B$. 
 In this case $\equiv$ is the identity relation on $N_{\A_+[\B_+]}$, so $N= N_{\A_+[\B_+]}/{\equiv} =\{[x^u], [y^u], [x^v], [y^v]\}$. Now define $\lambda: N \times N \to A_+[B_+]$ as in (\ref{eq:lambda_final}). From  (\ref{eq:alpha_beta_final}), we see that $\alpha: N \to N$ is given by 
 $\alpha([x^u]) = [y^u]$, $\alpha([y^u]) = [x^u]$, $\alpha([x^v]) = [y^v]$ and $\alpha([y^v]) = [x^v]$, while $\beta: N \to N$ is given by $\beta = \textnormal{id}_{N}$. 
By Lemma~\ref{lem:WreathProductNetworkConsistent}, $(N, \lambda, \alpha, \beta)$ is a consistent network for $\A_+[\B_+]$, and since $1_\A$ is the top element of $(A_+, \preccurlyeq_{\A_+})$, it follows from Lemma~\ref{lem:WreathProductNetworkSaturated} that it is saturated. 
Moreover, $\lambda$ is surjective, so the wreath product $\A_+[\B_+]$ has the network realisation property. This network is depicted in Figure~\ref{fig:network_for_wreath_product_ex1}.
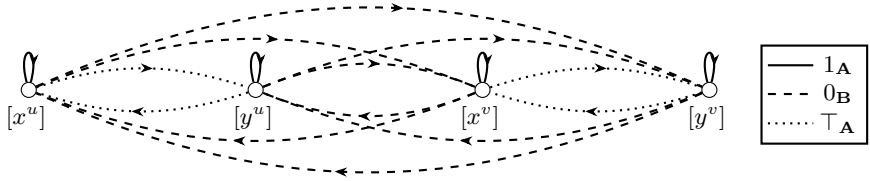
\begin{figure}[ht]
\centering
\begin{tikzpicture}[scale=1.5,pics/sample/.style={code={\draw[#1] (0,0) --(0.6,0) ;}},
Dotted/.style={
dash pattern=on 0.1\pgflinewidth off #1\pgflinewidth,line cap=round,
shorten >=#1\pgflinewidth/2,shorten <=#1\pgflinewidth/2},
Dotted/.default=3]
\footnotesize
\node[draw,circle,inner sep=2pt] (1) at (-2.5,0.7) {};
\node[draw,circle,inner sep=2pt] (2) at (-0.5,0.7) {};
\node[draw,circle,inner sep=2pt] (3) at (1.5,0.7) {};
\node[draw,circle,inner sep=2pt] (4) at (3.5,0.7) {};
\path (1) edge [bend left=18, dotted, thick, postaction={decorate},
decoration={
markings,
mark=at position 0.55 with {\arrow{Stealth[scale=0.7]}}
}] node {} (2);
\path (3) edge [bend left=18, dotted, thick, postaction={decorate},
decoration={
markings,
mark=at position 0.55 with {\arrow{Stealth[scale=0.7]}}
}] node {} (4);
\path (2) edge [bend left=18, dotted, thick, postaction={decorate},
decoration={
markings,
mark=at position 0.55 with {\arrow{Stealth[scale=0.7]}}
}] node {} (1);
\path (4) edge [bend left=18, dotted, thick, postaction={decorate},
decoration={
markings,
mark=at position 0.55 with {\arrow{Stealth[scale=0.7]}}
}] node {} (3);
\path (1) edge [bend left=22, dashed, thick, postaction={decorate},
decoration={
markings,
mark=at position 0.55 with {\arrow{Stealth[scale=0.7]}}
}] node {} (3);
\path (3) edge [bend left=22, dashed, thick, postaction={decorate},
decoration={
markings,
mark=at position 0.55 with {\arrow{Stealth[scale=0.7]}}
}] node {} (1);
\path (1) edge [bend left=24, dashed, thick, postaction={decorate},
decoration={
markings,
mark=at position 0.55 with {\arrow{Stealth[scale=0.7]}}
}] node {} (4);
\path (4) edge [bend left=24, dashed, thick, postaction={decorate},
decoration={
markings,
mark=at position 0.55 with {\arrow{Stealth[scale=0.7]}}
}] node {} (1);
\path (2) edge [bend left=22, dashed, thick, postaction={decorate},
decoration={
markings,
mark=at position 0.55 with {\arrow{Stealth[scale=0.7]}}
}] node {} (3);
\path (3) edge [bend left=22, dashed, thick, postaction={decorate},
decoration={
markings,
mark=at position 0.55 with {\arrow{Stealth[scale=0.7]}}
}] node {} (2);
\path (2) edge [bend left=22, dashed, thick, postaction={decorate},
decoration={
markings,
mark=at position 0.55 with {\arrow{Stealth[scale=0.7]}}
}] node {} (4);
\path (4) edge [bend left=22, dashed, thick, postaction={decorate},
decoration={
markings,
mark=at position 0.55 with {\arrow{Stealth[scale=0.7]}}
}] node {} (2);
\draw [thick] (1) edge[loop above,looseness=25,postaction={decorate},
decoration={
markings,
mark=at position 0.75 with {\arrow{Stealth[scale=0.7]}}
}]node{} (1); 
\draw [thick] (2) edge[loop above,looseness=25,postaction={decorate},
decoration={
markings,
mark=at position 0.75 with {\arrow{Stealth[scale=0.7]}}
}]node{} (2); 
\draw [thick] (3) edge[loop above,looseness=25,postaction={decorate},
decoration={
markings,
mark=at position 0.75 with {\arrow{Stealth[scale=0.7]}}
}]node{} (3); 
\draw [thick] (4) edge[loop above,looseness=25,postaction={decorate},
decoration={
markings,
mark=at position 0.75 with {\arrow{Stealth[scale=0.7]}}
}]node{} (4); 
\node[label,anchor=north,xshift=-1pt] at (1) {$[x^u]$};
\node[label,anchor=north,xshift=-1pt] at (2) {$[y^u]$};
\node[label,anchor=north,xshift=-1pt] at (3) {$[x^v]$};
\node[label,anchor=north,xshift=-1pt] at (4) {$[y^v]$};
\path (4.9,1.1) 
node[matrix,anchor=north east,draw,nodes={anchor=center},inner sep=2pt, thick]  {
\pic{sample=solid}; & \node{$1_\A$}; \\
\pic{sample=dashed}; & \node{$0_\B$}; \\
\pic{sample=dotted}; & \node{$\top_\A$}; \\
};
\normalsize
\end{tikzpicture}
\caption{A consistent saturated network for the wreath product $\A_+[\B_+]$ in Figure~\ref{fig:wreath_product_example_1}.}\label{fig:network_for_wreath_product_ex1}
\end{figure}
\end{example}

\begin{example}\label{ex:wreath_products_networks_ex2}
Consider the consistent saturated network $\mathcal{N}_{\C_+}$ in Example~\ref{ex:network_for_diamond_with_a} for the dual frame $\C_+$ of $\C$ in Figure~\ref{fig:small_examples_DqRAs}, and  the network $\mathcal{N}_{\D_+}$ in Example~\ref{ex:network_for_(S4)_+} for the dual frame $\D_+$ of $\D$.
Then $N_{\C_+[\D_+]} = \{x^s, y^s, x^t, y^t, x^u, y^u, x^v, y^v, x^w, y^w\}$ and $\lambda_{\C_+[\D_+]}: N_{C_+[D_+]} \times N_{C_+[D_+]} \to C_+[D_+]$ is given by the table below.

\begin{center}
$\begin{array}{|c|cccccccccc|}\hline
\lambda_{\C_+[\D_+]} & x^s &y^s& x^t & y^t & x^u & y^u & x^v & y^v & x^w & y^w \\\hline
x^s & 0_\D & 0_\D & \top_\D & \top_\D & 0_\D & 0_\D & 0_\D & 0_\D & 0_\D & 0_\D	\\
y^s & 0_\D & 0_\D & \top_\D & \top_\D & 0_\D & 0_\D & 0_\D & 0_\D & 0_\D & 0_\D \\
 x^t & \top_\D & \top_\D & 0_\D & 0_\D & 0_\D & 0_\D & 0_\D & 0_\D & 0_\D & 0_\D \\
 y^t & \top_\D & \top_\D & 0_\D & 0_\D & 0_\D & 0_\D & 0_\D & 0_\D & 0_\D & 0_\D\\
 x^u & \top_\D & \top_\D & \top_\D & \top_\D & 1_\C & a & 0_\D & 0_\D & 0_\D & 0_\D \\
 y^u & \top_\D & \top_\D & \top_\D & \top_\D & a & 1_\C & 0_\D & 0_\D & 0_\D & 0_\D \\
 x^v & \top_\D & \top_\D & \top_\D & \top_\D & \top_\D & \top_\D & 0_\D & 0_\D & \top_\D & \top_\D  \\
 y^v & \top_\D & \top_\D & \top_\D & \top_\D & \top_\D & \top_\D & 0_\D & 0_\D & \top_\D & \top_\D \\
 x^w & \top_\D & \top_\D & \top_\D & \top_\D & \top_\D & \top_\D & \top_\D & \top_\D & 0_\D & 0_\D \\
  y^w & \top_\D & \top_\D & \top_\D & \top_\D & \top_\D & \top_\D & \top_\D & \top_\D & 0_\D & 0_\D \\
 \hline
\end{array}$
\end{center}

\vskip 0.1cm 

\noindent
Here ${\equiv} = \textnormal{id}_{N_{\C_+[\D_+]}} \cup \{(x^s, y^s), (y^s, x^s), (x^t, y^t), (y^t, x^t), (x^v, y^v), (y^v, x^v), (x^w, y^w), (y^w, x^w)\}$, so $N = N_{\C_+[\D_+]}/{\equiv} = \{[x^s], [x^t], [x^u], [y^u], [x^v], [x^w] \}$. Now define $\lambda: N \times  N\to  C_+[D_+]$ as in~(\ref{eq:lambda_final}). From  (\ref{eq:alpha_beta_final}), we see that $\alpha: N \to N$ is given by 
 $\alpha([x^s]) = [x^t]$, $\alpha([x^t]) = [x^s]$, $\alpha([x^u]) = [y^u]$, $\alpha([y^u]) = [x^u]$, $\alpha([x^v]) = [x^w]$ and $\alpha([x^w]) = [x^v]$, while $\beta: N \to N$ is given by $\beta([x^s]) = [x^v]$, $\beta([x^v]) = [x^s]$, $\beta([x^t]) = [x^w]$, $\beta([x^w]) = [x^t]$, $\beta([x^u]) = [x^u]$ and $\beta([y^u]) = [y^u]$. 
By Lemma~\ref{lem:WreathProductNetworkConsistent}, $(N, \lambda, \alpha, \beta)$ is a consistent network for $\C_+[\D_+]$, and since every element of $D_+{\setminus}\{1_\D\}$ is doubly deficient, it follows from Lemma~\ref{lem:WreathProductNetworkSaturated} that it is saturated. 
As in the previous example, $\lambda$ is surjective, so the wreath product $\C_+[\D_+]$ has the network realisation property. This network is shown in Figure~\ref{fig:network_for_wreath_product_ex2}.

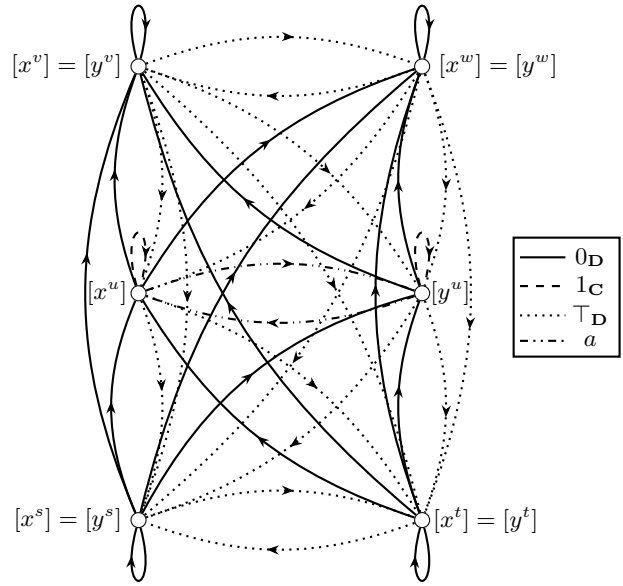
\begin{figure}[ht]
\centering
\begin{tikzpicture}[scale=1.5,pics/sample/.style={code={\draw[#1] (0,0) --(0.6,0) ;}},
Dotted/.style={
dash pattern=on 0.1\pgflinewidth off #1\pgflinewidth,line cap=round,
shorten >=#1\pgflinewidth/2,shorten <=#1\pgflinewidth/2},
Dotted/.default=3]
\footnotesize
\node[draw,circle,inner sep=2pt] (1) at (-1.25,2) {};
\node[draw,circle,inner sep=2pt] (2) at (1.25,2) {};
\node[draw,circle,inner sep=2pt] (3) at (-1.25,0) {};
\node[draw,circle,inner sep=2pt] (4) at (1.25,0) {};
\node[draw,circle,inner sep=2pt] (5) at (-1.25,-2) {};
\node[draw,circle,inner sep=2pt] (6) at (1.25,-2) {};
\node[label,anchor=east,xshift=0pt,yshift=0pt] at (1) {$[x^v]=[y^v]$};
\node[label,anchor=west,yshift=0pt, xshift=0pt] at (2) {$[x^w]=[y^w]$};
\node[label,anchor=east,xshift=2pt,yshift=0pt] at (3) {$[x^u]$};
\node[label,anchor=west,xshift=-3pt,yshift=0pt] at (4) {$[y^u]$};
\node[label,anchor=east,xshift=0pt] at (5) {$[x^s]=[y^s]$};
\node[label,anchor=west,xshift=-2pt] at (6) {$[x^t]=[y^t]$};
\draw [thick] (1) edge[loop above,looseness=45,postaction={decorate},
decoration={
markings,
mark=at position 0.75 with {\arrow{Stealth[scale=0.7]}}
}]node{} (1); 
\draw [thick] (2) edge[loop above,looseness=45, postaction={decorate},
decoration={
markings,
mark=at position 0.75 with {\arrow{Stealth[scale=0.7]}}
}]node{} (2); 
\draw [dashed,thick] (3) edge[loop above,looseness=45, postaction={decorate},
decoration={
markings,
mark=at position 0.75 with {\arrow{Stealth[scale=0.7]}}
}]node{} (3); 
\draw [dashed,thick] (4) edge[loop above,looseness=45, postaction={decorate},
decoration={
markings,
mark=at position 0.75 with {\arrow{Stealth[scale=0.7]}}
}]node{} (4);
\draw [thick] (5) edge[loop below,looseness=45, postaction={decorate},
decoration={
markings,
mark=at position 0.75 with {\arrow{Stealth[scale=0.7]}}
}]node{} (5);
\draw [thick] (6) edge[loop below,looseness=45, postaction={decorate},
decoration={
markings,
mark=at position 0.75 with {\arrow{Stealth[scale=0.7]}}
}]node{} (6);
\path (1) edge [bend left=20, dotted, thick, postaction={decorate},
decoration={
markings,
mark=at position 0.6 with {\arrow{Stealth[scale=0.7]}}
}] node {} (3);
\path (1) edge [bend left=21, dotted, thick, postaction={decorate},
decoration={
markings,
mark=at position 0.52 with {\arrow{Stealth[scale=0.7]}}
}] node {} (5);
\path (1) edge [bend left=20, dotted, thick, postaction={decorate},
decoration={
markings,
mark=at position 0.6 with {\arrow{Stealth[scale=0.7]}}
}] node {} (4);
\path (1) edge [bend left=18, dotted, thick, postaction={decorate},
decoration={
markings,
mark=at position 0.55 with {\arrow{Stealth[scale=0.7]}}
}] node {} (6);
\path (1) edge [bend left=20, dotted, thick, postaction={decorate},
decoration={
markings,
mark=at position 0.55 with {\arrow{Stealth[scale=0.7]}}
}] node {} (2);
\path (2) edge [bend left=20, dotted, thick, postaction={decorate},
decoration={
markings,
mark=at position 0.5 with {\arrow{Stealth[scale=0.7]}}
}] node {} (3);
\path (2) edge [bend left=20, dotted, thick, postaction={decorate},
decoration={
markings,
mark=at position 0.55 with {\arrow{Stealth[scale=0.7]}}
}] node {} (4);
\path (2) edge [bend left=18, dotted, thick, postaction={decorate},
decoration={
markings,
mark=at position 0.6 with {\arrow{Stealth[scale=0.7]}}
}] node {} (5);
\path (2) edge [bend left=21, dotted, thick, postaction={decorate},
decoration={
markings,
mark=at position 0.58 with {\arrow{Stealth[scale=0.7]}}
}] node {} (6);
\path (2) edge [bend left=20, dotted, thick, postaction={decorate},
decoration={
markings,
mark=at position 0.55 with {\arrow{Stealth[scale=0.7]}}
}] node {} (1);
\path (3) edge [bend left=20, thick, postaction={decorate},
decoration={
markings,
mark=at position 0.55 with {\arrow{Stealth[scale=0.7]}}
}] node {} (1);
\path (3) edge [bend left=20, thick, postaction={decorate},
decoration={
markings,
mark=at position 0.55 with {\arrow{Stealth[scale=0.7]}}
}] node {} (2);
\path (3) edge [bend left=20, dashdotdotted, thick, postaction={decorate},
decoration={
markings,
mark=at position 0.55 with {\arrow{Stealth[scale=0.7]}}
}] node {} (4);
\path (3) edge [bend left=20, dotted, thick, postaction={decorate},
decoration={
markings,
mark=at position 0.43 with {\arrow{Stealth[scale=0.7]}}
}] node {} (5);
\path (3) edge [bend left=20, dotted, thick, postaction={decorate},
decoration={
markings,
mark=at position 0.55 with {\arrow{Stealth[scale=0.7]}}
}] node {} (6);
\path (4) edge [bend left=20, thick, postaction={decorate},
decoration={
markings,
mark=at position 0.55 with {\arrow{Stealth[scale=0.7]}}
}] node {} (1);
\path (4) edge [bend left=20, thick, postaction={decorate},
decoration={
markings,
mark=at position 0.5 with {\arrow{Stealth[scale=0.7]}}
}] node {} (2);
\path (4) edge [bend left=20, dashdotdotted, thick, postaction={decorate},
decoration={
markings,
mark=at position 0.55 with {\arrow{Stealth[scale=0.7]}}
}] node {} (3);
\path (4) edge [bend left=20, dotted, thick, postaction={decorate},
decoration={
markings,
mark=at position 0.55 with {\arrow{Stealth[scale=0.7]}}
}] node {} (5);
\path (4) edge [bend left=20, dotted, thick, postaction={decorate},
decoration={
markings,
mark=at position 0.55 with {\arrow{Stealth[scale=0.7]}}
}] node {} (6);
\path (5) edge [bend left=23, thick, postaction={decorate},
decoration={
markings,
mark=at position 0.59 with {\arrow{Stealth[scale=0.7]}}
}] node {} (1);
\path (5) edge [bend left=18, thick, postaction={decorate},
decoration={
markings,
mark=at position 0.49 with {\arrow{Stealth[scale=0.7]}}
}] node {} (2);
\path (5) edge [bend left=23, thick, postaction={decorate},
decoration={
markings,
mark=at position 0.525 with {\arrow{Stealth[scale=0.7]}}
}] node {} (3);
\path (5) edge [bend left=23, thick, postaction={decorate},
decoration={
markings,
mark=at position 0.5 with {\arrow{Stealth[scale=0.7]}}
}] node {} (4);
\path (5) edge [bend left=20, dotted, thick, postaction={decorate},
decoration={
markings,
mark=at position 0.55 with {\arrow{Stealth[scale=0.7]}}
}] node {} (6);
\path (6) edge [bend left=18, thick, postaction={decorate},
decoration={
markings,
mark=at position 0.54 with {\arrow{Stealth[scale=0.7]}}
}] node {} (1);
\path (6) edge [bend left=20, thick, postaction={decorate},
decoration={
markings,
mark=at position 0.52 with {\arrow{Stealth[scale=0.7]}}
}] node {} (2);
\path (6) edge [bend left=20, thick, postaction={decorate},
decoration={
markings,
mark=at position 0.5 with {\arrow{Stealth[scale=0.7]}}
}] node {} (3);
\path (6) edge [bend left=20, thick, postaction={decorate},
decoration={
markings,
mark=at position 0.6 with {\arrow{Stealth[scale=0.7]}}
}] node {} (4);
\path (6) edge [bend left=20, dotted, thick, postaction={decorate},
decoration={
markings,
mark=at position 0.55 with {\arrow{Stealth[scale=0.7]}}
}] node {} (5);
\path (3,0.5) 
node[matrix,anchor=north east,draw,nodes={anchor=center},inner sep=2pt, thick]  {
\pic{sample=solid}; & \node{$0_\D$}; \\
\pic{sample=dashed}; & \node{$1_\C$}; \\
\pic{sample=dotted}; & \node{$\top_\D$}; \\
\pic{sample=dashdotdotted}; & \node{$a$}; \\
};
\normalsize
\end{tikzpicture}
\caption{A consistent saturated network for the wreath product $\C_+[\D_+]$ in Figure~\ref{fig:wreath_product_example_2}.}\label{fig:network_for_wreath_product_ex2}
\end{figure}
\end{example}


We now state and prove the main result of this paper.  

\begin{theorem}\label{thm:ComplexOfWreathProductRepresentable}
Let $\A$ and $\B$ be finite normally representable DqRAs such that $A\cap B=\varnothing$, $1_\A\in A_+$, $1_\B\in B_+$ and $1_{\B}^{-_{\B_+}}=1_{\B}^{\sim_{\B_+}}=1_\B^{\neg_{\B_+}} = 1_{\B}$. Suppose that either $1_\A$ is the top element of $(A_+, \preccurlyeq_{\A_+})$ or every element of $B_+{\setminus}\{1_\B\}$ is doubly deficient. Then
the complex algebra $\left(\A_+[\B_+]\right)^+$ of the wreath product of the frames $\A_+$ and $\B_+$ is representable. 
\end{theorem}
\begin{proof}
Let $\A$ and $\B$ be finite \emph{normally} representable DqRAs such that $A\cap B=\varnothing$, $1_\A\in A_+$, $1_\B\in B_+$ and $1_{\B}^{-_{\B_+}}=1_{\B}^{\sim_{\B_+}}=1^{\neg_{\B_+}}_\B = 1_{\B}$. 
Then, by Proposition~\ref{prop:SatConNetworkExists}, $\A_+$ and $\B_+$ have the network realisation property. Since $1_\A$ is the top element of $(A_+, \preccurlyeq_{\A_+})$ or every element of $B_+{\setminus}\{1_\B\}$ is doubly deficient, it follows from Proposition~\ref{prop:existence_network_wreath_product} that the wreath product $\A_+[\B_+]$ of $\A_+$ and $\B_+$ also has the network realisation property. Consequently, by Proposition~\ref{prop:consistent_sat_network_imply_representable}, $(\A_+[\B_+])^+$ is representable. 
\end{proof}

Finally, we use the above results and examples to show that the DqRAs with identifiers $D^6_{4, 6}$ and $D^6_{3,3}$ in Figure~\ref{fig:small_examples_DqRAs} are (finitely) representable.  

\begin{example}\label{ex:rep_complex_algebra_A+[B+]}
Since the wreath product $\A_+[\B_+]$ of the dual frames of $\A$ and $\B$ in Figure~\ref{fig:small_examples_DqRAs} is a frame (see Example~\ref{ex:wreath_product_example_1}), its complex algebra $$(\A_+[\B_+])^+ = \langle\{\varnothing, \{0_\B\}, \{1_\A\}, \{1_\A, 0_\B\}, \{\top_\A, 1_\A\}, \{\top_\A, 1_\A, 0_\B\}\}, \cap, \cup, \circ, {\sim}, -, \neg, \{1_\A\}\rangle$$ is a DqRA by Proposition~\ref{prop:complex_algebra_frame}. The complex algebra is depicted in Figure~\ref{fig:complex_algebra_and rep_example_1}, and is isomorphic to the DqRA with identifier $D^6_{4,6}$ in  Figure~\ref{fig:small_examples_DqRAs}. Since $\A$ and $\B$ are normally representable (see Examples~\ref{ex:AandCRepresentations} and~\ref{ex:BRepresentation}) and $1_\A$ is the top element of $(A_+, \preccurlyeq_{\A_+})$, the complex algebra of $\A_+[\B_+]$ (and hence $D^6_{4, 6}$) is representable by Theorem~\ref{thm:ComplexOfWreathProductRepresentable}.
Specifically, it is representable over the $4$-element anti-chain $(N, \textnormal{id}_N)$ with $N = \{[x^u], [y^u], [x^v], [y^v]\}$, $E_N = N^2$, $\alpha = \{([x^u], [y^u]), ([y^u], [x^u]), ([x^v], [y^v]), ([y^v], [x^v])\}$ and $\beta = \textnormal{id}_N$. The algebra $(\A_+[\B_+])^+$ is isomorphic to the subalgebra
$$\langle \{\varnothing, \textnormal{id}_{N},N^2{\setminus}(\textnormal{id}_{N} \cup R), N^2{\setminus}R, \textnormal{id}_N\cup R, N^2\} ,\cap, \cup, \mathbin{;},  {\sim},{-}, \neg, \textnormal{id}_N\rangle,$$ where $R = \{([x^u], [y^u]), ([y^u], [x^u]), ([x^v], [y^v]), ([y^v], [x^v])\}$, 
of $\mathbf{D}(N^2, \preccurlyeq_{N^2})$.
See Figure~\ref{fig:complex_algebra_and rep_example_1} for a depiction of this  representation of $(\A_+[\B_+])^+$.
\end{example}

\begin{example}\label{ex:rep_complex_algebra_C+[D+]}  
Recall from Example~\ref{ex:wreath_product_example_2} that the wreath product $\C_+[\D_+]$ of the dual frames of $\C$ and $\D$ in Figure~\ref{fig:small_examples_DqRAs} is a frame, so its complex algebra $(\C_+[\D_+])^+$
is a DqRA by Proposition~\ref{prop:complex_algebra_frame}. This algebra is depicted in Figure~\ref{fig:complex_algebra_and rep_example_2}, and is isomorphic to the DqRA with identifier $D^{6}_{3,3}$ in Figure~\ref{fig:small_examples_DqRAs}. Since $\C$ and $\D$ are both normally representable (see Examples~\ref{ex:AandCRepresentations} and~\ref{ex:DRepresentation}) and every element of $D_+{\setminus}\{1_\D\}$ is doubly deficient, the complex algebra of $\C_+[\D_+]$ (and hence $D^6_{3,3}$) is representable by Theorem~\ref{thm:ComplexOfWreathProductRepresentable}. Specifically, it is representable over the poset $(N, \leqslant_N)$ with $N = \{[x^s], [x^t], [x^u], [y^u], [x^v], [x^w] \}$, ${\leqslant_N} = \textnormal{id}_N\cup (\{[x^s],[x^t]\} \times \{[x^u],[x^v], [y^u],[x^w]\}) \cup (\{[x^u], [y^u]\} \times \{[x^v], [x^w]\})$, $E_N = N^2$, $\alpha= \{([x^s], [x^t]), ([x^t], [x^s]), ([x^u], [y^u]), ([y^u],[x^u]), ([x^v],[x^w]), ([x^w], [x^v])\}$ and $$\beta = \{([x^s], [x^v]), ([x^v],[x^s]), ([x^t], [x^w]), ([x^w], [x^t]), ([x^u], [x^u]), ([y^u], [y^u])\}.$$ The algebra $(\C_+[\D_+])^+$ is isomorphic to the subalgebra $$\langle \{\varnothing, {\leqslant}_{N}, {\leqslant}_{N}{\setminus}S, ({\leqslant}_{N}{\setminus}S)\cup T, {\leqslant_N}\cup T, N^2\} ,\cap, \cup, \mathbin{;}, {\sim},{-}, \neg, {\leqslant_N}\rangle,$$ where $S = \{([x^u], [x^u]), ([y^u], [y^u])\}$ and $T = \{([x^u], [y^u]), ([y^u], [x^u])\}$, 
of $\mathbf{D}(N^2, \preccurlyeq_{N^2})$. The representation is also depicted in Figure~\ref{fig:complex_algebra_and rep_example_2}.
\end{example}

\begin{figure}
\begin{tikzpicture}[scale=0.7]
\footnotesize
\begin{scope}
\node[draw,circle,inner sep=1.75pt,fill] (bot) at (-1,0) {};
\node[draw,circle,inner sep=1.75pt] (a) at (-2,1) {};
\node[draw,circle,inner sep=1.75pt,fill] (1) at (0,1) {};
\node[draw,circle,inner sep=1.75pt] (0) at (-1,2) {};
\node[draw,circle,inner sep=1.75pt,fill] (b) at (1,2) {};
\node[draw,circle,inner sep=1.75pt,fill] (top) at (0,3) {};
\path [-] (bot) edge node {} (a);
\path [-] (bot) edge node {} (1);
\path [-] (a) edge node {} (0);
\path [-] (1) edge node {} (0);
\path [-] (1) edge node {} (b);
\path [-] (0) edge node {} (top);
\path [-] (b) edge node {} (top);   
\node[label,anchor=south,xshift=1pt,yshift=-1pt] at (top) {$\{\top_\A, 1_\A, 0_\B\} = \{0_\B, 1_\A\}\circ \{0_\B\}$};
\node[label,anchor=east,xshift=1pt] at (0) {$\{1_\A, 0_\B\}$};
\node[label,anchor=west,xshift=-1pt] at (b) {$\{\top_\A, 1_\A\} = \{0_\B\}^2$};
\node[label,anchor=east,xshift=1pt] at (a) {$\{0_\B\}= \{\top_\A, 1_\A\}\circ \{0_\B\}$};
\node[label,anchor=west,xshift=1pt] at (1) {$\{1_\A\}$};
\node[label,anchor=east,xshift=1pt] at (bot) {$\varnothing$};
\node[yshift=5pt,xshift=-20pt] at (0,-1) {$(\A_+[\B_+])^+\cong D^6_{4,6}$};
\end{scope} 

\begin{scope}[xshift=10cm]
\node[draw,circle,inner sep=1.75pt,fill] (bot) at (-1,0) {};
\node[draw,circle,inner sep=1.75pt] (a) at (-2,1) {};
\node[draw,circle,inner sep=1.75pt,fill] (1) at (0,1) {};
\node[draw,circle,inner sep=1.75pt] (0) at (-1,2) {};
\node[draw,circle,inner sep=1.75pt,fill] (b) at (1,2) {};
\node[draw,circle,inner sep=1.75pt,fill] (top) at (0,3) {};
\path [-] (bot) edge node {} (a);
\path [-] (bot) edge node {} (1);
\path [-] (a) edge node {} (0);
\path [-] (1) edge node {} (0);
\path [-] (1) edge node {} (b);
\path [-] (0) edge node {} (top);
\path [-] (b) edge node {} (top);   
\node[label,anchor=south,xshift=1pt,yshift=-1pt] at (top) {$N^2$};
\node[label,anchor=east,xshift=1pt] at (0) {$N^2{\setminus} R$};
\node[label,anchor=west,xshift=-1pt] at (b) {$\textnormal{id}_N\cup R$};
\node[label,anchor=east,xshift=1pt] at (a) {$N^2{\setminus}(\textnormal{id}_N\cup R)$};
\node[label,anchor=west,xshift=1pt] at (1) {$\textnormal{id}_N$};
\node[label,anchor=east,xshift=1pt] at (bot) {$\varnothing$};
\node[xshift=-20pt,yshift=5pt] at (0,-1) {$(\A_+[\B_+])^+\hookrightarrow \mathbf{D}(N^2, \preccurlyeq_{N_2})$};
\end{scope} 
\normalsize
\end{tikzpicture}
\caption{The complex algebra of the wreath product $\A_+[\B_+]$ in Figure~\ref{fig:wreath_product_example_1} and its representation. }\label{fig:complex_algebra_and rep_example_1}
\end{figure}
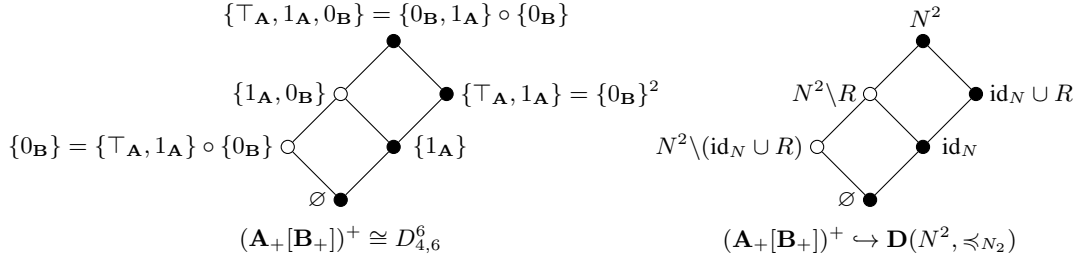

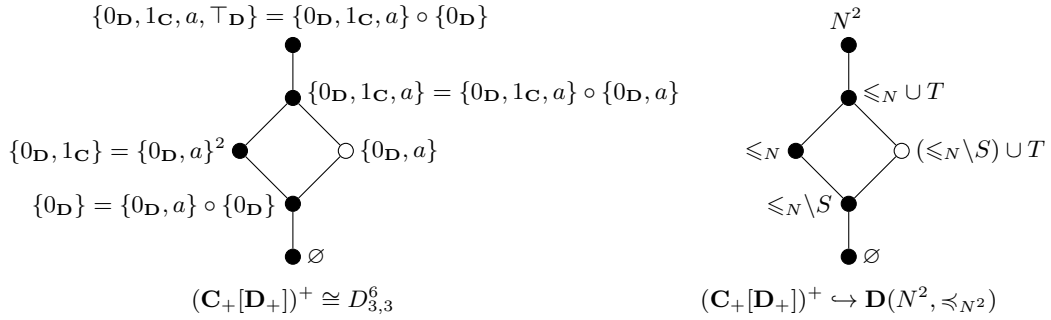
\begin{figure}
\begin{tikzpicture}[scale=0.7]
\footnotesize
\begin{scope}
\node[draw,circle,inner sep=2pt,fill] (a1) at (1,0) {};
\node[draw,circle,inner sep=2pt,fill] (a2) at (1,4) {};
\node[draw,circle,inner sep=2pt,fill] (b0) at (1,1) {};
\node[draw,circle,inner sep=2pt,fill] (b1) at (1,3) {};
\node[draw,circle,inner sep=2pt,fill] (b2) at (0,2) {};
\node[draw,circle,inner sep=2pt] (b3) at (2,2) {};
\draw[order] (b0)--(b2);
\draw[order] (b0)--(b3);
\draw[order] (b2)--(b1);
\draw[order] (b3)--(b1);
\draw[order] (a1)--(b0);
\draw[order] (b1)--(a2);
\node[label,anchor=west,xshift=-1pt,yshift=2pt] at (b1) {$\{0_\D, 1_\C, a\}= \{0_\D, 1_\C, a\}\circ \{0_\D, a\}$};
\node[label,anchor=east,yshift=-1pt] at (b0) {$\{0_\D\}=\{0_\D, a\}\circ\{0_\D\}$};
\node[label,anchor=east,xshift=1pt] at (b2) {$\{0_\D, 1_\C\}=\{0_\D, a\}^2$};
\node[label,anchor=west,xshift=-1pt] at (b3) {$\{0_\D, a\}$};
\node[label,anchor=west,xshift=-1pt] at (a1) {$\varnothing$};
\node[label,anchor=south,xshift=-1pt] at (a2) {$\{0_\D, 1_\C, a, \top_\D\}=\{0_\D, 1_\C, a\}\circ \{0_\D\}$};
\node[label,anchor=north,yshift=-5pt] at (a1) {$(\C_+[\D_+])^+ \cong D^6_{3,3}$};
\end{scope}
\begin{scope}[xshift=10.5cm]
\node[draw,circle,inner sep=2pt,fill] (a1) at (1,0) {};
\node[draw,circle,inner sep=2pt,fill] (a2) at (1,4) {};
\node[draw,circle,inner sep=2pt,fill] (b0) at (1,1) {};
\node[draw,circle,inner sep=2pt,fill] (b1) at (1,3) {};
\node[draw,circle,inner sep=2pt,fill] (b2) at (0,2) {};
\node[draw,circle,inner sep=2pt] (b3) at (2,2) {};
\draw[order] (b0)--(b2);
\draw[order] (b0)--(b3);
\draw[order] (b2)--(b1);
\draw[order] (b3)--(b1);
\draw[order] (a1)--(b0);
\draw[order] (b1)--(a2);
\node[label,anchor=west,xshift=-1pt,yshift=2pt] at (b1) {${\leqslant_N} \cup T$};
\node[label,anchor=east,yshift=-1pt] at (b0) {${\leqslant_N}{\setminus} S$};
\node[label,anchor=east,xshift=1pt] at (b2) {$\leqslant_N$};
\node[label,anchor=west,xshift=-1pt] at (b3) {$({\leqslant_N}{\setminus}S) \cup T$};
\node[label,anchor=west,xshift=-1pt] at (a1) {$\varnothing$};
\node[label,anchor=south,xshift=-1pt] at (a2) {$N^2$};
\node[label,anchor=north,yshift=-5pt] at (a1) {$(\C_+[\D_+])^+ \hookrightarrow \mathbf D(N^2, \preccurlyeq_{N^2})$};
\end{scope}
\normalsize
\end{tikzpicture}
\caption{The complex algebra of the wreath product $\C_+[\D_+]$ in Figure~\ref{fig:wreath_product_example_2} and its representation. }\label{fig:complex_algebra_and rep_example_2}
\end{figure}


    



    

{\bf Acknowledgements:} The authors are grateful to Peter Jipsen for drawing our attention to Comer's construction, to Andrew Craig for his contribution to the framing of the research question and initial collaboration on this project, and to Ja\v{s} \v{S}emrl for his input in the definition of the network of a DqRA frame.
  

\end{document}